\documentclass[runningheads]{llncs}
\usepackage{fullpage}
\usepackage{graphicx}
\usepackage{amsfonts, amsmath, amssymb}
\usepackage{comment}
\usepackage{xcolor}
\usepackage{breakcites}
\usepackage{xspace}
\usepackage{algorithm, algorithmic}
\usepackage{url}
\usepackage{comment}
\usepackage[normalem]{ulem}
\usepackage{hyperref}
\usepackage{booktabs}
\providecommand{\noopsort}[1]{}
\usepackage{array,multirow}
\usepackage{ifthen}
\usepackage{MnSymbol}
\usepackage{fancyhdr}
\mathchardef\mhyphen="2D
\def\bsk{\textsf{\textbf{sk}}}

\def\bct{{\textsf{\textbf{ct}}}}
\def\ba{{\mathbf{a}}}
\def\bb{{\mathbf{b}}}

\def\be{{\mathbf{e}}}

\def\ModPPMM{{\mathrm{Mod\mhyphen PP\mhyphen MM}}}
\def\ModPPMv{{\mathrm{Mod\mhyphen PP\mhyphen Mv}}}
\def\fpPPMM{{\mathrm{fp\mhyphen PP\mhyphen MM}}}

\def\Relin{{\mathsf{Relin}}}

\newcommand{\isst}[1]{#1^*}
\newcommand{\trace}[1]{\texttt{Tr}(#1)}
\newcommand{\enc}[1]{\textsf{Enc}\left(#1\right)}
\newcommand{\ct}{\textsf{ct}}
\newcommand{\ring}{\mathcal{R}}
\newcommand{\gal}[1]{\text{Gal}(#1)}
\newcommand{\sk}{\textsf{sk}}

\newcommand{\ctvec}{\mathsf{\mathbf{ct}}}
\newcommand{\auxvec}{\mathbf{aux}}

\newcommand{\cmt}{\textrm{C-MT}\xspace}
\newcommand{\tweak}{\textsf{Tweak}\xspace}
\newcommand{\transpose}{\textsf{Transpose}\xspace}

\newcommand{\rescale}{\textsf{Rescale}}
\newcommand{\aut}{\textsf{Aut}}

\newcommand{\modswitch}{\textsf{ModSwitch}}

\newboolean{publishedversion}
\setboolean{publishedversion}{false}

\newcommand{\rotmtx}[1]{\mathtt{Toep}(#1)}
\newcommand{\toep}[1]{\texttt{Toep}(#1)}
\newcommand{\VVec}[1]{\mathtt{Vec}(#1)}
\newcommand{\Pol}[1]{\mathtt{Pol}(#1)}

\newcommand{\pcmm}{\mathrm{PC\mhyphen MM}\xspace}
\newcommand{\ppmm}{\mathrm{PP\mhyphen MM}\xspace}
\newcommand{\ppmv}{\mathrm{PP\mhyphen Mv}\xspace}
\newcommand{\pcmv}{\mathrm{PC\mhyphen Mv}\xspace}
\newcommand{\ccmm}{\mathrm{CC\mhyphen MM}\xspace}
\newcommand{\cpmm}{\mathrm{CP\mhyphen MM}\xspace}
\newcommand{\cpmv}{\mathrm{CP\mhyphen Mv}\xspace}
\newcommand{\CPMM}{\mathrm{CP\mhyphen MM}}
\newcommand{\ccmv}{\mathrm{CC\mhyphen Mv}\xspace}

\def\tx{\widetilde{x}}

\newcommand{\style}[1]{\ensuremath{\mathsf{#1}}}

\newcommand{\swk}{\style{swk}}
\newcommand{\fmtswk}{\style{fmt}\mbox{-}\style{swk}}

\newcommand{\relin}{\style{Relin}}

\def\ZZ{{\mathbb{Z}}}
\def\RR{{\mathbb{R}}}

\def\R{{\mathcal{R}}}

\def\Ecd{\textsf{Ecd}}
\def\Dcd{\textsf{Dcd}}
\def\Enc{\textsf{Enc}}
\def\Dec{\textsf{Dec}}
\def\Add{\textsf{Add}}
\def\PCMult{\textsf{PCMult}}

\def\KeyGen{\textsf{KeyGen}}
\def\KeySwitch{\textsf{KeySwitch}}

\def\SWKGen{\textsf{SWKGen}}

\def\Rescale{\textsf{Rescale}}

\def\PCMult{\textsf{PCMult}}

\def\ModDecomp{\textsf{ModDecomp}}

\def\pk{\mathsf{pk}}
\def\sk{\mathsf{sk}}
\def\rk{\mathsf{rk}}

\def\ct{\mathsf{ct}}
\def\coeff{\textrm{coeff}}

\def\Aut{{\textsf{Aut}}}

\def\bx{{\mathbf{x}}}
\def\by{{\mathbf{y}}}
\def\LWE{{\textsf{LWE}}}
\def\RLWE{{\textsf{RLWE}}}
\def\RGSW{{\textsf{RGSW}}}
\def\MLWE{{\textsf{MLWE}}}
\def\GGSW{{\filledstar\textsf{GSW}}}
\def\GLWE{{\filledstar\textsf{LWE}}}
\def\shs{\textrm{sh-$s$}}
\def\sha{\textrm{sh-$a$}}
\def\ModPack{{\textsf{ModPack}}}

\newcommand{\FFLASBibkeyhack}[7]{}
\newcommand{\HEaaNBibkeyhack}[5]{}

\NewDocumentCommand{\opnorm}{sO{}m}{
  \IfBooleanTF{#1}{
    $\left|\opnormkern\left|\opnormkern\left|
    #3
    \right|\opnormkern\right|\opnormkern\right|
  $}{
    \mathopen{#2|\opnormkern #2|\opnormkern #2|}
    #3
    \mathclose{#2|\opnormkern #2|\opnormkern #2|}
  }
}
\newcommand{\opnormkern}{\mkern-1.5mu\relax}

\def\BA{{\bf{A}}}
\def\BD{{\bf{D}}}
\def\wBA{{\widetilde{\BA}}}
\def\BB{{\bf B}}
\def\BE{{\bf{E}}}
\def\wBB{{\widetilde{\BB}}}
\def\BM{{\bf M}}
\def\BS{{\bf S}}
\def\BU{{\bf U}}

\def\bb{{\bf b}}
\def\ba{{\bf a}}

\def\bv{{\bf v}}

\def\bk{{\bf k}}

\def\bm{{\bf m}}
\def\BSv{{\bf S}^{\bf v}}
\def\BSU{{\bf S}^{\bf U}}

\def\BSU{\BS^\BU}
\def\BSM{\BS^\BM}

\def\bv{\mathbf{v}}

\begin{document}
\title{
Fast Homomorphic Linear Algebra with BLAS
}
\titlerunning{ }

\author{Youngjin Bae\inst{1}
\orcidID{0000-0001-6870-4504} 
\and Jung Hee Cheon\inst{1,2}
\orcidID{0000-0002-7085-2220} 
\and \\
Guillaume Hanrot\inst{3} 
\orcidID{0000-0001-9319-0365}
\and Jai~Hyun Park\inst{3}
\orcidID{0000-0002-5401-8949} 
\and \\  
Damien Stehlé\inst{3}\orcidID{0000-0003-3435-2453}
}

\authorrunning{Y. Bae, J. H. Cheon, G. Hanrot, J. H. Park, and D. Stehlé}
\institute{}
\institute{CryptoLab Inc., Seoul, Republic of Korea \and Seoul National University, Seoul, Republic of Korea \and CryptoLab Inc., Lyon, France}

\maketitle

\begin{abstract}
Homomorphic encryption is a cryptographic paradigm allowing to compute on encrypted data, opening a wide range of applications in privacy-preserving data manipulation, notably in AI. Many of those applications require significant linear algebra computations (matrix-vector products, and matrix-matrix products). 

This central role of linear algebra computations goes far beyond homomorphic algebra and applies to most areas of scientific computing. This high versatility led, over time, to the development of a set of highly optimized routines, specified in 1979 under the name BLAS (basic linear algebra subroutines). 

Motivated both by the applicative importance of homomorphic linear algebra and the access to highly efficient implementations of cleartext linear algebra able to draw the most out of available hardware, we explore the connections between CKKS-based homomorphic linear algebra and floating-point plaintext linear algebra. The CKKS homomorphic encryption system is the most natural choice in this setting, as it natively handles real numbers and offers a large SIMD parallelism.

We provide reductions for matrix-vector products, vector-vector products for moderate-sized to large matrices to their plaintext equivalents. Combined with BLAS, we demonstrate that the efficiency loss between CKKS-based encrypted square matrix multiplication and double-precision floating-point square matrix multiplication is a mere 4-12 factor, depending on the precise situation.
\end{abstract}
\textbf{Corresponding author}: Jai Hyun Park, \texttt{jaihyunp@cryptolab.co.kr}\\
\smallskip

\noindent \textbf{Keywords:} Encrypted linear algebra, matrix-matrix product, matrix-vector product, homomorphic encryption, CKKS. 
\section{Introduction}
\begingroup
\renewcommand{\thefootnote}{}
\footnotetext{
\copyright~IACR 2026. This article is the final version submitted by the author(s) to the IACR and to Springer-Verlag on March 19, 2026. The version published by Springer-Verlag is available at \url{https://doi.org/10.1007/s00145-026-09580-x}.\vspace{2pt}\\
This is an extended version of articles published by the same authors in the proceedings of Crypto 2024~\cite{BCHPS24} and by the fourth author in the proceedings of Eurocrypt 2025~\cite{Park25}.
}
\endgroup

This work focuses on homomorphic linear algebra, which we define in its most general form as the following problem: given two matrices~$\BU \in \RR^{d_1\times d_2}$ and $\BM \in \RR^{d_2\times d_3}$ that can be either encrypted or plaintext, return ciphertexts decrypting to the product $\BU \cdot \BM \in \RR^{d_1 \times d_3}$. For $d_3 = 1$, we obtain, as a particular case of this problem, the homomorphic matrix-vector problem, which will be denoted by Mv. The general matrix-matrix problem will itself be denoted by MM. 
We aim at homomorphic linear algebra algorithms that are compatible with fully homomorphic encryption (FHE), notably with the CKKS scheme~\cite{CKKS17}.\footnote{We primarily focus on CKKS for the sake of simplicity, even though most of our techniques can be  adapted to BGV and BFV.} 
This constraint enables our algorithms to be used as parts of larger homomorphic computations.
In more details, we shall consider the following problems.
\begin{itemize}
    \item[$\bullet$] \textbf{Homomorphic matrix-vector product}. Homomorphic matrix-vector product takes as inputs a matrix $\BM = (m_{i,j})_{0\le i < d_1, 0\le j < d_2}$ and a vector $\bv = (v_i)_{0\le i < d_2}$, either of which can be ciphertext or plaintext, giving the three flavours $\cpmv$, $\pcmv$ and 
    $\ccmv$. 
  Matrix-vector products occur in a large variety of contexts where privacy concerns exist. It is highly common in private inference for neural networks (fully connected layers, convolutional layers), in private information retrieval (one of the ways to implement PIR is through representation of a database by a matrix, in which case the query is encoded into 
    a one-hot vector), in approximate vector search (consisting in comparing a template vector to a database; see, e.g., \cite{HDCZ23}). In most of these situations, the matrix and vector dimension are typically large, and the matrix is known beforehand and can be precomputed upon. 
    
    \item[$\bullet$] \textbf{Plaintext-ciphertext matrix multiplication\label{not:cpmm} ($\pcmm$ and $\cpmm$)}\footnote{We distinguish the $\pcmm$ and the $\cpmm$ cases. Indeed, ciphertexts encrypting a matrix contain the coefficients packed either in a row-wise or column-wise manner. Once a packing is fixed, $\pcmm$ and $\cpmm$ become two different problems.} $\pcmm$ and $\cpmm$ take as  inputs a matrix ${\bf U}=(u_{i,j})_{0\le i< d_1, 0 \le j<d_2}$ (in clear) and ciphertext(s) corresponding to a  matrix ${\bf M}=(m_{i,j})_{0\le i<d_2,0\le j<d_3}$ and return ciphertext(s) corresponding to the matrix~$\BM \cdot \BU$ (for CP-MM) or~${\bf U} \cdot {\bf M}$ (for $\pcmm$). $\cpmm$ arises in the context of privacy-preserving machine learning, notably in the inference phase of transformers for large language models (LLM), such as GPT~\cite{GPT}, BERT~\cite{bert} and LLaMA~\cite{llama}.  $\cpmm$ has been applied for the private evaluation of such large language models along with secure multiparty computation~\cite{iron,PZM+24,east}
    or fully homomorphic encryption~\cite{nexus}.  In~\cite{FCES+23}, $\cpmm$ is used for federated principal component analysis between different data providers, and, in~\cite{LZ22}, applications to smart contracts, health and finance are described.
    
      \item[$\bullet$] \textbf{Ciphertext-ciphertext matrix multiplication ($\ccmm$)}. $\ccmm$ takes as inputs two sets of ciphertext(s) encrypting two input matrices and outputs ciphertext(s) encrypting the product matrix. $\ccmm$ plays a central role in privacy-preserving machine learning (PPML) when a server trains or performs inference on machine learning models using encrypted data from the client, for example, during privacy-preserving training and inference of large language models (LLM) in~\cite{PZM+24,iron,nexus}. 
\end{itemize}

Many of these applications require the ability to handle large dimensions. We illustrate this fact in the case of LLM evaluation. Overall, transformers perform four different types of matrix multiplications, parameterized by $\mathsf{dim}$, $\mathsf{ntol}$, $\mathsf{nheads}$ and~$\mathsf{dff}$: 
\begin{itemize}
\item[$\bullet$] in the so-called attention phase, a weight matrix  ${\bf U} \in \RR^{\mathsf{dim}\times\mathsf{dim}}$ is multiplied by several user data matrices ${\bf M}\in \RR^{\mathsf{dim}\times \mathsf{ntok}}$;
\item[$\bullet$] right after the attention phase, the resulting matrices are split as matrices $\BM_i \in \RR^{\mathsf{dim}/\mathsf{nheads} \times \mathsf{ntok}}$, and two such encrypted matrices are multiplied as ${\BM'_i}^t \cdot \BM_i \in \RR^{\mathsf{ntok} \times \mathsf{ntok}}$;
\item[$\bullet$] in the feed-forward phase, larger weight  matrices  ${\bf U} 
 \in \RR^{\mathsf{dff}\times \mathsf{dim}}$ are multiplied by processed user data  ${\bf M} \in \RR^{\mathsf{dim}\times \mathsf{ntok}}$ before the activation step; 
 \item[$\bullet$] later in the feed-forward phase, the processed data is reduced back to dimension $\mathsf{dim}$, by multiplying  another large weight matrix in~$\RR^{\mathsf{dim}\times \mathsf{dff}}$ with the current user data matrix in~$\RR^{\mathsf{dff}\times \mathsf{ntok}}$.   
\end{itemize}

For privacy-preserving LLM inference, the weight matrices above are cleartext, whereas the user data matrices are encrypted. In mainstream LLMs, the $\mathsf{ntok}$ dimension ranges from~$128$ to $16\ 384$ (in GPT-3.5) or more, and the dimension $\mathsf{dim}$ ranges from 4 096 (LLaMA-7B) to 18\ 432 (PaLM~540B). 
This illustrates the need for efficient algorithms for homomorphic linear algebra that scale well, in order to handle large dimensions. 

In contrast, most of the existing algorithms have been developed with smaller dimensions in mind and scale poorly to such large dimensions. State-of-the-art $\cpmm$ and $\ccmm$ algorithms rely on fully homomorphic encryption schemes (FHE) based on the ring learning with errors problem~\cite{SSTX09,LPR10} (RLWE), such as BGV~\cite{BGV14}, BFV~\cite{Brakerski12,FV12} and CKKS~\cite{CKKS17}.  Most have asymptotic run-times that are linear in the product of the dimensions (or less for large dimensions, using recursive blocking techniques such as Strassen's algorithm). However, their 
practical performance is often much worse than the corresponding plaintext computation. For instance, for $\cpmm$, an adaptation of~\cite{JKLS18} is used in~\cite{FCES+23}, which reports a multiplication of a $256 \times 256$ plaintext matrix with a $256 \times 8$ ciphertext matrix  in 3.8s on an Intel Xeon 2.5GHz with 24 threads on 12 cores; as a comparison, a similar double-precision floating-point computation using OpenBLAS~\cite{openblas} takes of the order of~$20\mu s$ using single-threaded computation on a similar CPU, a difference of more than 5 orders of magnitude. 

The core difficulty of encrypted linear algebra stems from the  \textbf{}structure of the ciphertexts. RLWE-based FHE schemes encrypt a vector in a ciphertext as a whole, requiring ``key-switching'' operations to rearrange data in the vector. The existing approaches for homomorphic linear algebra rely on the basic operations on ciphertexts provided by the homomorphic encryption system: SIMD additions, SIMD multiplications, and key-switchings.  These approaches suffer from two drawbacks. 
Firstly, even though most of those approaches achieve the correct arithmetic complexity (the same order of base-ring add/mult operations as used in clear), the number of key-switching operations is often so large that these dominate the execution time in practice. 
Secondly, the involved mixture of operations performed  by these algorithms, in particular the frequent use of key-switchings, makes the computational and memory access pattern complex. As a result, these algorithms and their implementations are difficult to optimize, leading to frequent bottlenecks in memory access, with numerous steps becoming memory-bound due in particular to the large size of the switching keys.

\subsection{Reductions to cleartext linear algebra}

We also work with \RLWE-format ciphertexts. However, our approach differs from most previous works in that we obtain \emph{reductions to plaintext linear algebra} computations (i.e., $\ppmm$ and $\ppmv$), following a path initiated by~\cite{LZ22}. 

From a theoretical point of view, this goes against the broad belief that homomorphic computation is doomed to be several orders of magnitude slower than cleartext computations. The existence of such a reduction proves that reality is more complex and that certain tasks can be performed homomorphically at a modest cost overhead. 
From a practical point of view, this allows us to rely on the long work of the high-performance linear algebra community, which has led to the development of highly optimized libraries for the basic linear algebra subroutines (BLAS)~\cite{openblas,flint}.  

Our  reductions approach has two main limitations. First, the dimension of the operations in the reduction may not be the exact same as the dimension of the original operation. In particular, they may increase the arithmetic complexity in the case of matrices with dimensions smaller than the RLWE ring-degree.  Second, most of our reductions are not simple calls to cleartext linear algebra operations, but require some tasks before or after those calls. We distinguish two different types of such tasks. We call \emph{processing} (pre- and post-) tasks that have a negligible costs  compared to that of the cleartext linear algebra component. In the case where a task is more costly, it can be performed at the outset and we call it \emph{pre-computation}. Typically, such a task can be amortized in the context of numerous calls to the main routine.  

Despite these limitations, the reduction approach still makes the computational patterns much more regular and allows us to benefit, through the use of BLAS implementations,  from the clear understanding of this computational and memory pattern provided by years of work of the linear algebra community. It also yields algorithms that are very simple to implement, and get the most out of any piece of hardware for which a BLAS implementation is available (including in particular GPU architectures). 

This approach was initiated by~\cite{LZ22}, which reduces one $\cpmm$ to two $\ppmm$'s modulo an integer. Even though no server-side experiment is reported in~\cite{LZ22}, this hinted to the fact that~$\cpmm$ implementations can be based on fast linear algebra software and hence be very efficient. The algorithm from~\cite{LZ22} does not move data within \RLWE\ ciphertexts, but still requires that the matrix dimensions should be above or equal to the RLWE ring-degree. In practice, to enable computations with reasonable precision and 128-bit security, the degree should be at least~$2^{12}$ for a stand-alone $\cpmm$, and~$2^{13}$ if we want to enable key-switching for compatibility with FHE.

\subsection{Contributions}
We introduce reductions from the encrypted linear algebra tasks with various dimensions to cleartext linear algebra computations with similar dimensions. The ciphertexts we consider are defined as polynomials modulo some integer~$q$, and even though the cleartext matrix multiplications are approximations to reals with some precision, the linear algebra computations we perform on cleartexts are modulo integers~$q'$ of bit-size similar to that of~$q$. The bit-size of~$q'$ is within a small additive constant from that of~$q$.

We describe several $\cpmm$ (resp. $\cpmv$) algorithms that reduce to one or two $\ppmm$'s (resp.\ $\ppmv$'s), for a row dimension on the ciphertext matrix below and above the RLWE ring-degree, in a black-box manner and with limited overhead.
We also introduce a $\ccmm$ algorithm that reduces $\ccmm$ to four modular $\ppmm$'s for input dimensions above the RLWE ring-degree, and a general reduction algorithm that can handle any matrix-vector or matrix-matrix product, reducing it to eight plaintext operations of the same kind.

\paragraph{The algorithms.}
Our $\cpmm$ algorithms extend the one from~\cite{LZ22}, providing greater dimension flexibility and greater efficiency. We give two algorithms, for ciphertext matrices of dimensions~$d_1 \times d_2$, depending on whether~$d_1$ is smaller 
 or larger than the RLWE ring-degree~$N$ (note that~\cite{LZ22} works only for~$d_1 \geq N$). In both cases, our algorithms consist in reducing $\cpmm$ to two modular $\ppmm$'s, where the $\ppmm$ dimensions are~$N \times d_2 \times d_3$ and~$d_1 \times d_2 \times d_3$, if the plaintext matrix dimensions are $d_2 \times d_3$.
 
We also describe two algorithms (for small and large~$d_1$) which only use a single $d_1 \times d_2 \times d_3$ modular $\ppmm$, if one allows for precomputations on the cleartext matrix. Finally, we show that the latter modular $\ppmm$ can be replaced by a floating-point $\ppmm$,
hence obtaining reductions from $\cpmm$ to a single floating-point $\ppmm$ with the same dimension, with precomputation. 
In the case $d_3 = 1$, our algorithms yield $\cpmv$ algorithms.
\smallskip

For matrices of dimension larger or equal to the RLWE degree, our first $\ccmm$ algorithm consists of reducing $\ccmm$ to four modular $\ppmm$s, where $\ppmm$'s have the same dimension as the given $\ccmm$. As this algorithm relied on a large evaluation key size, 
we provide lightweight $\ccmm$ algorithms for large matrices, which use much fewer evaluation keys and have comparable efficiency. 
\smallskip

Our $\ccmv$ algorithm is a reduction to eight modular $\ppmv$'s with a modulus of similar 
size. It handles matrix-vector products of dimension $d_1 \times d_2$, with a matrix encrypted with a ring-degree~$N$ and a vector encrypted with a ring-degree~$N'$ that is not necessarily equal to~$N$. The algorithm performs two modular PP-Mv for each one of the four pairs of dimensions of $\{d_1, N\} \times \{d_2, N'\}$. This extends to~$\ccmm$ and~$\pcmm$ in a variety of dimensions by viewing the second matrix as a batch of vectors. 

\paragraph{Technical tools.}
The main new technical tool for $\cpmm$ is the use of compact formats for multiple \RLWE\ ciphertexts. Recall that a \RLWE\ ciphertext consists of a pair~$(a,b)$ of polynomials such that~$a\cdot \sk+b \approx m$, where~$\sk$ is the secret key and~$m$ is the underlying message. 
We say that several \RLWE\ ciphertexts are in shared-$a$ format if their first components~``$a$'' are identical. This format has been used many times for \emph{security proofs} in lattice-based cryptography, in the context of the lossy-mode proof technique (see, among many others, \cite{PW08,GKPV10,BLPRS13}). We use this format for \emph{computation} purposes. We provide
conversion algorithms from shared-secret \RLWE\ ciphertexts to shared-$a$ ciphertexts, and vice versa. 
 
The main new technical tool for $\ccmm$ is a new fast ciphertext matrix transpose ($\cmt$) algorithm. The algorithm uses $\widetilde{O}(N^2)$ mod-$q$ arithmetic operations to transpose an $N\times N$ matrix encrypted in $N$ ciphertexts, where~$q$ is the ciphertext modulus. 
Although $\cmt$ is an important matrix operation, 
we emphasize that it also has applications beyond $\ccmm$.

Finally, the main technical tool for our most general reduction is a generalization of (ring) GSW encryption~\cite{GSW13} to matrices, built on top of the \RLWE-type structures used in the $\cpmm$ approach, and an extension to Mv setting of the classical external product from \RGSW\ and \RLWE\ to \RLWE\ ciphertexts~\cite{BP16, CGGI17}.
All the various techniques developed for the $\cpmm$ algorithm extend to this situation, allowing for a wide range of efficient reductions. We also describe a method, relying on the \cmt algorithm, to obtain a \RGSW\ encryption of a matrix  from a \RLWE\ encryption of that matrix, in the case of a square $N\times N$ matrix.

The algorithms and the ciphertext formats they rely on are summarized in Table~\ref{tab:list_algorithms}.
Some tasks include others as particular cases. For instance, as $\ccmm$ encompasses all other problems, our $\ccmm$ algorithms can be used for all other problems. However:
\begin{itemize}
  \item[$\bullet$] the plaintext-ciphertext algorithms are more flexible in terms of dimension that they can handle; they are  also more efficient by a constant factor; 
  \item[$\bullet$] even though the Mv algorithms are particular cases of the MM algorithms, they are analyzed independently, as certain pre- or post-processings  which are negligible in the square MM situation become significant or even dominant in the Mv case;
  \item[$\bullet$] when both an \RLWE-based algorithm  and an \RLWE-and-\RGSW-based algorithm are available (such as for $\ccmm$), the former, when applicable, is usually more efficient, and the latter is more flexible in terms of dimensions. 
\end{itemize}

\begin{table}
    \centering
    \begin{tabular}{|c|c|c|c|}
    \hline
    Task    & Ciphertext format(s)      & Reference   \\ \hline 
    $\cpmm$ & $\GLWE_{d_1}$                          & Section~\ref{sse:pcmm}~or Section~\ref{sse:precomputation}$^\dag$\\\hline
    $\pcmm$ & $\GLWE_{d_2}$   & Sections~\ref{sec:transpose}~and~\ref{sse:pcmm}, or Section~\ref{sse:gsw}$^\dag$\\\hline

    $\ccmm$ & $\GLWE_{d_1}$~and~$\GLWE_{d_2}$                        & Section~\ref{sse:ccmm}\\\hline
    $\ccmm$ & $\GGSW_{d_1,d_2}$ and $\GLWE_{d_2}$     & Section~\ref{sse:gsw}\\\hline

    \end{tabular}
    \caption{List of reductions from encrypted linear algebra to cleartext linear algebra. The input matrices are $d_1\times d_2$ and $d_2\times d_3$ matrices, and taking $d_3=1$ (resp.\ $d_1=1$) give a matrix-vector (resp.\ vector-matrix) case. 
    The notation $\GLWE_{d}$  refers to Shared-a \RLWE\ if $d>N$, $\RLWE$ if $d=N$ and $\MLWE$ if $d<N$, and the notation~$\GGSW_{d,d'}$ is defined similarly (see Section~\ref{sec:gsw_format}). The $\GGSW$-based reduction takes an input matrix in $\GGSW_{d,d'}$ format and the other one in $\GLWE_{d'}$ format, and returns a matrix in $\GLWE_d$ format.
    The references with $\dag$ correspond to reductions that require precomputation on the plaintext matrix.}
    \label{tab:list_algorithms}
\end{table}

\paragraph{Implementations.}
We implemented several of our algorithms using the HEaaN library~\cite{heaanlib}. In Section~\ref{sec:experiments}, we provide experimental data that focuses on: (1)~the cost of $\cpmm$ with RLWE-based schemes and no precomputation
(2)~the cost of $\ccmm$ for large square matrices, and (3)~the cost of $\ccmm$ and $\ccmv$ using the RGSW-based approach.

\subsection{Technical overview}\label{sec:technical_overview}
For the sake of simplicity, in this overview, we assume that matrices are square. Let~$d$ denote their dimension. 
Let~$q \geq 2$, $N$ be a power of~2, $\R_N = \mathbb{Z}[X]/(X^N+1)$  and $\R_{q,N} = \mathbb{Z}_q[X]/(X^N+1)$. When there is no ambiguity regarding the value of $N$, we shall simply write $\R$ and $\R_q$.

\medskip
\noindent
{\bf The LZ algorithm.}
Our starting point is the algorithm from~\cite{LZ22}. Assume first that~$d=N$.
The matrix~${\bf M}$ is provided by~$d$ \RLWE\ ciphertexts~$(a_i,b_i)$, that encrypt the columns. We have, over~$\R_q$: 
\[
\forall i: a_i \cdot \sk + b_i \ \approx  \ \sum_j m_{j,i}X^j  \bmod q\enspace,
\]
where~$\sk \in \R$ is the secret key. We can rewrite the above in matrix form, as follows:
\begin{equation}
\label{eq:LZinput}
\rotmtx{\sk} \cdot \BA + \BB \ \approx  \ \BM \ \bmod q \enspace,
\end{equation}
where the $i$-th column of~$\BA$ (resp.\ $\BB$) consists of the coefficients of~$a_i$ (resp.\ $b_i$), for all $0\leq i < d$. We let $\rotmtx{\sk}$\label{not:toep} denote the matrix whose $i$-th column consist of the coefficients of $X^i \cdot \sk \in \R$, for all~$0 \leq i < d$. 
Note that the matrix~$\rotmtx{\sk}$ is structured, whereas~$\BA$ and~$\BB$ do not have a particular structure. 
Now, multiplying by the plaintext matrix~$\BU \in \mathbb{R}^{d \times d}$ on both sides of Equation~\eqref{eq:LZinput} gives: 
\begin{equation}
\label{eq:LZoutput}
 \rotmtx{\sk} \cdot \left(\BA \cdot \BU\right) + \left( \BB \cdot \BU\right) \ \approx \ \BM \cdot \BU \ \bmod q \enspace.
\end{equation}
Note that the error term hidden in the~$\approx$ symbol has also been multiplied by~${\bf U}$, which should be taken into account when setting parameters. Now, note that Equation~\eqref{eq:LZoutput} is of the same form as Equation~\eqref{eq:LZinput}. Defining~$a_i'$ (resp.\ $b_i'$) as the element of~$\R_{q}$ whose coefficients correspond to the $i$-th column of~${\bf A} \cdot {\bf U}$ (resp.~${\bf B} \cdot {\bf U}$) for all~$0 \le i < d$, we obtain that the~$(a_i,b_i)$'s are an encryption of~${\bf M} \cdot {\bf U}$. Overall, the $\cpmm$~${\bf M} \cdot {\bf U}$ reduces to the $\ppmm$'s ${\bf A} \cdot {\bf U}$ and ${\bf B} \cdot {\bf U}$ (modulo~$q$).

 When~$d = k N$ for some integer~$k \geq 1$, the approach generalizes as follows. Each column is encrypted using several \RLWE\ ciphertexts (in total, there are~$d^2/N$ of them). This leads to an equation of the form: 
\[
({\bf I} \otimes \rotmtx{\sk})\cdot \BA + {\bf B } \ \approx \ {\bf M} \ \bmod q \enspace .
\]
By multiplying on the right by~${\bf U}$, one again obtains a $\cpmm$ algorithm that consists of two $\ppmm$'s modulo~$q$.

\medskip
\noindent
{\bf Shared-$a$ \RLWE\ ciphertexts and large-dimensional $\cpmm$.}
To decrease the cost of large-dimensional $\cpmm$, 
we use multi-secret \RLWE\ ciphertexts sharing their $a$-parts. 
The matrix $\textbf{M}$ is provided by $d^2/N$ multi-secret \RLWE\ ciphertexts $(a_{i}, b_{j,i})_{0\le i<d, 0\le j<k}$. Note that the number of~$a_i$'s is only~$d$. We have,  over~$\R_q$: 
\[
\forall i,j: a_i \cdot \sk_j + b_{j,i} \ \approx  \ \sum_t m_{Nj+t,i}X^t \bmod q\enspace,
\]
where the $\sk_j$'s in~$\R$ are the secret keys. In matrix form, this gives: 
\begin{equation}\label{eq:msrlwe-in}
    \begin{bmatrix}
        \rotmtx{\sk_0}\\ \hline
        \vdots\\ \hline
        \rotmtx{\sk_{k-1}}
    \end{bmatrix} \cdot \BA 
    + {\bf B} \ \approx  \ {\bf M} \ \bmod q \enspace,
\end{equation}
where ${\bf A}$ (resp.\ $\textbf{B}$) consists of the coefficients of~$a_i$ (resp.\ $b_{j,i}$) and is of size $N\times d$ (resp.\ $d\times d$). As in \RLWE\ case, we multiply $\textbf{U}$ on both sides of the equation: 
\begin{equation*}
    \begin{bmatrix}
        \rotmtx{\sk_0}\\ \hline
        \vdots\\ \hline
        \rotmtx{\sk_{k-1}}
    \end{bmatrix} \cdot (\BA \cdot \BU) 
    + ({\bf B} \cdot \BU) \ \approx  \ ({\bf M}\cdot \BU) \ \bmod q \enspace.
\end{equation*}
Since it has the same form as the original matrix equation, we can view it as $d^2/N$ multi-secret \RLWE\ ciphertexts $(a_{i}', b_{j,i}')$ with respect to secret keys $\sk_0,\ldots,\sk_{k-1}$. 
This gives us a reduction from $\cpmm$ (with multi-secret \RLWE) of $\textbf{U}$ and $\textbf{M}$ to the $\ppmm$'s of $\textbf{U}$ with~$\textbf{A}$ and~$\textbf{U}$ with~$\textbf{B}$. Importantly, the number of columns of~$\textbf{A}$ is smaller than that of~$\textbf{M}$, by a factor $k$, giving a significant saving compared to~\cite{LZ22} for large~$d$.

\medskip
\noindent
{\bf \MLWE\ ciphertexts and small-dimensional $\cpmm$.}
For smaller matrices, with $dk=N$ for some integer~$k$,
\RLWE\ ciphertexts contain several columns of the input ciphertext matrix. Operating directly on these ciphertexts for the $\cpmm$ seems difficult, as it is likely to involve expensive key-switching operations.  
To handle this setup, we rely on \MLWE\ ciphertexts~\cite{BGV14,LS15}.
The matrix $\textbf{M}$ is provided by $d$ \MLWE\ ciphertexts $({\bf a}_{i}, b_{i})_{0\le i<d}$, i.e., we have,   over~$\mathbb{Z}_q[X]/(X^d+1)$: 
\[
\forall i,j: \langle \bsk, {\bf a}_i \rangle + b_{i} \ \approx  \ \sum_t m_{j,i}X^j \bmod q\enspace,
\]
where $\bsk=(\sk_0,\ldots,\sk_{k-1})\in (\mathbb{Z}_q[X]/(X^d+1))^{k}$ is the secret key. We can rewrite the above in matrix form as follows:
\begin{equation*}
     [\rotmtx{\sk_0}|\ldots|\rotmtx{\sk_{k-1}}]\cdot \BA  + \BB  \ \approx  \ \BM \ \bmod q \enspace,
\end{equation*}
where ${\bf A}$ (resp.\ $\textbf{B}$)  consists of the coefficients of~$a_i$ (resp.\  $b_{i}$) and is of size $N\times d$ (resp.\ $d\times d$). Again, multiplying both sides of the equation by~$\textbf{U}$ gives a $\cpmm$ algorithm. 
The result corresponds to $d$ \MLWE\ ciphertexts~$({\bf a}_{i}', b_{i}')$.

\medskip
\noindent
{\bf Shared-$a$ ciphertexts and $\cpmm$ with precomputation.}
Allowing precomputation, we can even reduce $\cpmm$ to a single $\ppmm$ modulo $q$ during the online phase. For ease of discussion, assume that $d=N$. Using again a shared-$a$ encryption, we can rewrite the multi-secret ciphertexts~$(a, b_i)$, which encrypt the columns of a matrix $\textbf{M}$ under secret keys $\sk_i$, in matrix form, as follows:
\begin{equation}
\label{eq:msrlwe-in2}
 \rotmtx{a} \cdot {\bf S} + {\bf B} \ \approx  \ {\bf M} \ \bmod q \enspace,
\end{equation}
where the $i$-th column of~${\bf S}$ (resp.\ ${\bf B}$ and~$\rotmtx{a}$) consists of the coefficients of~$\sk_i$ (resp.\ $b_i$ and~$x^i \cdot a \in \R_{q}$), for all~$0 \leq i < d$. 
Note that this is a different form of matrix equation from Equation~\eqref{eq:msrlwe-in}. 
By multiplying $\textbf{U}$ to both sides, we get:
\[
\rotmtx{a} \cdot ({\bf S}\cdot {\bf U}) + ({\bf B}\cdot{\bf U}) \ \approx  \ ({\bf M}\cdot{\bf U}) \ \bmod q \enspace.
\]
This has the same form as Equation~\eqref{eq:msrlwe-in2}, and hence we can view it as multi-secret \RLWE\ ciphertexts under secret key ${\sk}_i'$, where ${\sk}_i'$ consists of the $i$-th column of~$\textbf{S}'=\textbf{S}\cdot\textbf{U}$.
If the computation continues with a specific secret key, e.g., $\sk$, we can switch the keys from $\sk_i'$ to~$\sk$ after $\cpmm$, using key-switching. 
The switching keys can be precomputed  (at the cost of a $\cpmm$ modulo~$pq$ for some auxiliary modulus $p \approx q$, as switching keys require an auxiliary modulus), and the key-switching cost $\widetilde{O}(N^2)$ is relatively small compared to the cost of one~$\ppmm$.

For matrices of size $d>N$ or $d<N$, we describe analogous approaches in Section~\ref{sse:precomputation}. For $d>N$, we use multi-secret \RLWE, and for $d<N$, we use multi-secret \MLWE. In all cases, in terms of online cost, we reduce one $\cpmm$ to one $\ppmm$. 

\medskip
\noindent
{\bf $\cpmm$ with a single floating-point $\ppmm$.}
So far, all $\ppmm$'s are modular $\ppmm$s, i.e., matrix products over $\ZZ_q$ for some integer $q$. Modular $\ppmm$ is typically slower than floating-point $\ppmm$, as it does not directly benefit from the high-performance of numerical BLAS libraries, even though it can be reduced to several floating-point $\ppmm$. 

However, for approximate homomorphic encryption, the least significant bits of the $b$-parts of the \RLWE\ ciphertexts contain numerical and \RLWE\ errors that are not relevant. Thus, to multiply $\textbf{U}$ and~$\textbf{B}$, we use floating-point $\ppmm$ on the most significant bits of~$\textbf{B}$, instead of computing $\ppmm$ modulo~$q$. 

We can apply this optimization to all the $\cpmm$ algorithms mentioned above. 
It is most effective for large-dimensional matrices (as the
cost of computing~${\bf B} {\bf U}$ is then higher than that 
of computing~${\bf A} {\bf U}$) and for the algorithms using precomputation.  
In particular, for $\cpmm$ with precomputation, this optimization provides a reduction to a single floating-point $\ppmm$.

\subsubsection{Ciphertext matrix transpose.}
Before describing the first reduction from $\ccmm$ to $\ppmm$, we introduce its main technical component, namely, ciphertext matrix transpose (\cmt). 
\cmt takes as inputs ciphertexts encrypting a matrix row-by-row (resp. column-by-column), and returns ciphertexts encrypting the same matrix column-by-column (resp. row-by-row).
In this work, we focus on large matrices where each ciphertext encrypts  one single row (or column). We  note that \cmt is an interesting problem beyond its application as a tool for $\ccmm$.
For example, consider FHE scenarios with multiple parties, such as multi-party HE~\cite{MTBH21} and proxy re-encryption~\cite{GSB+23}. Each client party encrypts its data with multiple features in a ciphertext, sends ciphertext(s) to the computing server, and the server computes tasks over the aggregated ciphertexts.
During the computation, the server may need to convert the client-wise  ciphertexts into feature-wise  ciphertexts. Another $\cmt$ can convert  feature-wise ciphertexts to client-wise ciphertexts before sending the results back to each client.
This is exactly $\cmt$, and we can directly apply our \cmt algorithm to these scenarios. We refer to 
Figure~\ref{fig:data-transposition} for an illustration.

\begin{figure}[ht!]
    \centering
    \includegraphics[width=10cm]{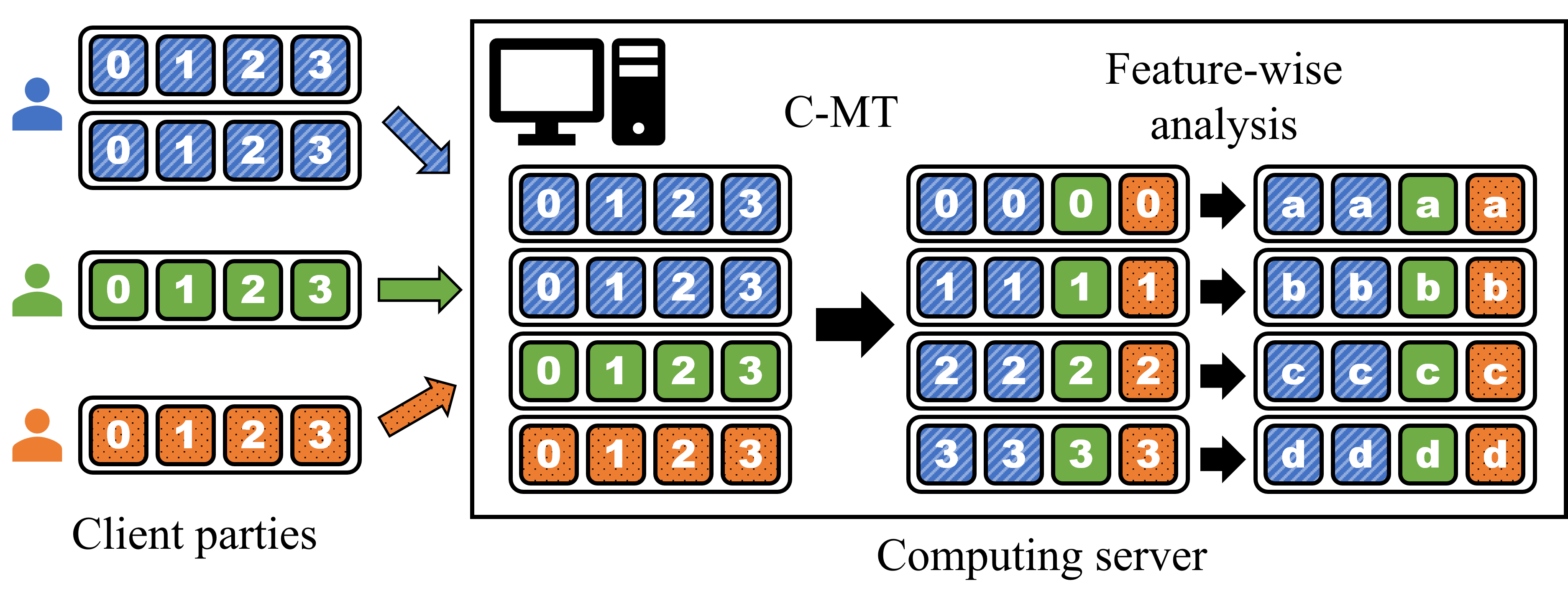}
    \caption{Visualization of a $\cmt$ application to re-formatting client-wise  ciphertexts into feature-wise  ciphertexts. }
    \label{fig:data-transposition}
\end{figure}

We now survey our \cmt algorithm. 
We start from the well-known observation on the ring $\ring_q=\mathbb{Z}_q[X]/(X^N+1)$ that
$$N\cdot m_j = \sum_{\sigma\in\gal{\ring/\mathbb{Z}}} \sigma(X^{-j}\cdot m)$$
for $0 \leq j<N$,
where $m(X)=\sum_{0 \leq i <N} m_i X^i$ is an element in $\ring_q$, and $\gal{\ring/\mathbb{Z}}$ is the group of automorphisms of~$\mathbb{Z}[X]/(X^N+1)$ induced by the Galois automorphisms of~$\mathbb{Q}[X]/(X^N+1)$ that fix~$\mathbb{Q}$. By abuse, we let the latter act on $\ring_q$; if $m = \sum_{0 \leq i <N} m_i X^i  \in \ring_q$, we define $\sigma(m) = \sum_{0 \leq i <N} m_i \sigma(X)^i \in \ring_q$. 

Given $N$ plaintexts $\{m_i=\sum_j M_{i,j}X^j\}_{0\le i<N}$ in $\ring_q$ that each encode one row (or column) of an $N\times N$ matrix $\textbf{M}$, the plaintexts $\{m'_j\}_{0\le j<N}$ that encode the transpose matrix $\textbf{M}^t$ are: 
\begin{equation*}
\begin{split}
m'_j &= \sum_{0 \leq i <N} M_{i,j}X^i = \sum_{0 \leq i <N} \left(N^{-1}\cdot\sum_{\sigma\in\gal{\ring/\mathbb{Z}}}\sigma({X^{-j}\cdot m_i})\right)\cdot X^i\\
     &= N^{-1}\cdot\sum_{\sigma\in\gal{\ring/\mathbb{Z}}} \sigma\left(\sum_{0 \leq i <N} m_i\cdot \sigma^{-1}( X^i)\right) \sigma(X^{-j})
\end{split}
\end{equation*}
for each $0 \leq j < N$. The goal of $\cmt$ is to obtain $\{m'_j\}_{0\le j<N}$ from $\{m_i\}_{0\le i<N}$ in  encrypted state.

We proceed in three steps: 
\begin{enumerate}
    \item Compute
$\{\widetilde{m}_\sigma = \sum_i m_i\cdot\sigma^{-1}( X^i)\}_{\sigma\in\gal{\ring/\mathbb{Z}}}$ from $\{m_i\}_{0\le i<N}$;
    \item Compute
$\{\Bar{m}_\sigma = \sigma(\widetilde{m}_\sigma)\}_{\sigma\in\gal{\ring/\mathbb{Z}}}$ using automorphisms;
    \item Computing $\{m'_j = \sum_\sigma {\Bar{m}_\sigma \cdot \sigma (X^{-j})}\}_{0\le j<N}$ from $\{\Bar{m}_\sigma\}_{\sigma\in\gal{\ring/\mathbb{Z}}}$. 
\end{enumerate}

The second step is achieved in  encrypted state using $N$ key-switching operations, which costs $\widetilde{O}(N^2)$ mod-$q$ arithmetic operations. 
For Steps~$1$ and~$3$, we devise a fast divide-and-conquer algorithm, reducing the cost to $O(N^2\log{N})$ mod-$q$ arithmetic operations. 
Putting it together, the overall cost of the $\cmt$ algorithm is~$\widetilde{O}(N^2)$. 

\subsubsection{$\pcmm$.}
Using the \cmt algorithm, we can reduce $\pcmm$ to $\cpmm$. We note that $\textbf{U}\cdot \textbf{M}=(\textbf{M}^t\cdot \textbf{U}^t)^t$ for matrices $\textbf{U}$ and $\textbf{M}$, even when $\textbf{M}$ is encrypted. 
To be more precise, given a plaintext matrix $\textbf{U}$ and encrypted matrix $\enc{\textbf{M}}$, we obtain the encrypted matrix $\enc{\textbf{U}\cdot \textbf{M}}$ as follows:
first, one computes $\enc{\textbf{M}^t}$ with a \cmt; then, using a $\cpmm$ algorithm, one multiplies it with the plaintext matrix $\textbf{U}^t$, obtaining $\enc{\textbf{M}^t\cdot \textbf{U}^t}$; finally, one transposes the result to get the desired $\enc{\textbf{U}\cdot\textbf{M}}$. 

As a consequence, for square matrices, we can convert all $\cpmm$ algorithms into $\pcmm$ algorithms by using two \cmt's. We note that the costs of \cmt for $d\times d$ matrix is $\widetilde{O}(dN)$, which is  relatively minor compared to $\cpmm$ unless the matrix dimension~$d$ is smaller than $N^{1/2}$. 
The non-square cases  can be handled using the GSW-based approach that we shall describe later.

\subsubsection{Ciphertext-ciphertext matrix multiplication.} 
We now outline the reduction from $\ccmm$ to $\ppmm$'s. The basic idea is to multiply two \RLWE-based encryptions in matrix forms 
\[
\toep{\sk}\cdot \textbf{A} + \textbf{B}~\approx~\textbf{M} \bmod q \enspace,
\]
while preserving the Toeplitz-like structure using $\cmt$. 
Assume that we are given two bundles of ciphertexts that encrypt each column of matrices $\textbf{M}$ and $\textbf{M}'$, respectively. In matrix form, we are given matrices~$\textbf{A}$, $\textbf{B}$, $\textbf{A}'$ and $\textbf{B}'$ such that in modulo $q$:
$$
\toep{\sk}\cdot \textbf{A} +\textbf{B}~\approx~\textbf{M} \bmod q~~~\text{and}~~~
\toep{\sk} \cdot \textbf{A}'+\textbf{B}'~\approx~\textbf{M}' \bmod q \enspace.
$$
Using the $\cmt$ algorithm, we transpose the column-wise encryption $(\textbf{A}', \textbf{B}')$ of $\textbf{M}'$. 
In matrix form, the \cmt algorithm outputs the row-wise encryption $(\wBA', \wBB')$ of $\textbf{M}'$ that satisfies
$$
\toep{\sk}\cdot\textbf{A}+\textbf{B}~\approx~\textbf{M} \bmod q~~~\text{and}~~~
\wBA'\cdot\toep{\sk}^t+\wBB'~\approx~\textbf{M}' \bmod q\enspace.
$$
We multiply the  two matrix equation above to obtain (modulo~$q$)
\begin{equation*}
    \begin{split}
\textbf{MM}'
~&\approx~(\toep{\sk}\cdot\textbf{A}+\textbf{B})\cdot
(\wBA'\cdot\toep{\sk}^t+\wBB')\\
&=~\toep{\sk}\cdot\textbf{C}_{0,0}\cdot\toep{\sk}^t 
+ \toep{\sk}\cdot\textbf{C}_{0,1}
+ \textbf{C}_{1,0}\cdot\toep{\sk}^t
+ \textbf{C}_{1,1} \enspace,
    \end{split}
\end{equation*}
where $\textbf{C}_{0,0}=\BA\wBA'$, $\textbf{C}_{0,1}=\textbf{A}\wBB'$, $\textbf{C}_{1,0}=\textbf{B}\wBA'$ and~$\textbf{C}_{1,1}=\textbf{B}\wBB'$. 
We note that the Toeplitz format is preserved.

We view $(\textbf{C}_{0,0}, \textbf{0})$ as a row-wise encryption of $\textbf{C}_{0,0} \cdot \toep{\sk}^t$, and apply the \cmt algorithm to it. Then, we obtain a row-wise encryption $(\textbf{D}_0, \textbf{D}_1)$  of $\toep{\sk}\cdot\textbf{C}_{0,0}$. We rewrite it in matrix form as follows:
$$
\textbf{C}_{0,0} \cdot  \toep{\sk}^t~\approx~ \toep{\sk} \cdot \textbf{D}_0 + \textbf{D}_1 \bmod q\enspace.
$$
Similarly, we obtain $(\textbf{D}_2, \textbf{D}_3)$ such that 
$$
\textbf{C}_{1,0}\cdot \toep{\sk}^t ~\approx~ \toep\sk \cdot \textbf{D}_2 + \textbf{D}_3 \bmod q
$$
by applying \cmt to $(\textbf{C}_{1,0}, \textbf{0})$. Putting it together, we have
$$
\textbf{M}\textbf{M}'~\approx~\toep{\sk^2}\cdot \textbf{D}_0 + \toep\sk\cdot(\textbf{D}_1+\textbf{D}_2+\textbf{C}_{0,1})  + (\textbf{D}_3+\textbf{C}_{1,1}) \bmod q \enspace.
$$
After the column-wise relinearization, made of key-switchings from $\sk^2$ to~$\sk$, we finally obtain a matrix equation
$$\toep{\sk}\cdot \textbf{A}'' + \textbf{B}''~\approx~\textbf{M}\textbf{M}' \bmod q  \enspace.
$$
This is a matrix form of column-wise encryption of the product matrix $\textbf{M}\textbf{M}'$. 
This reduces $\ccmm$ to four $\ppmm$'s, three $\cmt$'s, and $N$ key-switchings.

\subsubsection{Lightweight \cmt and $\ccmm$ with small key sizes.} 
There is a potential concern about the large key size of our $\ccmm$ and $\cmt$ algorithms. The previous $\cmt$ algorithm requires $N$ evaluation keys for each of the $N$ homomorphic automorphisms. This leads to evaluation keys of size $\widetilde{\Omega}(N^2)$, which might be problematic as~$N$ can be large. 

To handle this difficulty, we describe a lightweight $\cmt$ algorithm that uses only three evaluation keys. The basic idea is to repeatedly update and use a single evaluation key for all homomorphic automorphisms. We need one evaluation key for all homomorphic automorphisms and two other keys to ``update'' the evaluation key. 
This idea is motivated by the hierarchical key management system in~\cite{LLKN23}.
While the update procedure requires additional computations, the asymptotic complexity is the same as the original algorithm.

A lightweight $\ccmm$ algorithm  follows directly from the lightweight $\cmt$ algorithm, requiring four evaluation keys, three of which are for the lightweight $\cmt$ algorithm and one for relinearization. 

 \subsubsection{General reduction using \RGSW-format ciphertexts.}
We describe the \RGSW-like approach, and start by recalling how scalar homomorphic computation is handled in~\cite{GSW13}. \RGSW\ encryption is compatible with the use of \emph{gadget decomposition}; our approach extends \emph{mutatis mutandis} to that setting (see, e.g., \cite{BEHZ16} for RNS gadget decomposition). 
To be more precise, we can multiply a \RGSW-format ciphertext modulo~$pq$ with a \RLWE-format ciphertext in modulus~$q$, where~$p$ is the auxiliary modulus, similarly to \RLWE\ key switching introduced in~\cite{GHS12}.

Let  $p$ be an integer. A \RGSW\ encryption of $u \in \R_{q, N}$ under the secret key $\sk \in \R_N$ consists of two pairs $(a, b), (a', b') \in \R_{pq, N}^2$ such that 
\[
  p u \approx a \cdot \sk + b \bmod pq~~~\mbox{and}~~~p u \cdot \sk \approx a' \cdot \sk + b' \bmod pq \enspace. 
\]
We now recall the so-called external product between a  $\RGSW$ ciphertext and a \RLWE-ciphertext that produces a  \RLWE-ciphertext.  Let $(a, b), (a', b')\in\R_{pq, N}^2$ be a \RGSW\ encryption of~$u$ and $(\alpha, \beta) \in\R_{q,N}^2$ be a \RLWE\ encryption of $v \in\R_{q,N}$. Then
\[
   (\alpha', \beta') = \left(\left\lfloor \frac{a' \alpha + a \beta}{p} \right\rceil, \left\lfloor \frac{b' \alpha + b \beta}{p} \right\rceil \right) \in \R_{q,N}^2 
\]
is such that~$\alpha' \cdot \sk + \beta' \approx u\cdot v \bmod q$. For the error to be sufficiently small, it is required that~$p \approx q$. Note that this can be viewed as reducing a homomorphic multiplication with modulus~$q$ to four plaintext multiplications modulo~$pq$. 

We extend \RGSW\ encryption to matrices in the form of a pair of matrix encryptions\footnote{A more general version of the RGSW format consists of encryptions of $p\BM$ and $p \BM \BS'$ under the secret matrix~$\BS$, for $\BS'\ne \BS$, like RLWE key-switching keys; we will use this more general version in Section~\ref{sse:gsw}.} 
\[
    p\BM  \approx \BS  \BA_0 + \BB_0  \bmod pq~~~\mbox{and}~~~p\BM \BS  \approx \BS   \BA_1 + \BB_1 \bmod pq \enspace. 
\]
These encryptions are compatible with all the various matrix encryption formats (shared-$a$-based, \MLWE-based, shared-secret-based) considered for $\cpmm$. The external product can be extended to this setting to handle the product of an \RGSW-encrypted matrix by an \RLWE-encrypted vector, so that, if $\bv$ is a vector encrypted as $(\ba, \bb)$, in the sense that $\BS\cdot \ba + \bb \approx \bv \bmod q$, we have that
    \[
    (\ba', \bb') = \left( \left\lfloor \frac{1}{p}(\BA_1\ba + \BA_0\bb)\right\rceil ,  \left\lfloor \frac{1}{p}(\BB_1\ba+\BB_0\bb)\right\rceil  \right)
    \bmod q
    \]
is an \RLWE-encryption of the product $\BM \cdot \bv$ using the same encryption format as~$\BM$. This extends from a vector  to a matrix by viewing the matrix as a family of vectors. 

This yields a reduction of any single linear algebra operation to at most four modular operations of similar type on plaintext data, with a modulus that is twice larger. The dimensions may vary depending on the encryption format.

\subsection{Related works}
There has been a massive effort on encrypted linear algebra, with a large body of works considering the encrypted Mv and MM problems. 
Among those, we may mention~\cite{JKLS18,HCJG24,MMJG24,ZL23,HHW+21,GQH+24,ZLW23,CYWLF24,RT22,SCCTA20,HWC24,JLKN+22} in the MM case, and~\cite{HS14, HS18, ACLS18, ALPRSSY21, CHAM23, Rhombus24,SCCTA20} in the Mv case. None of these provides reductions to cleartext Mv/MM, preventing from using highly efficient cleartext linear algebra like BLAS. 
Their limitation stems from the intertwining of the linear algebra and homomorphic operations. 

We are aware of only two previous works that provide algorithms for linear algebra which proceed by reduction to cleartext data. We already mentioned~\cite{LZ22}, which provides a reduction from one $\cpmm$ to two $\ppmm$'s. However, it requires matrix dimensions above or equal to the RLWE ring-degree. In practice, to enable computations with a reasonable precision and 128-bit security, the matrix dimension should be at least~$2^{12}$. 
Another work~\cite{GHV10} adapts the BGN cryptosystem~\cite{BGN05} to LWE. It also provides a reduction from $\ccmm$ to $\ppmm$'s. However, this is a pre-FHE scheme that supports only a single matrix-matrix multiplication and
no further homomorphic computations beyond additions.  In particular, the output ciphertexts do not have the same structure as its input ciphertext. In contrast, our $\ccmm$ algorithms can be used in an FHE computation.

Following the initial submission of this article, and during the review process, two relevant works on encrypted matrix multiplication appeared. In~\cite{GL25}, the authors proposed a homomorphic encryption scheme that supports fast encrypted matrix arithmetic, including matrix-matrix multiplications and matrix transposition. Both~\cite{GL25} and~\cite{CKL25} consider fast batch $\ccmm$, i.e., $\ccmm$ on multiple (pairs of) inputs in parallel. We observe that the solution from~\cite{CKL25} is directly compatible with the CKKS scheme, while that of~\cite{GL25} uses a different algebraic framework. These are reductions to batch $\ppmm$'s.

\subsection{Additional contributions compared to the conference publications}
This work is an extended version of articles published by the same authors in the proceedings of the Crypto~2024 conference~\cite{BCHPS24} and by the fourth author in the proceedings of the Eurocrypt 2025 conference~\cite{Park25}. It includes the fast encrypted matrix multiplication algorithms of both articles, and the reductions from encrypted matrix multiplications to cleartext matrix multiplications. 

Building on the two conference articles, this work summarizes and generalizes the framework for homomorphic linear algebra across various ciphertext formats, including RGSW format. 
Additionally, it introduces new algorithms and reductions for performing matrix-vector multiplications on encrypted data using the RGSW format in Section~\ref{sse:gsw}. 
This work also improves on~\cite{BCHPS24} by explicitly describing several strategies for reducing modular $\ppmm$ to floating-point $\ppmm$ and providing comprehensive experiments for the proposed algorithms in Section~\ref{ssec:impl-aspects}.
\section{Preliminaries}\label{sec:prelim}
\label{sec:omega}\label{sec:R}
Vectors (resp.\ matrices) are denoted in bold lower-case (resp.\ upper-case) letters. Unless explicitly stated otherwise,we consider column vectors. We let $\langle \cdot,\cdot \rangle$ denote the formal dot product of two vectors: for any ring~$R$ and dimension~$k$,
for $\bx = (x_i)_i$ 	and~$\by = (y_i)_i \in R^k$, we write
$\langle \bx, \by \rangle = \sum_i x_i y_i$. 
For~$N$ a power of two and~$q \geq 2$, we define $\R_N = \ZZ[X] / (X^N+1)$ and\label{not:R} $\R_{q,N} = \ZZ_q[X] / (X^N+1)$. We skip the index~$N$ when it is clear from the context.  For $v\in\R_{q,N}$, we let $\VVec{v} \in \ZZ_q^{N}$ be the coordinate vector of~$v$ in the monomial basis. 
Conversely, if $\bv \in \ZZ_q^{N}$, we write $\Pol{\bv} = \sum_{0 \leq i < N} v_{i} X^i\in\R_{q,N}$. As in Section~\ref{sec:technical_overview}, the Toeplitz matrix of $m=\sum_i{m_iX^i} \in \R$  is defined as $$\toep{m}=
\begin{bmatrix} 
m_0&-m_{N-1}&\ldots&-m_{1}\\
m_{1}&m_0&\ldots&-m_{2}\\
\vdots&\vdots&\ddots&\vdots\\
m_{N-1}&m_{N-2}&\ldots&m_{0}
\end{bmatrix} \enspace.$$

For a real number~$x$, we let 
$\lfloor x \rceil$ denote the integer part of~$x + 1/2$.
The notation~$\log$ refers to the base-2 logarithm. We let~$\omega$ refer to the complexity exponent~\label{not:omega} of square matrix multiplication and assume that~$\omega>2$. Finally, in complexity analyses, we use the notation $\widetilde{O}(f(n_1, \dots, n_\ell))$ to mean 
\[
\bigcup_{(k_1, \dots, k_\ell) \in \ZZ_{\ge 0}^{\ell}} O\left(f(n_1, \dots, n_\ell) (\log n_1)^{k_1} \dots (\log n_\ell)^{k_{\ell}}\right).\]

\subsection{\RLWE\ and \MLWE\ ciphertext formats}
\label{sec:RLWE}

For $N$ a power of two  and $q\geq 2$ an integer, \label{not:RLWE} a $\RLWE_{q,N}^{\phantom{2}}$ ciphertext for a plaintext~$m \in \R_{q,N}^{\phantom{2}}$ under a secret key~$\sk \in \R_{N}^{\phantom{2}}$ is a pair $(a,b) \in \R_{q,N}^2$ such that  $a\cdot \sk +b \approx m \bmod q$. In this work, we consider plaintexts~$m$ that are:
\begin{itemize}
    \item[$\bullet$] small, i.e., have coefficients with small absolute values compared to~$q$; 
     \item[$\bullet$] approximate, i.e., $m$ could as well be $m+e$ for a small~$e$. 
\end{itemize}
Such ciphertexts can be generated in a symmetric manner, if the encryptor knows~$\sk$, or in an asymmetric manner if the encryptor knows encryptions of~$0$. The ciphertexts are indistiguishable from uniform for any adversary that does not know~$ \sk$, under the RLWE hardness assumption~\cite{SSTX09,LPR10}.

\smallskip

 The $\MLWE$ ciphertext format~\cite{BGV14,LS15} \label{not:MLWE} is the following generalization. Let~$k \geq 1$. An $\MLWE_{q,N}^{(k)}$ ciphertext for a plaintext~$m \in \R_{q,N}^{\phantom{2}}$ under a secret key~$\bsk \in \R_{N}^k$ is a pair $({\bf a},b) \in \R_{q,N}^k \times \R_{q,N}^{\phantom{k}}$ such that  $\langle {\bf a},  \bsk \rangle +b \approx m \bmod q$. Here~$k, N$ and~$q$ are  respectively referred to as (module) rank, (ring) degree and modulus. Note that $\MLWE_{q,N}^{(1)}$ and $\RLWE_{q,N}^{\phantom{(1)}}$ coincide.

\subsection{The CKKS  scheme}\label{sec:keyswitch}\label{sec:s2c}\label{sec:ecdslot}\label{sec:ls2c}
CKKS~\cite{CKKS17} is an RLWE-based FHE scheme that performs approximate computations on real and complex numbers. The CKKS 
plaintexts are elements of~$\R_N$. 
Even though the main encoding of data into plaintexts used in the context of CKKS is the \emph{slot-encoding}, in the present work, we only consider the \emph{coefficient-encoding}, which we now describe.

\smallskip
\noindent
{\bf Encoding.}
For $z\in \mathbb{R}^N$, we define 
\[ \Ecd_{\coeff}(z) = \sum_{0 \leq i<N} \lfloor \Delta \cdot z_i \rceil X^i \in \R \enspace, \]
where~$\Delta$ is a scaling factor.
The corresponding decoding function is defined, for $p = \sum_{0 \leq i<N} p_i X^i \in \R$, by 
\[
\Dcd_{\coeff}(p) = \left(\frac{p_i}{\Delta}\right)_{0\le i < N} \in \RR^{N}\enspace. 
\]

Coefficient encoding is typically mostly used for internal operations, in the bootstrapping process~\cite{CHKKS18b}, but will be the only one considered in this work. 

\smallskip
\noindent
{\bf Basic functionalities.}
CKKS provides the following elementary operations on keys and ciphertexts:
\begin{itemize}
\item[$\bullet$] $\KeyGen$. Given a security parameter~$\lambda$, $\KeyGen(1^\lambda)$ 
  returns public key $\pk  \in \R_{q}^2$ for some~$q$, and a secret key $\sk\in \R$.
\item[$\bullet$] $\Enc$.  For a plaintext~$m$ (which is obtained by an encoding function as explained above), $\Enc(\pk, m)$ returns a \RLWE-format ciphertext~$(a,b) \in \R_{q}^2$ such that~$a \cdot \sk +b = m +e \bmod q$, where $e$ has small-magnitude coefficients. 
\item[$\bullet$] $\Dec$. Given a ciphertext~$(a,b)  \in \R_{q}^2$, $\Dec(\sk, (a,b))$ returns $a \cdot \sk +b \bmod q$. 
\item[$\bullet$] $\SWKGen$. Given two integers $p, q$ and two keys $\sk, \sk'\in \R_{N}$, $\SWKGen(\sk, \sk')$
  returns a switching key $\swk_{q,p,\sk \rightarrow \sk'}$ from
  key~$\sk$ to key~$\sk'$ for ciphertexts modulo~$q$ and with auxiliary
  integer~$p$. It is of the form:
  \[ 
\swk_{q,p,\sk \rightarrow \sk'} 
= (a_{\swk}, b_{\swk}) 
= (a, -a \cdot \sk' + e + p \cdot \sk)
\in\R_{pq}^2 \enspace ,
\]
where~$e \in\R_N$ has small-magnitude coefficients.\footnote{Switching keys can also be defined with more ring
elements when using a so-called gadget rank larger than~$1$ (see~\cite{BEHZ16}). Our
techniques also carry over to that setup, but we restrict ourselves to
switching keys as above for the sake of simplicity.}

\item[$\bullet$] $\KeySwitch$\label{not:keyswitch}. Given as input a mod-$q$ ciphertext~$\ct$ for a message~$m$ under~$\sk$, 
$\KeySwitch(\swk_{q,p, \sk\rightarrow \sk'}, \ct)$ returns a mod-$q$ ciphertext~$\ct'$ for~$m$ under~$\sk'$. From a computational viewpoint, key-switching reduces to $O(1)$ polynomial multiplications over $\ZZ_{pq}$, hence has a complexity of~$O(N\log N)$ operations in $\ZZ_{pq}$, where the involved constant in the~$O(\cdot)$ notation depends on the so-called \emph{decomposition number} \cite{GHS12, HK20}. 
\end{itemize}

\noindent 
{\bf Error.} CKKS is an inherently approximate homomorphic encryption system. Errors occur from several sources:
\begin{itemize}
    \item[$\bullet$] using CKKS, input data is discretized, inducing an error which propagates throughout all the computation; this error is similar to the one occurring in a cleartext fixed-point evaluation; 
    \item[$\bullet$] as an RLWE-based cryptosystem, CKKS-encrypted material (ciphertexts, keys) includes a cryptographic error; 
    \item[$\bullet$] finally, homomorphic computations introduce further errors, the most important ones being rescaling errors and key-switching errors; these  can be controlled by tuning homomorphic parameters (scaling factor, auxiliary modulus, gadget decomposition). 
\end{itemize}   
 
CKKS encoding does not separate the errors from the various sources. In the sequel, we thus use ``error'' as a generic term for all those error types. As we discuss linear algebra, there is little to be done on the first type of error, the second one is intrinsic, and we mostly focus on the third one.  
We account for these errors using the $\approx$ symbol, which means that  equations are valid up to the accumulation of errors coming from the three sources  discussed in this paragraph. 

\pagebreak
\noindent
{\bf Homomorphic operations.}
We now describe a subset of the homomorphic operations provided 
by CKKS. 
\begin{itemize}
\item[$\bullet$] $\Add$. 
  Given two ciphertexts~$\ct$  and~$\ct'$ for the same modulus~$q$, $\Add$ computes and returns $\ct'' = \Add(\ct, \ct') = \ct + \ct'$. If~$\Dec(\sk, \ct) \approx m$ and~$\Dec(\sk, \ct') \approx m'$, then we have~$\Dec(\sk, \ct'') \approx m + m'$. 
\item[$\bullet$] $\PCMult$. Given a ciphertext $\ct = (a, b) \in \R_{q}^2$ and a monomial $X^t \in R$, $\PCMult$ computes and returns $\ct' = (X^t \cdot a, X^t \cdot b) \in \R_{q}^2$. If $\Dec(\sk, \ct) \approx m$, then $\Dec(\sk, \ct') \approx X^t \cdot m$. We restrict to this particular case of plaintext-ciphertext multiplication as this  is the only one we require. Note that it does not  increase the error. 
\item[$\bullet$] $\Aut$. Given $\ct$ such that $\Dec(\sk, \ct) \approx \Ecd_{\coeff}(x_0, \ldots, x_{N-1})$ and an odd index $\ell$, $\Aut$ computes and returns~$\ct'$ such that $\Dec(\sk, \ct') \approx \Ecd_{\coeff}((-1)^{\lfloor \ell' i / N \rfloor}\cdot x_{\ell' i \bmod N})_{0\le i < N}$, where $\ell'$ is the inverse of~$\ell$ modulo~$2N$. This requires an automorphism key~$\rk_\ell$, which is a switching key from $\sk(X^\ell)$ to $\sk$. We assume that the corresponding switching key is given as an input, and we omit it for the sake of readability.

\item[$\bullet$] $\relin$. Given a ciphertext $\ct$ that decrypts to $m$ under the secret $\sk^2\in\ring$, $\relin$ computes and returns~$\ct'$ that decrypts to~$m$ under the secret $\sk\in\ring$. 
It requires a relinearization key, namely a switching key from~$\sk^2$ to~$\sk$. 
\end{itemize}

We shortly discuss levels and rescaling,  which is required in order to correct the scaling factors after a multiplication of two encoded objects. As $\Ecd(x) \approx \Delta \cdot x$, the multiplication of two encodings leads to a result that is scaled by a factor~$\Delta^2$:
\[
\Ecd(x) \times  \Ecd(x') \ \approx \ 
\left(\Delta \cdot \Pol{x}\right)\times  \left(\Delta \cdot \Pol{x'}\right) \ \approx \ \Delta^2 \cdot (\Pol{x} \times \Pol{x'})\enspace ,
\]
where~$\times$ refers to the polynomial product in $\R$. In order to restore a $\Delta$-scaled encoding,
one can divide this result by~$\Delta$. At the ciphertext level, this is approximately achieved by means of the following function.
\begin{itemize}
\item[$\bullet$] \label{not:rescale}$\Rescale_{q\rightarrow q'}$. Given $q > q'$ such that $q/q' \approx \Delta$ and a ciphertext~$\ct$ modulo~$q$, $\Rescale$ computes and returns
  $\lfloor (q'/q) \cdot \ct \rceil \in \R_{q'}^2$. If~$\ct$ decrypts to~$m$, then so does~$\ct'$ (up to some error). Note that the ciphertext modulus has changed from~$q$ to~$q'$.\footnote{Most CKKS descriptions usually account for this change through the notion of level. As all our algorithms use one single level, we restrict to the two moduli $q$ and $q'$; the reader more familiar with the level view should think of moduli $Q_{\ell}$ and $Q_{\ell-1}$, for some integer $\ell$.}  
\end{itemize}
\section{Formats}
We turn to a discussion of encryption formats for vectors $\bv \in \RR^{d_1}$ encoded as plaintext vectors $\lfloor \Delta \bv \rceil \in \ZZ_q^{d_1}$ and matrices $\BM \in \RR^{d_1\times d_2}$ encoded as plaintext matrices as $\lfloor \Delta \BM \rceil \in \ZZ_q^{d_1\times d_2}$. To ease notations, we drop the~$\lfloor \cdot \rceil$ and shall use approximate equalities  to account both for the encoding error (bounded from above by~$1/2$) and cryptographic error stemming from \RLWE\ encryption. The approximate equality $u \approx v \bmod q$ means that the (unique) representative of $u-v \in \RR/q\ZZ$ in $(-q/2, q/2]$ is small. 

Our goal is to describe and characterize a number of encryption formats for vectors of matrices, in order to be able to efficiently handle  various dimensions both smaller and larger than the RLWE ring-degree. Ultimately, these formats also yield a simple and uniform algebraic description, with structure constraints on one of the matrices. In order for the reader to be able to spot at a glance the structure constraints, we  designate the matrices to which these constraints apply by a $\isst{}$ symbol. 

We shall start with \RLWE\ formats, and prove that most of these formats, in their column-based version where each ciphertext encrypts a single column or a portion of it, can be described by approximate linear algebra identities of the form
\begin{align}
    \isst{\BS} \cdot \ba + \bb & \approx \Delta \cdot \bv \bmod q \mbox{\ (for vectors)} \enspace, \label{eq:enc_vec}\\
    \isst{\BS} \cdot \BA + \BB & \approx \Delta \cdot \BM \bmod q \mbox{\ (for matrices)}\enspace, \label{eq:enc_mat}
\end{align}
where the matrix $\isst{\BS}$ has a structured shape which characterizes a given format: conversely, Equation~\eqref{eq:enc_vec} (resp.\ Equation~\eqref{eq:enc_mat}) with a matrix~$\isst{\BS}$ having the right structure shows that $(\ba, \bb)$ (resp.\ $(\BA, \BB)$) is an encryption of a vector (resp. matrix) in the corresponding format. This result is formalized in Lemmas~\ref{le:vec_encodings} (for vectors) and~\ref{le:mat_struct_s} (for matrices). 

Two of our matrix formats, which we shall identify as ``structured-$\BA$ formats'' swap the role of $\BA$ and $\BS$ through a different use of the shared-$a$ technique; in that case, the identity characterizing the column version will be
\[
\isst{\BA} \cdot \BS + \BB \approx \Delta \cdot \BM \bmod q \enspace, 
\]
where $\isst{\BA}$ is structured while $\BS$ may be arbitrary. This is formalized in Lemma~\ref{le:mat_struct_a}. 

These formats allow for the most efficient reductions, in the context where the matrix $\BU$ by which we want to multiply $\BM$ is cleartext and known beforehand, and precomputations are allowed. 
All those \RLWE\ formats have row versions, which are obtained by transposing the identities above; the structure of the matrix $\BS$ (or~$\BA$) is, in that case, also transposed. 
Finally, we shall also derive a \RGSW-version of those \RLWE\ formats.

 \subsection{Column vector encryption formats}\label{sec:vec_encryption formats}
We shall separate the discussion into three cases, depending on the dimension $d_1$ of the vector: $d_1 = N$, $N < d_1$ and $d_1 < N$. In the last two cases, we shall assume that the vector is padded so that $N | d_1$ (resp.~$d_1 | N$~--~note that when $d_1 \in (N/2, N)$, padding leads to~$d_1 = N$). As CKKS uses power-of-two ring degrees~$N$, this can be achieved at the cost of an increase of~$d_1$ by a factor~$\leq 2$.

\subsubsection{\RLWE\ format.}\label{se:vec_rlwe}
We start with the case~$N = d_1$. 
The encoded plaintext vector $\lfloor \Delta \bv \rceil \in \ZZ_q^{N}$ is associated to the plaintext polynomial $\Pol{\lfloor \Delta \bv \rceil}$, which we represent as $(a, b)$ using \RLWE\ encryption, namely 
\[
a\cdot \sk + b \approx \Delta\cdot \Pol{\bv} \bmod q \enspace. 
 \]
For $s \in \ring_N$, $\toep{s}$ is the matrix of the linear map $u \mapsto u\cdot s$ of $\R_N$ in the monomial basis. We can thus rewrite the previous identity as 
\[
\toep{\sk}\cdot \VVec{a} + \VVec{b} \approx \Delta \cdot \bv \bmod q \enspace. 
\]
Setting $\BS_{\RLWE}(\sk) = \toep{\sk} \in \ZZ_q^{N\times N}$, $\ba = \VVec{a}$ and $\bb = \VVec{b}$, this identity becomes
\[
\BS_{\RLWE}(\sk) \cdot \ba + \bb \approx  \Delta \cdot \bv \bmod q \enspace.
 \]

\subsubsection{Shared-$s$ encryption format.}
We move to the case where the dimension $d_1$ is divisible by $N$. We split the encoded input vector $\bv = (v_i)_{0\le i <d_1}$ into $N$-dimensional vectors $\bv^{(j)} = (v_{jN + i})_{0\le i< N}$, for $0\le j < d_1/N$. Each of these vectors is represented in basic \RLWE\ format as above, i.e., as $(a_j, b_j) \in\R_{q,N}^2$ such that 
\[
\BS_{\RLWE}(\sk)\cdot \VVec{a_j} + \VVec{b_j} \approx \Delta \cdot \bv^{(j)} \bmod q, \ \ 0\le j < \frac{d_1}{N}\enspace.
\]

We define $\BS_{\shs}(\sk) = {\bf I}_{d_1/N} \otimes \BS_{\RLWE}(\sk) \in \ZZ_q^{d_1\times d_1}$ and $\ba, \bb \in \ZZ_q^{d_1}$ as the $d_1$-dimensional vectors obtained by vertically concatening  the vectors $\VVec{a_j}$ and $\VVec{b_j}$ for $0\le j < d_1/N$. With these notations, we obtain the identity 
\[
\BS_\shs(\sk) \cdot \ba + \bb \approx \Delta \cdot \bv \bmod q \enspace.
\]

\subsubsection{Shared-$a$ encryption format.}
In the same setting $N|d_1$, sharing the $a$-part allows a different encryption format. We define~$k = d_1/N$. 
In that case, we assume that  the ciphertexts encrypting the vectors~$\bv^{(j)}$ are in the shared-$a$ format, i.e., that we have
\[
   a \cdot \sk^{(j)} + b^{(j)} \approx \Delta \cdot \Pol{\bv^{(j)}} \bmod q, \ \ 0\le j < k \enspace, 
\]
or, in matrix terms
\[
   \toep{\sk^{(j)}}\cdot \VVec{a} + \VVec{b^{(j)}} \approx \Delta \cdot \bv^{(j)} \bmod q,  \ \ 0\le j < k\enspace. 
\]

Let $\BS_\sha((\sk^{(j)})_{0\le j < k}) \in \ZZ_q^{d_1 \times N}$  be the vertical concatenation $(\toep{s_0}^t | \toep{s_1}^t | \ldots | \toep{s_{k-1}}^t)^t$, $\ba = \VVec{a}$ and 
$\bb \in \ZZ_q^{d_1}$ be the vertical concatenation of the vectors $\VVec{b^{(0)}}, \ldots, \VVec{b^{(k-1)}}$. We obtain the matrix form
\[
\BS_\sha((\sk^{(j)})_{0\le j < k}) \cdot \ba + \bb \approx \Delta \cdot \bv \bmod q \enspace, 
\]
where $\BS_\sha((\sk^{(j)})_{0\le j < k}) \in \ZZ_q^{d_1\times N}$, $\ba \in \ZZ_q^{N}$ and~$\bb \in \ZZ_q^{d_1}$. 

\subsubsection{\MLWE\ encryption format.}
In the case where $d_1 < N$, we assume that $d_1$ divides $N$ and write $\ell = N/d_1$. We then use \MLWE\ encryption with dimension $\ell$ and ring degree~$d_1$.  
If $(a^{(0)}, \ldots, a^{(\ell - 1)}, b)$ is such an encryption of $\Pol{\lfloor \Delta \cdot \bv\rceil}$ viewed as an element of $\R_{q,d_1}$, we have 
\[
\sum_{0 \leq j<\ell} a^{(j)} \cdot \sk^{(j)} + b \approx  \Delta \cdot \sum_{0 \leq i< d_1} v_{i} X^i \bmod (q, X^{d_1}+1) \enspace, 
\]
where $a^{(j)}, b \in\R_{q,d_1}$ and $\sk^{(j)} \in\R_{d_1}$ for all~$j <\ell$.
If we let $\BS_\MLWE((\sk^{(j)})_{0\le j < \ell})\in \ZZ_q^{d_1 \times N}$ be the horizontal concatenation $(\toep{\sk^{(0)}} | \toep{\sk^{(1)}} | \ldots | \toep{\sk^{(\ell-1)}})$,  $\ba\in \ZZ_q^{N}$ be the vertical concatenation of $\VVec{a^{(0)}}, \ldots, \VVec{a^{(\ell-1)}}$ and $\bb = \VVec{b}$, then we obtain the identity 
\begin{equation}\label{eq:mlwe}
\BS_\MLWE((\sk^{(j)})_{0\le j < \ell}) \cdot \ba + \bb \approx \Delta \cdot \bv \bmod q\enspace, 
\end{equation}
with $\BS_\MLWE((\sk^{(j)})_{0\le j < \ell}) \in \ZZ_q^{d_1\times N}$, $\ba \in \ZZ_q^{N}$ and~$\bb \in \ZZ_q^{d_1}$.  

\medskip
We summarize the above discussion in the following unified lemma. 
\begin{lemma}\label{le:vec_encodings}
Let $\bv \in \RR^{d_1}$ be a vector, with $d_1 | N$ or $N | d_1$; we can then write the \{\RLWE, shared-$a$ \RLWE, \MLWE\} encryption under $\sk$ of its CKKS-encoding as follows. There exist a matrix $\isst{\BS}\in \ZZ_{q}^{d_1\times N}$, a vector~$\ba \in \ZZ_q^{N}$ and a vector $\bb \in \ZZ_q^{d_1}$ such that  
\begin{equation}\label{eq:col_lwe}
\isst{\BS} \cdot \ba + \bb \approx \Delta \cdot \bv \bmod q \enspace.
\end{equation}

Furthermore, if Equation~\eqref{eq:col_lwe} holds with
\begin{itemize}
    \item[$\bullet$]  $\isst{\BS} = \BS_\RLWE(\sk)$ for some $\sk\in \ring_N$, then $(\ba, \bb)$ is a \RLWE\ encryption of  $\bv$;
    \item[$\bullet$]  $\isst{\BS} = \BS_{\sha}((\sk^{(j)})_{0\le j < d_1/N})$ for $\sk^{(j)} \in \ring_N$ for all~$j$, then $(\ba, \bb)$ is a shared-$a$ encryption of~$\bv$; 
    \item[$\bullet$]  $\isst{\BS} = \BS_{\MLWE}((\sk^{(j)})_{0\le j < N/d_1})$ for $\sk^{(j)} \in\R_{d_1}$ for all~$j$, then $(\ba, \bb)$ is an \MLWE\ encryption of~$\bv$. 
\end{itemize}

When $N|d_1$, we can write the shared-$s$ encryption of the CKKS-encoding of~$\bv$ in the following form.  There exist a matrix $\isst{\BS}\in \ZZ_{q}^{d_1\times d_1}$, a vector $\ba \in \ZZ_q^{d_1}$ and a vector $\bb \in \ZZ_q^{d_1}$ such that  
\begin{equation}\label{eq:shared-s-2}
\isst{\BS}\cdot \ba + \bb \approx \Delta \cdot \bv \bmod q \enspace.
\end{equation}
In particular, if Equation~\eqref{eq:shared-s-2} holds with $\isst{\BS} = \BS_{\shs}(\toep{s})$ for some $s\in \ring_N$, then $(\ba, \bb)$ is a shared-$s$ encryption of~$\bv$.
\end{lemma}

In summary, Equations~\eqref{eq:col_lwe} and~\eqref{eq:shared-s-2} with $\BS$ having the suitable structure and dimensions characterize the various encryption formats for vectors. 
Lemma~\ref{le:vec_encodings} also has a row encryption format version, which can be derived directly or obtained by transposition, in which $\ba$, $\bb$ and~$\bv$ are row vectors and Equation~\eqref{eq:col_lwe} is replaced by 
\[
\ba \cdot {\isst{\BS}}^t + \bb \approx \Delta\cdot \bv \bmod q \enspace. 
\]

\subsection{Matrix formats with structured $\BS$}\label{sec:mat_enc_formats}
We now turn to matrix formats. Let $\BM \in \ZZ_q^{d_1\times d_2}$. The algebraic roles of $a$ and $\sk$ being symmetrical in \RLWE\ encryption, we shall build two types of matrix formats. The first one, and the one that we shall mostly use, is a direct generalization of Lemma~\ref{le:vec_encodings}; we shall call it the \emph{structured-$\BS$} case.  
By independently handling the columns of~$\BM$ using Section~\ref{sec:vec_encryption formats}, we obtain the following generalization of Lemma~\ref{le:vec_encodings}.
\begin{lemma}\label{le:mat_struct_s}
   Let $\BM \in \RR^{d_1 \times d_2}$ be a matrix with $d_1 | N$ or $N | d_1$. We can write its \{\RLWE, shared-$a$ \RLWE, \MLWE\} encryption in the following form: there exist a matrix $\isst{\BS}\in \ZZ_{q}^{d_1\times N}$, a matrix $\BA \in \ZZ_q^{N \times d_2}$ and a matrix $\BB \in \ZZ_q^{d_1 \times d_2}$ such that  
\begin{equation}\label{eq:mat_lwe}
\isst{\BS}\cdot \BA + \BB \approx \Delta\cdot \BM \bmod q \enspace.
\end{equation}

Further, when $N|d_1$, we can write its shared-$s$ encryption in the following form.  There exist a matrix~$\isst{\BS}\in \ZZ_{q}^{d_1\times d_1}$, a matrix $\BA \in \ZZ_q^{d_1 \times d_2}$ and a matrix $\BB \in \ZZ_q^{d_1 \times d_2}$ such that  
\begin{equation}\label{eq:sh-s}
\isst{\BS}\cdot \BA + \BB \approx \Delta \cdot \BM \bmod q \enspace.
\end{equation}
\end{lemma}

As in Lemma~\ref{le:vec_encodings}, Equations~\eqref{eq:mat_lwe} and~\eqref{eq:sh-s} with a matrix $\isst{\BS}$ of the right format conversely characterize each encryption format. 
Again, we can derive the corresponding matrix row encryption formats, under the form
\begin{equation}
\label{eq:mat_row_format}
\BA \cdot {\isst{\BS}}^t + \BB \approx \Delta \cdot \BM \bmod q \enspace, 
\end{equation}
  where the transposition of $\isst{\BS}$ shows that the structure of the secret matrix is transposed compared to column encryption formats.
 
\subsection{Matrix formats with structured $\BA$}
There is a dual way to view the shared-$a$ and \MLWE\ formats, which is obtained by swapping the roles of~$a$ and~$\sk$. 
We explore this dual format, which we call \emph{structured-$\BA$}. We shall use it in Section~\ref{sec:algorithms} to describe efficient $\cpmm$ algorithms with precomputation, when the plaintext matrix is known beforehand.

\subsubsection{Shared-$a$ format with structured $\BA$.}
Above, the shared-$a$ format was sharing the same $a$-part among all parts of a given column vector. Here, we introduce a dual version of this format, sharing $a_i$ among the $i$-th parts of all column vectors. This yields the following polynomial identity: 
\begin{equation}\label{eq:shasta}
  a_i \cdot \sk_{j} + b_{ij} \approx \sum_{0 \leq k<N} \Delta \cdot m_{Ni + k, j} X^k \bmod q, \ \ 0\le i < d_1/N, \ \ 0\le j < d_2 \enspace. 
\end{equation}
which we rewrite, in vector form, as: 
\[
  \toep{a_i} \cdot \VVec{\sk_{j}} + \VVec{b_{ij}} \approx \Delta \cdot \VVec{(m_{Ni + k, j})_{0\le k < N}} \bmod q, \ \ 0\le i < d_1/N, \ \ 0\le j < d_2\enspace. 
\]
We stress that~$\sk_j$ does not depend on~$i$: if we convert several packs of ciphertexts (here, the packs corresponding to each ``horizontal slice'' $0\le i < d_1/N$ of the matrix) to the shared-$a$ form using a switching key in the same format, then they will end up having the same secret keys~$\sk_j$. 

Defining $\isst{\BA} = (\toep{a_0}^t | \ldots | \toep{a_{d_1/N-1}}^t)^t \in \ZZ_q^{d_1\times N}$, and $\BS \in \ZZ^{N\times d_2}$ as the horizontal concatenation of the vectors $\VVec{\sk_{j}}, 0\le j < d_2$, the latter identity then becomes, in matrix form: 
\begin{equation}\label{eq:this-shared-a}
   \isst{\BA} \cdot \BS + \BB \approx \Delta \cdot \BM \bmod q \enspace.
\end{equation}

Despite the similarity of this equation with~Equation~\eqref{eq:mat_row_format}, we stress that this is a  different type of encryption. In particular, in Equation~(\ref{eq:this-shared-a}), the matrix $\BA$ is structured whereas the secret matrix $\BS$ is arbitrary, and we use a column encoding of the matrix. In Equation~\eqref{eq:mat_row_format}, the roles of $\BA$ and~$\BS$ are swapped, and we use row encoding of the matrix. 

Note that the $a$-part can be shared even further across all ciphertext components encrypting the matrix~$\BM$. This results in a structure where $d_1\cdot d_2$ ciphertexts share a single common $a$ while utilizing $d_1\cdot d_2$ distinct secret keys. This yields yet another encryption structure:
\[
  a \cdot \sk_{ij} + b_{ij} \approx \sum_{0 \leq k<N} \Delta \cdot m_{Ni + k, j} X^k \bmod q, \ \ 0\le i < d_1/N, \ \ 0\le j < d_2 \enspace. 
\]
This corresponds, in matrix form, to $({\bf I}_{d_1/N} \otimes \toep{a}) \cdot \BS + \BB \approx \Delta \cdot \BM\bmod q$. In this case, the matrix $\BS \in \ZZ^{d_1\times d_2}$ is defined by $\BS_{Ni+k,j} = (\sk_{ij})_k$, the $k$-th coefficient of the secret key $\sk_{ij}$, and
$\isst{\BA} = {\bf I}_{d_1/N} \otimes \toep{a} \in \ZZ^{d_1\times d_1}$. 
This situation is similar to the shared-$s$-structured-$s$ case, with the roles and structures of $\BA$ and $\BS$ being exchanged.  

\subsubsection{\MLWE\ matrix format with structured $\BA$.}
Similarly, in the case $d_1 | N$ we can assume that the columns of~$\BM$ share their $a$-parts; again, this formulation corresponds to swapping the roles of $a$ and~$s$ in the \MLWE\ format. It leads to an encryption identity in matrix form
\[
\isst{\BA} \cdot \BS_{\textrm{M-sh-}a-\MLWE} + \BB \approx \Delta \cdot \BM \bmod q \enspace,
\]
where $\isst{\BA} \in \ZZ_q^{d_1\times N}$, $\BS_{\textrm{M-sh-}a-\MLWE} \in \ZZ^{N\times d_2}$ and $\BB \in \ZZ_q^{d_1\times d_2}$. 

\medskip
Overall, we obtain the following lemma.
\begin{lemma}\label{le:mat_struct_a}
Let $\BM \in \RR^{d_1 \times d_2}$ be a matrix  with $d_1 | N$ or $N | d_1$. We can write its structured-$\BA$ \{shared-$a$ \MLWE, \RLWE\} encryption in the following form: there exist a matrix $\BS \in \ZZ_{q}^{N\times d_2}$, a matrix $\isst{\BA} \in \ZZ_q^{d_1 \times N}$ and a matrix $\BB \in \ZZ_q^{d_1 \times d_2}$ such that  
\begin{equation}\label{eq:mat_lwe_structa}
\isst{\BA}\cdot \BS + \BB \approx \Delta \cdot \BM \bmod q\enspace.
\end{equation}
\end{lemma}
As in Lemmas~\ref{le:vec_encodings} and \ref{le:mat_struct_s}, Equation~\eqref{eq:mat_lwe_structa} characterizes the structured-$\BA$ \MLWE\ format, as soon as the matrix~$\isst{\BA}$ has the relevant structure.

\subsection{Matrix \RGSW\ encryption}\label{sec:gsw_format}
We now extend the \RGSW\ encryption for ciphertexts to the matrix setting. Let $p, q$ be two positive integers with~$p\cdot\Delta \ge q$.

A \RGSW\ column encryption of a plaintext matrix $\lfloor \Delta \BM \rceil \in \ZZ^{d_1\times d_2}$ with underlying encryption being  \RLWE, \MLWE\ or shared-$a$, and auxiliary modulus~$p$, is a pair $(\BA_0, \BB_0), (\BA_1, \BB_1) \in \ZZ_{pq}^{N\times d_2} \times \ZZ_{pq}^{d_1 \times d_2}$ such that 
\begin{align*}
p\cdot (\Delta \cdot \BM) \approx \isst{\BS} \cdot \BA_0 + \BB_0 \bmod pq \enspace, \\
p\cdot (\Delta \cdot \BM)\cdot {\isst{\BS'}} \approx \isst{\BS} \cdot \BA_1 + \BB_1 \bmod pq \enspace. 
\end{align*}

A specificity of the \RGSW\ format is to enable two ring degrees and two matrix secret keys: 
\begin{itemize}
\item[$\bullet$] the first ring degree~$N$ corresponds to the encryption of the matrix~$\BM$, with $\isst{\BS} \in \ZZ_q^{d_1\times N}$ being the associated matrix secret key,
\item[$\bullet$] the second ring degree~$N'$ corresponds to the encryption of the vector~$\bv$ by which~$\BM$ is to be multiplied, with ${\isst{\BS'}} \in \ZZ_q^{d_2\times N'}$ being the associated matrix secret key. 
\end{itemize}

The structures of $\isst{\BS}$ and $\isst{\BS'}$ are determined by the encryption formats chosen for $\BM$ and $\bv$. In the general case, we can thus define $\GGSW_{d,d'}\ $ to be the encryption format for which:
\begin{itemize}
\item[$\bullet$] $\BM$ is RLWE-encrypted ($d = N$), or MLWE-encrypted $(d <N)$, or shared-$a$ RLWE-encrypted ($d \ge N$);
\item[$\bullet$] $\bv$ is RLWE-encrypted ($d' = N'$), or MLWE-encrypted $(d' <N')$, or shared-$a$ RLWE-encrypted ($d' \ge N'$).
\end{itemize}

Similarly to the \RGSW\ format for ciphertexts, it is possible to give a more general definition involving gadget decomposition, notably to handle the case where $p \ll q$; we shall leave this extension to the reader. 

\subsection{Which format for which algorithm?}

We summarize the various formats described in this section in Table~\ref{tbl:formats} and point to the algorithms of Section~\ref{sec:algorithms} in which they are used. There are two families of formats, the structured-$\BS$ formats (Lemmas~\ref{le:vec_encodings} and~\ref{le:mat_struct_s}) and the structured-$\BA$ formats (Lemma~\ref{le:mat_struct_a}). They are used in different algorithms, with structured-$\BA$ formats being specifically used for the $\cpmm$ algorithm with precomputation (Algorithm~\ref{alg:pcmm-with-precomp}), where the plaintext matrix is known beforehand.

Among the structured-$\BS$ formats, the dimensions of the matrix or of the vector dictate which format should be used, except for the case of shared-$s$ vs shared-$a$. In this case, shared-$s$ gives an algorithm with cheaper pre- and post-processing, while shared-$a$ gives an algorithm with a better reduction to plaintext linear algebra; overall, shared-$s$ is preferable for $\cpmm$ with dimensions $d_1 \times d_2 \times d_3$ when $d_3$ is small, the exact cutoff point being implementation dependent.

Regarding the choice between the \RLWE\ and the \RGSW\ formats, we observe that the \RGSW\ is more flexible in terms of dimensions, but may require pre-processing to compute the \RGSW\ encryption. When both approaches can be used, however, the \RLWE\ approach gives a better reduction to plaintext linear algebra. 

\begin{table}[H]
\centering
\caption{Overview of matrix/vector formats. The second column indicates whether we discuss a vector or a matrix format, while the third column indicates the conditions on dimensions for which this format is used. The fourth column provides the lemmas in which the formats are discussed, whereas the last column gives the algorithms in which these formats are used. 
We assume that matrices have dimensions  $d_1 \times d_2$ and vectors have dimension $d_1$. The RLWE ring-degree is denoted by~$N$. }
\label{tbl:formats}
\begin{tabular}{|c|c|c|c|c|}
\hline
Format & M/v & dimensions & Statement & Algorithms\\
\hline
\RLWE & M, v & $d_1 = N$ & Lemmas~\ref{le:vec_encodings} \& \ref{le:mat_struct_s} & Algorithms~\ref{alg:pcmm}, \ref{alg:ccmm} \& \ref{alg:general}\\ \hline
shared-$a$ & M, v & $N|d_1$ & Lemmas~\ref{le:vec_encodings} \& \ref{le:mat_struct_s} & Algorithms~\ref{alg:pcmm} \& \ref{alg:general}\\ \hline
shared-$s$ & M, v & $N | d_1$ & Lemmas~\ref{le:vec_encodings} \& \ref{le:mat_struct_s} & Algorithms~\ref{alg:pcmm} \& \ref{alg:general} (modified)\\ \hline
\MLWE & M, v & $d_1 | N$ & Lemmas~\ref{le:vec_encodings} \& \ref{le:mat_struct_s} & Algorithms~\ref{alg:pcmm} \& \ref{alg:general}
\\ \hline \hline
structured-$\BA$ shared-$a$ & M & $N | d_1$ & Lemma~\ref{le:mat_struct_a} & Algorithm~\ref{alg:pcmm-with-precomp} \\ \hline
structured-$\BA$ shared-$a$ MLWE & M & $d_1 | N$ & Lemma~\ref{le:mat_struct_a} & Algorithm~\ref{alg:pcmm-with-precomp}\\
\hline 
\end{tabular}
\end{table}

\section{Format conversions}\label{sec:conversions}

Our reductions will make use of the formats described in the previous section. In practice however, in many homomorphic computations (typically, in AI applications where the standard neural network evaluation alternates linear and non-linear steps), linear algebra occurs intertwined with a number of other steps during a much larger computation, and it is not realistic to assume that the packing of the data and encryption structure can be set based upon the requirements of the linear algebra steps only. In such a situation, one thus needs to apply a format conversion algorithm. 

We choose column-packed shared-$s$ \RLWE\ encryption as a base point. We shall give algorithms which allow one to convert other types of encodings to and from this one. 
A large fraction of the results described in this section can be summarized in the following lemma.
\begin{lemma}
Let $d, d' \ge 1$ be two integers such that $N|d$ or $d|N$.

Regarding format conversion for a matrix $\BM \in \RR^{d\times d'}$,
for each matrix encryption format $F$ among structured-$s$ shared-$a$ \RLWE\ (if $N|d$), structured-$\BA$ shared-$a$ \RLWE\ (if $N|d$), \MLWE\ (if $d|N$) and structured-$\BA$ \MLWE\ (if $d|N$), there are two algorithms converting from (shared-$s$) \RLWE\ to $F$ and back. The costs of these  algorithms can be  found in Table~\ref{tab:conversion_costs}.
\end{lemma}

The costs for converting from and to MLWE can be found in Sections~\ref{sec:moddecomp},~\ref{sec:modpack}. For conversion to and from structured-$\BS$ format, we need to convert all the ciphertexts encrypting a given column to and from the shared-$a$ format; there are $d'$ such columns, each made of $d/N$ ciphertexts. For the structured-$\BA$ case, the situation is symmetrical: we need to convert a given chunk of size $N$ (corresponding to a ciphertext) of all the columns to and from the shared-$a$ format; there are $d/N$ such chunks, and $d'$ columns. 

\begin{table}
\caption{\label{tab:conversion_costs}Conversion costs between encryption formats for a $d\times d'$ matrix, expressed in numbers of operations in~$\ZZ_{pq}$, where~$q$ is the ciphertext modulus and~$p$ is an auxiliary modulus that may be chosen as~$\approx q$. The \textrm{St.-}\BS \textrm{ sh.-}a and  \MLWE\  columns also address the vector case by taking $d' = 1$ (and the formats of the two other columns are not applicable to vectors).}
\[
\begin{array}{c|c|c|c|c|}
\cline{2-5}
&
\multicolumn{2}{c|}{\textrm{Structured-$\BS$ formats}}& 
\multicolumn{2}{c|}{\textrm{Structured-$\BA$ formats}}
\\
\cline{2-5}
& \textrm{St.-}\BS \textrm{ sh.-}a\ \ (\mbox{for }N|d)& \textrm{St.-}\BS\ \MLWE\ \ (\mbox{for }d|N)& \textrm{St.-}\BA \textrm{ sh.-}a \ \  (\mbox{for }N|d)& \textrm{St.-}\BA\ \MLWE\ \ (\mbox{for }d|N)\\
\hline
\multicolumn{1}{|c|}{\multirow{2}{*}{\textrm{ sh.-}$s \rightarrow$}} & O(d d'\log N\log (d/N))  & O(dd') & O(d d' \log d' \log N) & O(d d') \\
\multicolumn{1}{|c|}{} &  \textrm{Sec.}~\ref{sec:fastforwardswitch} 
& \textrm{Sec.}~\ref{sec:moddecomp} &
\textrm{Sec.}~\ref{sec:fastforwardswitch} & \textrm{Sec.~\ref{sec:moddecomp}}
\\
\hline 
\multicolumn{1}{|c|}{\multirow{2}{*}{$\rightarrow \textrm{sh.-}s$}} & O(d d' \log N)&O(\max(d', N/d) N\log N)~  & O(d d'\log N)  & O(\max(d', N/d) N\log N)~\\
\multicolumn{1}{|c|}{} &\textrm{Sec.}~\ref{sec:backfmtswitch} & \textrm{Sec.}~\ref{sec:modpack} & \textrm{Sec.}~\ref{sec:backfmtswitch}  & \textrm{Sec.}~\ref{sec:modpack}
\\
\hline
\end{array}
\]
\end{table}
These conversion algorithms can also handle \RGSW\ encryptions, by viewing them as  pairs of \RLWE\ encryptions. We also study the conversion from row encodings to column encodings, namely transposition, in the case of square matrices, and describe an algorithm with quasi-linear complexity in the case of \RLWE\ encoding. From this algorithm, we derive a conversion from a \RLWE\ encryption modulo~$pq$ into a matrix \RGSW\ encryption in the case of square matrices. 

We note that the conversion algorithms described in this section do not consume any multiplicative level, so they do not affect the FHE parameters, usability, and efficiency when used in combination with the linear algebra algorithms from Section~\ref{sec:algorithms}.

\subsection{From shared-$s$ to shared-$a$}
\label{sse:conversion1}
\label{sec:fmtswitch}
Assume that we have~$n \geq 1$ ciphertexts $\ct_i= (a_i,b_i) \in \R_q^2$ in shared-$s$ format, i.e., such that~$a_i\cdot \sk +b_i \approx m_i \bmod q$ for some shared secret key~$\sk$ and plaintexts~$m_i$ (for all~$0 \le  i < n$). We aim at converting them to~$n$ ciphertexts $\ct'_i= (a',b'_i) \in \R_q^2$ such that~$a'\cdot \sk_i' + b'_i \approx m_i \bmod q$ (for all~$0 \le  i < n$) that share their $a$-parts.  In this context, the secret key~$\sk$ and the row vector of secret keys $(\sk'_0, \ldots, \sk'_{n-1})$ are all sampled during key generation, and 
we assume that the key generation algorithm also outputs a format-switching key: 
\[
\fmtswk_{q,p, \sk \rightarrow \{\sk'_i\}_{0\le i< n}} 
= ({\bf a}_{\fmtswk}, {\bf B}_{\fmtswk})  \in \R_{pq}^n \times \R_{pq}^{n \times n}  \enspace ,
\] 
for some auxiliary integer~$p$ with  ${\bf a}_{\fmtswk}$ sampled uniformly in~$\R_{pq}^n$, and
\begin{equation}
\label{eq:fmtswk}
 {\bf B}_{\fmtswk} 
= 
-{\bf a}_{\fmtswk}
\cdot 
\left(\sk'_0, \ldots, \sk'_{n-1} \right) 
+
{\bf E}
+
p \cdot \sk \cdot {\bf I}_n \bmod pq \enspace,
\end{equation}
where~${\bf E} \in \R^{n \times n}$ has small-magnitude coefficients, and ${\bf a}_{\fmtswk}
\cdot \left(\sk'_0, \ldots, \sk'_{n-1} \right)$ is the multiplication of a column vector and a row vector, resulting in an $n\times n$ matrix.
Note that ${\bf a}_{\fmtswk}$ can be the output of an extendable output function evaluated on a seed, to reduce memory consumption. Nevertheless, the matrix~${\bf B}_{\fmtswk}$ still contains~$n^2$ elements of~$\R_{pq}$.
By a standard hybrid argument (see, e.g.,~\cite[Lemma~6.2]{PW08}),  the format conversion key~$\fmtswk_{q,p, \sk \rightarrow \{\sk'_i\}_{0\le i< n}}$ can be shown to be computationally indistinguishable from uniform.  

Algorithm~\ref{alg:fmtswitch} is the format switching algorithm.
We note that the $i$-th column of $\BB_{\fmtswk}$ contains $b$-parts of switching keys from $\sk$ to~$\sk'_i$, while the $j$-th row of~$\BB_{\fmtswk}$ can be seen as a set of $b$-parts for ciphertexts in shared-$a$ format. In both cases, the corresponding $a$-part is found in the vector~${\bf a}_{\fmtswk}$. Computing the product $(a_0, \ldots, a_{n-1}) \cdot   {\bf B}_{\fmtswk}$ preserves those properties and yields a vector of $b$-parts of switching keys from~$\sk$ to~$\sk'_i$ under the common $a$-part $(a_0, \ldots, a_{n-1}) \cdot {\bf a}_{\fmtswk}$. The rest of the algorithm  completes the actual key-switching.

\begin{algorithm}
    \caption{\label{alg:fmtswitch}Shared-$s$ to shared-$a$ conversion.}
\begin{algorithmic}[1]
\REQUIRE Format-switching key $\fmtswk_{q,p, \sk \rightarrow \{\sk'_i\}_{0\le i< n}} = ({\bf a}_{\fmtswk}, {\bf B}_{\fmtswk})  \in \R_{pq}^n \times \R_{pq}^{n \times n}$.
\REQUIRE
Ciphertexts $((a_i, b_i))_{0\le i < n}$ such that $a_i \cdot \sk + b_i \approx m_i \bmod q$  for all $0\le i < n$.
\ENSURE
A polynomial $a'$ and polynomials $( b'_i)_{0\le i < n}$ 
such that  
  $a' \cdot \sk'_i + b'_i \approx m_i \bmod q$ for all $0\le i < n$.
\STATE $a'\leftarrow \left\lfloor \frac{1}{p} \Big( \langle (a_0, \ldots, a_{n-1})^t , {\bf a}_{\fmtswk} \rangle \bmod pq \Big) \right\rceil  \ $;
\STATE $(c_0, \ldots, c_{n-1}) \leftarrow \left\lfloor \frac{1}{p} \Big( (a_0, \ldots, a_{n-1}) \cdot   {\bf B}_{\fmtswk} \bmod pq \Big) \right\rceil \ $;
\FOR{$i = 0$ to $n-1$}
\STATE $b'_i \leftarrow b_i + c_i \bmod q$\ ;
\ENDFOR
\RETURN $a', \{ b'_i\}_{0\le i < n}$\ .
\end{algorithmic}
\end{algorithm}

\begin{lemma}\label{lem:shared-a}
Assuming~$p \geq q$, Algorithm~\ref{alg:fmtswitch} is correct. Further, it consumes $\widetilde{O}(n^2 N)$ arithmetic operations in~$\ZZ_{pq}$.
\end{lemma}

\begin{proof}
Let~$0\le i< n$. There exist~$e_a,e_b \in \mathbb{R}[X]/(X^N+1)$ with coefficients in $(-1/2,1/2]$ and~$k_a,k_b \in \R$ such that, over~$\mathbb{R}[X]/(X^N+1)$:
\begin{eqnarray*}
a' \cdot \sk'_i & = & \frac{1}{p}  \langle (a_0,\ldots,a_{n-1})^t ,  {\bf a}_{\fmtswk} \rangle \cdot \sk'_i + e_a \cdot \sk'_i + k_aq \enspace , \\  
b_i' & = & b_i + \frac{1}{p}  (a_0,\ldots,a_{n-1}) \cdot {\bf B}_{\fmtswk,i} + e_b  + k_bq \enspace ,
 \end{eqnarray*}
where ${\bf B}_{\fmtswk,i}$ refers to the $i$-th column of~${\bf B}_{\fmtswk}$. Taking the $i$-th column of Equation~\eqref{eq:fmtswk}, we obtain that there exists~$k_i \in \R$ such that the following holds over~$\R$:
\[
{\bf a}_{\fmtswk} \cdot \sk'_i + {\bf B}_{\fmtswk,i} = {\bf e}_i + p \cdot (0,\ldots,0,\sk, 0, \ldots, 0)^t + k_i pq \enspace , 
\]
where~${\bf e}_i$ refers to the $i$-th column of~${\bf E}$ and~$\sk$ is in the $i$-th coefficient of $(0,\ldots,0,\sk, 0, \ldots, 0)^t$. 
Combining the last three equations, we obtain:
\begin{eqnarray*}
a' \cdot \sk'_i + b'_i & = & a_i \cdot \sk + b_i  + (k_a+k_b+k_i) q  \\ && + \left( \frac{1}{p} \langle (a_0,\ldots,a_{n-1})^t , {\bf e}_i \rangle   + e_a \cdot \sk_i' + e_b\right) \\
& \approx & m_i + (k_a+k_b+k_i) q \enspace .
\end{eqnarray*}
Here, we used the fact that the $a_i$'s have coefficients in $(-q/2,q/2]$ and~$p\geq q$, implying that the ring element $\frac{1}{p}  \langle (a_0,\ldots,a_{n-1})^t, {\bf e}_i \rangle$ has small-magnitude coefficients. Reducing modulo~$q$ gives the result.

The cost of this algorithm is dominated by the matrix-vector product of Step~2, of dimension $1\times n \times n$ over $\R_{pq}$. It thus implies $O(n^2N\log N)$ additions and multiplications in $\ZZ_{pq}$.
\qed
\end{proof}

\subsection{Improved shared-$s$ to shared-$a$ conversion}
\label{sec:fastforwardswitch}
A drawback of the format conversion approach described above is the quadratic growth of the size of the format conversion key and the quadratic number of arithmetic operations in $\ZZ_{pq}$, as a function of~the number of ciphertexts. We propose an alternative recursive approach for $n$ ciphertexts: first, each pair of ciphertexts is transformed so that they share their $a$-parts (i.e., there are only~$n/2$ distinct $a$-parts); then each pair of pair of ciphertexts is transformed so that they share their $a$-parts (i.e., there are only $n/4$ distinct $a$-parts), etc. This is formalized in Algorithm~\ref{alg:fast_fmtswitch}.

For the sake of simplicity, assume that~$n=2^\kappa$ for some integer~$\kappa \geq 1$. 
We define~$\bsk^{(0)}=\sk$ and~$\bsk^{(\kappa)} = (\sk'_0,\ldots,\sk'_{2^\kappa-1})$ (which is a row vector). For~$\ell \in [1,\kappa-1]$, we sample secret keys~$\sk^{(\ell)}_0,\ldots,\sk^{(\ell)}_{2^\ell-1}$ and define~$\bsk^{(\ell)}$ as their concatenation (we could alternatively define~$\bsk^{(\ell)}$ as the first $2^\ell$ elements of~$\bsk^{(\kappa)}$). We define the recursive format conversion key as follows: 
\[
\widetilde{\fmtswk}_{q,p, \sk \rightarrow \{\sk'_i\}_{0\le i < n}}
= \{({\bf a}^{(\ell)}_{\fmtswk}, {\bf B}^{(\ell)}_{\fmtswk})\}_{\ell \in [1,\kappa]} \enspace ,
\] 
for some auxiliary integer~$p$, with  ${\bf a}^{(\ell)}_{\fmtswk}$ sampled uniformly in~$\R_{pq}^2$ (and possibly seeded), and
\begin{eqnarray*}
{\bf B}^{(\ell)}_{\fmtswk}
 & = & 
- {\bf a}^{(\ell)}_{\fmtswk}
\cdot 
\bsk^{(\ell)}
+ {\bf E}^{(\ell)} \\ 
& & +
p \cdot \left( \begin{array}{cccccc}
\sk_0^{(\ell-1)} & \ldots & \sk_{2^{\ell-1}-1}^{(\ell-1)} &&& \\
&&&  \sk_0^{(\ell-1)} & \ldots & \sk_{2^{\ell-1}-1}^{(\ell-1)} 
\end{array} \right)
  \in \R_{pq}^{2 \times 2^{\ell}} \enspace ,
\end{eqnarray*}
where~${\bf E}^{(\ell)}$ has small-magnitude coefficients.

\begin{algorithm}
    \caption{\label{alg:fast_fmtswitch}Fast shared-$s$ to shared-$a$ conversion.}
\begin{algorithmic}[1]
\REQUIRE Recursive format conversion key 
$\widetilde{\fmtswk}_{q,p, \sk \rightarrow \{\sk'_i\}_{0\le i < n}}$, where $n$ is a power of two. 

\REQUIRE
Ciphertexts $((a_i, b_i))_{0\le i < n}$.

\ENSURE
A polynomial $a'$ and polynomials $( b'_i)_{0\le i < n}$ such that
  $a' \cdot \sk'_i + b'_i = m_i \bmod q$ for all $0\le i < n$.
\FOR{$i$ from $0$ to $n-1$}
\STATE $(a_i^{(0)}, b_i^{(0)}) \leftarrow (a_i, b_i)$;
\ENDFOR
\FOR{$\ell$ from $1$ to $\log n$}{}
\FOR{$i$ from $0$ to $n/2^\ell - 1$}
\STATE
$a^{(\ell)}_{i}  \leftarrow\left\lfloor 
\frac{1}{p} \Big( \langle (a^{(\ell-1)}_{2i}, a^{(\ell-1)}_{2i+1})^t ,  {\bf a}^{(\ell)}_{\fmtswk} \rangle \bmod pq \Big) \right\rceil  \ $ ;
\STATE 
$(b^{(\ell)}_{i 2^{\ell}}, \ldots, b^{(\ell)}_{(i+1) 2^{\ell}-1}) \leftarrow  
(b^{(\ell-1)}_{i 2^{\ell}}, \ldots, b^{(\ell-1)}_{(i+1) 2^{\ell}-1}) 
 +\left\lfloor 
\frac{1}{p} \Big( (a^{(\ell-1)}_{2i}, a^{(\ell-1)}_{2i+1}) \cdot   {\bf B}_{\fmtswk}^{(\ell)} \bmod pq \Big) \right\rceil \ $;

\ENDFOR
\ENDFOR
\RETURN $a_0^{(\log n)}, ({b_i}^{(\log n)})_{0\le i < n}$\ .
\end{algorithmic}
\end{algorithm}

\begin{theorem}
Assuming~$p \geq q$, Algorithm~\ref{alg:fast_fmtswitch} is correct. Further, it consumes $O(n  N (\log n) (\log N))$ arithmetic operations in~$\ZZ_{pq}$.
\end{theorem}
\begin{proof}
We prove by induction that, for all $\ell$, we have
\[
   a_i^{(\ell)} \cdot \sk_j^{(\ell)} + b_{i 2^{\ell}+j}^{(\ell)} \approx m_{i 2^{\ell}+j}^{\phantom{(\ell)}} \bmod q, \ \ 0\le i < n/2^{\ell}, \ \ 0 \le j  < 2^\ell \enspace.
\]
For $\ell = 0$, this follows from the definition of $(a_i, b_i)$. For $\ell = \log n$, this will yield correctness. 

Assume that it holds for $\ell - 1$. We compute, for $0\le i < n/2^\ell$, $0 \le j < 2^\ell$:
\begin{align*}
a_i^{(\ell)} \cdot \sk_j^{(\ell)} + b_{i2^\ell +j}^{(\ell)} & = 
\left\lfloor
\frac{1}{p} \Big( \langle (a^{(\ell-1)}_{2i}, a^{(\ell-1)}_{2i+1})^t ,  {\bf a}^{(\ell)}_{\fmtswk} \rangle \bmod pq \Big) \right\rceil
\cdot \sk_j^{(\ell)}  \\ 
& \hspace*{1cm} + b_{i2^{\ell}+j}^{(\ell -1 )} + 
\left\lfloor \frac{1}{p} \Big( (a^{(\ell-1)}_{2i}, a^{(\ell-1)}_{2i+1}) \cdot   {\bf B}_{\fmtswk}^{(\ell)} \bmod pq \Big)_j \right\rceil\enspace,
\end{align*}
where the last ``index~$j$'' indicates that we take the $j$-th entry of $(a^{(\ell-1)}_{2i}, a^{(\ell-1)}_{2i+1}) \cdot   {\bf B}_{\fmtswk}^{(\ell)} \bmod pq$.
The right hand side computation then corresponds to Algorithm~\ref{alg:fmtswitch}, with $n=2$ and the following ciphertexts as inputs:
\begin{eqnarray*}
(a_{2i}^{(\ell-1)}, b_{2i\cdot 2^{\ell-1}+j}^{(\ell-1)}) & \mbox{ and } &  
(a_{2i+1}^{(\ell-1)}, b_{(2i+1)\cdot 2^{\ell-1} + j}^{(\ell-1)})\enspace,
\end{eqnarray*}
for each $0 \le j < 2^{\ell-1}$ and $0 \le i < n/2^\ell$. 
We thus deduce, thanks to Lemma~\ref{lem:shared-a}, that for each $0 \le j < 2^{\ell-1}$ and $0 \le i < n/2^\ell$,
\begin{eqnarray*}
a_i^{(\ell)} \cdot \sk_{j}^{(\ell)} + b_{2i\cdot2^\ell +j}^{(\ell)} 
&~\approx~& a_{2i}^{(\ell-1)} \cdot \sk_{j}^{(\ell - 1)} + b_{2i\cdot2^\ell+j}^{(\ell-1)} \hspace{.65cm} \approx~m_{2i\cdot2^\ell+j}\enspace, \ \textrm{ and}\\
a_i^{(\ell)} \cdot \sk_{j+2^\ell}^{(\ell)} + b_{(2i+1)\cdot2^\ell +j}^{(\ell)} 
&~\approx~& a_{2i+1}^{(\ell-1)} \cdot \sk_{j}^{(\ell - 1)} + b_{(2i+1)\cdot2^\ell+j}^{(\ell-1)}~\approx~m_{(2i+1)\cdot2^\ell+j} \enspace,
\end{eqnarray*}
where the last relations hold thanks to the induction hypothesis. 
It proves that our hypothesis holds for~$\ell$.

For each $\ell$, Algorithm~\ref{alg:fast_fmtswitch} makes 
$n/2^{\ell}$ calls to  Algorithm~\ref{alg:fmtswitch} with two secrets for the $a$-part and $n$ calls to Algorithm~\ref{alg:fmtswitch} with two secrets for the $b$-part. The total cost is thus $O(n \log n)$ key-switchings. \qed
\end{proof}

For a moderate number of ciphertexts, this asymptotic improvement may not show in practice as this method involves more changes of modulus (from~$q$ to~$pq$ and backwards). Indeed, modulus changes can be costly in practice.

\subsection{Shared-$a$ to shared-$s$ conversion}
\label{sec:backfmtswitch}
Suppose some computations have been run on shared-$a$ ciphertexts, initially under keys~$(\sk_i')_{0\le i < n}$ and we want to convert back to ciphertexts for a shared secret~$\sk$. 
This may be implemented using the \KeySwitch\  algorithm for each of the ciphertexts.
We observe that this introduces a circular security assumption, as the format conversion key from~$\sk$ to the~$\sk_i'$'s contains
\RLWE\ encryptions of~$p \cdot \sk$ under the $\sk_i'$'s, and the backward format conversion keys contain \RLWE\ encryptions of~$p \cdot \sk_i'$'s under~$\sk$. The total cost of this conversion is $O(n)$ key switchings, so $\widetilde{O}(nN)$ arithmetic operations in~$\ZZ_{pq}$.

\subsection{From $\RLWE$ to $\MLWE$}\label{sec:moddecomp}
We consider  the following conversion from $\RLWE$ to $\MLWE$: 
\begin{itemize}
    \item[$\bullet$] \ModDecomp\label{not:moddecomp}. On input $\ct = (a, b) \in \R_{q,N}^2$, \ModDecomp\ returns  $({\ba}_i, b_i) \in \R_{q,N'}^{k+1}$ for $0\le i < k$, such that the underlying plaintexts correspond to the plaintext underlying~$(a,b)$. The cost of \ModDecomp\ is linear in the size of its input -- it is a reorganization of data. 
\end{itemize}
This is a reformulation of the identity $a\cdot \sk + b \approx m \bmod q$: decomposing the polynomials $a$, $\sk$, $b$ and~$m$ of  degrees smaller than~$N$  in the form $P = \sum_{0 \leq i <k} p_i(Y)\cdot X^{i}$, with $\deg p_i < N'=N/k$ and $Y=X^k$, and writing 
\[
a\cdot \sk + b \approx m \ \Longleftrightarrow 
   \sum_{0 \leq i <k} a_i(Y) \cdot \sk_{(j-i) \bmod k}(Y) \cdot Y^{e_{ij}} + b_j(Y) \approx m_j(Y)\bmod q, \ \ 0\le j < k \enspace, 
\]
where 
\[
e_{ij} = \left\{ \begin{array}{ll}
1& \textrm{if}\ i + j \ge k\enspace, \\
0& \textrm{otherwise}\enspace.\end{array}\right.
\]
This is an \MLWE\ encryption of~$m_j$. 
Depending on whether the term $Y^{e_{ij}}$ is merged with $a_i$ or $\sk_{(j-i) \bmod k}$, we shall obtain a   decomposition in $k$ ciphertexts with common $\sk$ part (shared-$s$) or common $a$ part (shared-$a$).

\subsection{From $\MLWE$ to $\RLWE$}\label{sec:modpack}
Going back from $k=N/N'$ $\MLWE_{q,N'}^{k}$ ciphertexts to one $\RLWE_{q,N}$ ciphertext is an instance of the  {\em ring-packing} problem. This problem has been extensively studied (see, e.g., \cite{CGGI17, BGGJ20, CDKS21, BCK+23}) in the context of packing multiple $\LWE$ ciphertexts into an $\RLWE$ ciphertext. 
The more general task of packing $\MLWE_{q,N'}$ ciphertexts into an $\MLWE_{q,N}$ ciphertext is addressed in~\cite{BCK+23}.

\begin{itemize}
    \item[$\bullet$] \ModPack\label{not:modpack}. On input $({\bf a}_i, b_i) \in \R_{q,N'}^{k+1}$ for $0\le i < \ell$ for a common key~$\bsk \in \R_{N'}^k$, \ModPack\ returns  $\lceil  \ell/k \rceil$ ciphertexts $(a^{(j)},b^{(j)}) \in \R_{q,N}^{2}$, $0\le j < \lceil \ell/k\rceil$, whose underlying plaintexts correspond to the tuple of plaintexts underlying the $({\bf a}_i, b_i)$'s.
    This requires switching keys to the modulus $pq$, for a suitable auxiliary modulus $p$. The cost of \ModPack\ is $O(k \cdot \lceil \ell/k\rceil)$ operations in $\R_{pq,N}$, so $O(k\cdot \lceil \ell/k\rceil \cdot N\log N)$ operations in $\ZZ_{pq}$.
\end{itemize}
We refer to~\cite{BCK+23} for details. We note that $k \cdot \lceil \ell/k \rceil = O(\max(k, \ell))$; in the sequel, we shall use this simplified alternative form.

\subsection{Ciphertext matrix transpose}\label{sec:transpose}
We propose a fast ciphertext matrix transposition (\cmt) algorithm for $N \times N$ matrices. 
The algorithm converts 
$N$ ciphertexts that encrypt each row of a given matrix 
to
$N$ ciphertexts that encrypt each column.
To be precise, for a matrix $\textbf{M} = (M_{i,j})_{0 \leq i,j <N}$,
the \cmt algorithm takes as inputs
$N$ ciphertexts $\{(a_i, b_i)\}_{0\le i <N}$ such that
$$a_i\cdot\sk + b_i ~\approx~\sum_{0\le j < N} \Delta \cdot M_{i,j} X^j \bmod q, \ \ 0\le i < N\enspace,$$
and returns $N$ ciphertexts $\{(a'_j, b'_j)\}_{0\le j < N}$ such that
$$a'_j\cdot\sk + b'_j ~\approx~\sum_{0\le i<N} \Delta \cdot M_{i,j} X^i\bmod q, \ \ 0\le j < N\enspace,$$
where $\sk$ is the secret key.
In matrix form, these identities translate as 
\begin{equation}\label{eq:cmt_matrixform}
\BA \cdot \toep{\sk} + \BB \approx \toep{\sk}^t \cdot \BA' + \BB' \enspace. 
\end{equation}

Below, we propose a \cmt algorithm with $\widetilde{O}(N^2)$ operations in~$\mathbb{Z}_q$. 
We remark that  a transposition requires  $\Omega(N^2)$ operations to read and write the data.
While our algorithm focuses on $N\times N$ matrices, it can be extended to larger matrices by considering $N \times N$ blocks and transposing them individually. 
In addition, we note one can generalize the proposed transpose algorithm for \MLWE\ and shared-$a$ \RLWE\ formats, 
to transpose row-wise encryptions of those formats for matrices of various dimensions.

\subsubsection{High-level description.}\label{ssec:transpose-motive}
We first describe the motivation of our algorithm with cleartext ring elements. 
The algorithm starts from the observation that the trace of a ring element $m(X)=\sum_{0\le i<N} m_iX^i \in\ring_q$ is, up to a multiplicative factor, the constant term~$m_0$. Formally, 
$$N\cdot m_0 = \trace{m(X)} = \sum_{0\le t < N} m(X^{2t+1})\enspace.$$
Similarly, for each $0\le j < N$, the trace of $X^{-j}\cdot m(X)$ gives $m_j$:
$$N\cdot  m_j = \trace{X^{-j}\cdot m(X)} = \sum_{0\le t <N} X^{-j(2t+1)}\cdot m(X^{2t+1})\enspace.$$

Assume now that we are given $N$ ring elements $m_0, \ldots, m_{N-1}$ in $\ring$: 
$$m_i(X)=\sum_{0\le j < N} {m_{i,j}X^j}\in\ring\enspace.$$ 
We aim to obtain $m'_j$ such that
$$m'_j(X)=\sum_{0\le i < N} {m_{i,j}X^i}\in\ring\enspace,$$ 
for each $0\le j < N$.
The above observation implies that
\begin{equation*}\begin{split}
N\cdot  m'_j(X)&=\sum_{0\le i< N} \trace{X^{-j}\cdot m_i}\cdot X^i \\ 
& = \sum_{0\le i < N}\sum_{0\le t< N} X^{-j(2t+1)+i}\cdot m_{i}(X^{2t+1})\\
         &=\sum_{0\le t<N} \left(\sum_{0\le i< N} X^{i} \cdot m_{i}(X^{2t+1})\right) X^{-j(2t+1)},~~~0\le j < N \enspace.
\end{split}\end{equation*}
Importantly, the term $\sum_{0\le i < N} X^{i}\cdot  m_{i}(X^{2t+1})$ is independent from $j$ for all $0\le t <N$. 

Based on this discussion, the outline of the  $\cmt$ algorithm is as follows:
\begin{enumerate}
    \item Compute $\{\widetilde{m}_t=\sum_{0\le i< N} m_i(X)\cdot X^{i\cdot (2t+1)^{-1}} \}_{0\le t <N}$ from $\{m_i\}_{0\le i < N}$;
    \item Compute $\{\overline{m}_t=\widetilde{m}_t(X^{2t+1})\}_{0\le t <N}$ from $\{\widetilde{m}_t\}_{0\le t < N}$;
     \item Compute $\{m'_j=\sum_{0\le t<N} {\bar{m}_t(X)}\cdot X^{-j(2t+1)}\}_{0\le j < N}$ from $\{\bar{m}_t\}_{0\le t < N}$. 
\end{enumerate}
In the first step, the term $(2t+1)^{-1}$ refers to the inverse of~$2t+1$ modulo~$2N$. 
We can homomorphically perform all the above steps in encrypted state by using key switching and arithmetic operations of FHE schemes.
The homomorphic algorithm contains only $N$ key switchings (Step~2), but requires $N^2$ ring additions during the first and third steps, which consume $O(N^3)$ operations in~$\mathbb{Z}_q$. 

To avoid this high cost, we devise
\tweak, an algorithm that can perform the first and third steps using~$\widetilde{O}(N^2)$ operations instead of~$O(N^3)$. 
It relies on the observation  that the involved ring additions  are structured and have the form 
$$\sum_{0\le i<N} {f_i(X)}\cdot X^{2ij},~~~0\le j <N\enspace,$$
for some $f_i(X)$'s, $0 \leq i < N$. 
In particular, the first and third steps yield $\sum_{i} {(X^i\cdot m_i(X))\cdot X^{2i\frac{(2t+1)^{-1}-1}{2}}}$ and $\sum_{t} {(X^{-j}\cdot \bar{m}_t(X))\cdot X^{2t(-j)}}$, respectively.

\subsubsection{The \tweak algorithm.}\label{ssec:tweak}
Algorithm~\ref{alg:tweak} takes as inputs $n$ ciphertexts $(\ctvec_i)_{0\le i < n}$ and outputs the following $n$ ciphertexts:
$$\left\{ \sum_{0\le i< n} X^{2ij \frac{N}{n}} \cdot \ctvec_i\right\}_{0\le j < n} \enspace.$$ 

To compute the above sums, it iteratively updates $\{\ctvec_{j}'\}_j$ to store the relevant sums in Equation~\eqref{eq:ct_prime_j},
then $\{\ctvec_0,\ldots,\ctvec_3\}$, and so on, until obtaining $\{\ctvec_{j}'\}_{0\le j<n}$. 
To compute $\{\ctvec_{j}'\}_{0\le j<2^\ell}$ from $\{\ctvec_{j}'\}_{0\le j<2^{\ell-1}}$, it solves and applies the result of a subproblem on $\{\ctvec_{i}\}_{0\le i<2^\ell}$. 
After repeating this process  $\log{n}$ times,  the set~$\{\ctvec_j'\}_j$ contains the desired output. 
Algorithm~\ref{alg:tweak} describes the process in detail.

\begin{algorithm}[H]
\caption{\tweak}\label{alg:tweak} 
\begin{algorithmic}[1]
\REQUIRE Ciphertexts $(\ctvec_i)_{0 \leq i < n}$, where~$n$ is a power of two.
\ENSURE Ciphertexts $(\ctvec'_j)_{0 \leq j < n}$ such that $\ctvec'_j = \sum_{0\le i<N} {X^{2ij\frac{N}{n}}\ctvec_i}$ for all $j$.
\IF {$n=1$}
    \RETURN $\ctvec_0$.
\ENDIF
\STATE $\ctvec'_0=\ctvec_0$;
\FOR {$\ell\gets 0~\TO~\log{n}-1$}
    \STATE $\auxvec \gets \tweak\left((\ctvec_{(2j+1)n/2^{\ell+1}})_{0 \leq j<2^\ell}\right)$\ ;
    \FOR {$j\gets 0~\TO~2^\ell -1$}
        \STATE $\ctvec'_{j+2^\ell}  \gets \ctvec'_j - X^{\frac{N}{2^\ell}j} \cdot \auxvec_j$\ ;
        \STATE $\ctvec'_{j} \gets \ctvec'_j + X^{\frac{N}{2^\ell}j} \cdot \auxvec_j$\ ;
    \ENDFOR
\ENDFOR
\RETURN $(\ctvec'_j)_{0\le j <n}$.
\end{algorithmic}
\end{algorithm}

\begin{theorem}\label{thm:tweak-correctness}
    Algorithm~\ref{alg:tweak} is correct. Further, it consumes $\widetilde{O}(nN)$ arithmetic operations in $\mathbb{Z}_{q}$. 
\end{theorem}
\begin{proof}
    We prove the correctness by using induction on~$n$. When $n=1$, it is direct. Then, for each $n>1$, suppose that the algorithm is correct for any~$n_0<n$.  We prove  that the following holds for all~$\ell$:
    \begin{equation}
    \label{eq:ct_prime_j}
    \ctvec'_j = \sum_{0\le i < 2^{\ell+1}} X^{2ij\frac{N}{2^{\ell+1}}}\cdot\ctvec_{i\frac{n}{2^{\ell+1}}}, \ \ 0 \leq j < 2^{\ell+1} \enspace,
    \end{equation}  
    after the $\ell$-th  iteration of the loop consisting of Steps~$6$--$10$. 
    Before starting the loop, the claim holds with $\ell=-1$. 
    Suppose the claim holds after the $(\ell-1)$-th iteration. By the induction hypothesis, 
    \begin{equation*}\begin{split}
        \ctvec'_{j} &= \sum_{0\le i < 2^\ell} X^{2ij\frac{N}{2^\ell}}\cdot\ctvec_{i\frac{n}{2^\ell}} 
    + X^{\frac{N}{2^\ell}j} \left(\sum_{0\le i < 2^\ell}{X^{2ij\frac{N}{2^\ell}}\cdot \ctvec_{(2i+1)n/2^{\ell+1}}}\right)\\
    &=\sum_{0\le i < 2^\ell} X^{2(2i)j\frac{N}{2^{\ell+1}}}\cdot\ctvec_{2i\frac{n}{2^{\ell+1}}} 
    + \sum_{0\le i < 2^\ell}{X^{2(2i+1)j\frac{N}{2^{\ell+1}}}\cdot \ctvec_{(2i+1)n/2^{\ell+1}}}\\
    &=\sum_{0\le i < 2^{\ell+1}} X^{2ij\frac{N}{2^{\ell+1}}}\cdot\ctvec_{i\frac{n}{2^{\ell+1}}} \enspace, 
    \end{split}\end{equation*}
    for each $0 \leq j< 2^\ell$. Also, since $X^N=-1$ in $\ring$, we have
    \begin{equation*}\begin{split}
        \ctvec'_{j+2^\ell} &= \sum_{0\le i < 2^\ell} X^{2ij\frac{N}{2^\ell}}\cdot\ctvec_{i\frac{n}{2^\ell}} 
    - X^{\frac{N}{2^\ell}j} \left(\sum_{0\le i < 2^{\ell}}{X^{2ij\frac{N}{2^\ell}}\cdot \ctvec_{(2i+1)n/2^{\ell+1}}}\right)\\
    &=\sum_{0\le i < 2^{\ell}} X^{2(2i)j\frac{N}{2^{\ell+1}}}\cdot\ctvec_{2i\frac{n}{2^{\ell+1}}} 
    - \sum_{0\le i < 2^{\ell}}{X^{2(2i+1)j\frac{N}{2^{\ell+1}}}\cdot \ctvec_{(2i+1)n/2^{\ell+1}}}\\
    &=\sum_{0\le i < 2^{\ell+1}} (-1)^i X^{2ij\frac{N}{2^{\ell+1}}}\cdot\ctvec_{i\frac{n}{2^{\ell+1}}} \\ 
    &=\sum_{0\le i < 2^{\ell+1}} X^{2i(j+2^\ell)\frac{N}{2^{\ell+1}}}\cdot\ctvec_{i\frac{n}{2^{\ell+1}}}\enspace.
    \end{split}\end{equation*}
  Overall, this gives  Equation~\eqref{eq:ct_prime_j}. Taking~$\ell = \log n -1$ completes the proof of the induction step. This provides the correctness. 
  
  We now turn to the complexity analysis. 
Let the cost of \tweak on $n$ ciphertexts be denoted by $T(n)$. 
Steps~$6$--$10$ can be performed with $T(2^\ell)+4\cdot 2^\ell N$ operations over~$\ZZ_q$.
Therefore, for all power-of-two integers~$n$, we have:
$$T(n) \leq T(n/2)+T(n/4)+\ldots+T(1)+4(n-1)N\enspace.$$
Finally, using this inequality, we prove by induction on~$n$ that~$T(n) \leq 2nN \log n$. Note that $T(1)=0$, as the algorithm does not require any computation when~$n=1$. If $T(n)\le 2nN\log{n}$ for all $n<n_0$, then 
\begin{align*}
T(n_0)~&\le~\sum_{i=1}^{\log{n_0}}{2\frac{n_0}{2^i}N\log{\frac{n_0}{2^i}}} +4(n_0-1)N\\~&=~2n_0N\left(\log{n_0}\sum_{i=1}^{\log{n_0}}\frac{1}{2^i} -\sum_{i=1}^{\log{n_0}}\frac{i}{2^i}+2-\frac{2}{n_0}\right)
\\~&=~2n_0N\left(\log{n_0}\left(1-\frac{1}{n_0}\right)-\left(2-\frac{\log{n_0}+2}{n_0}\right)+2-\frac{2}{n_0}\right)\\~&=~2n_0N\log{n_0}\enspace.
\end{align*}
Therefore, we have $T(n)\le 2nN\log{n}$ for all power-of-two integers $n$. 
\qed
\end{proof}

\subsubsection{\cmt algorithm.}\label{ssec:transpose}
We now use  \tweak to build a  \cmt algorithm that consumes $\widetilde{O}(N^2)$ arithmetic operations.
In a nutshell,  Algorithm~\ref{alg:transpose} tweaks the input ciphertexts $\ctvec$ to obtain $\auxvec$ and tweaks them again to obtain the transposed ciphertexts~$\ctvec'$.

\begin{algorithm}[ht]
\caption{\transpose}\label{alg:transpose} 
\begin{algorithmic}[1]
\REQUIRE Ciphertexts $(\ctvec_i)_{0 \leq i < N}$ such that $\ctvec_i$ encrypts $\bm_i = (m_{i,0},\ldots,m_{i,N-1})$ in its coefficient for each $i$.
\REQUIRE Switching keys $\swk_{pq,q, \sk(X^{2t+1})\rightarrow \sk}$ for $1\le t < N$.
\ENSURE Ciphertexts $(\ctvec'_j)_{0 \leq j < N}$ such that $\ctvec'_j$ encrypts $\bm'_j = (m_{0,j},\ldots,m_{N-1,j})$
for each $j$.
\STATE $(\auxvec_j)_{0 \le j < N} \gets \tweak\left((X^i\cdot \ctvec_i)_{0\le i <N}\right)$\ ;
\FOR {$j\gets 0~\TO~N-1$}
    \STATE $\auxvec'_j \gets (N^{-1}\bmod{q}) \cdot \auxvec_{(-1+(2j+1)^{-1}\bmod{2N})/2}$\ ;
    \STATE $\auxvec'_j \gets \aut(\auxvec'_j, 2j+1; \swk_{pq,q, \sk(X^{2j+1})\rightarrow \sk})$ \ ;
\ENDFOR
\STATE $(\ctvec''_j)_{0 \le j < N} \gets \tweak\left((\auxvec_j')_{0 \le j < N}\right)$ \ ;
\FOR {$j\gets 1~\TO~N$}
    \STATE $\ctvec'_{j\bmod{N}} \gets -X^{N-j} \cdot\ctvec''_{N-j}$ \ ; 
\ENDFOR
\RETURN $(\ctvec'_j)_{0 \le j < N}$.
\end{algorithmic}
\end{algorithm}

At Step~3, the current ciphertext is multiplied by the constant $N^{-1}\bmod{q}$. 
As our algorithm involves traces over~$\ring_q$, it creates a scaling of the plaintext by a factor~$N$. Pre-multiplying by $N^{-1}\bmod{q}$ allows us to avoid this scaling. 
This requires that $N$ is invertible modulo~$q$, which is most often the case in CKKS, where~$q$ is typically chosen odd (and~$N$ is a power of two). 
This trick is borrowed from~\cite{CDKS21}.

\begin{theorem}\label{thm:transpose-correctness}
    Assume that~$q$ is coprime with~$N$. Then Algorithm~\ref{alg:transpose} is correct. Further, it consumes $\widetilde{O}(N^2)$ arithmetic operations in $\mathbb{Z}_{pq}$, where~$p \geq q$ is the auxiliary modulus used in \aut's key switchings.
\end{theorem}
\begin{proof}
From the correctness of Algorithm~\ref{alg:tweak} (Theorem~\ref{thm:tweak-correctness}), we have, after Step~$1$,
$$
    \auxvec_t = \sum_{0\le i<N}{X^{2it}\cdot(X^i \cdot\ctvec_i)} = \sum_{0\le i<N}{X^{i(2t+1)}\cdot\ctvec_i}, \ \ 0 \le t < N \enspace.
$$
After Steps $2$--$5$, we have, for each $0\le t <N$:
\begin{eqnarray*}
    N\cdot\auxvec'_t & = &  \aut\left(\sum_{0\le i<N}{X^{i(2t+1)^{-1}}\cdot\ctvec_i},~2t+1\right) \enspace,
\end{eqnarray*}
which decrypts approximately to 
$\sum_{0\le i<N} X^i \cdot \lfloor\Delta\cdot\bm_i(X^{2t+1})\rceil.$

Invoking Theorem~\ref{thm:tweak-correctness} again, we obtain that the following holds for all~$0 \le j < N$ after Step~$6$:
\begin{equation*}\begin{split}
    N\cdot\ctvec''_j  &= N\cdot\sum_{0\le t<N} (X^{2jt}\cdot\auxvec'_t) \enspace. 
\end{split}\end{equation*}
This decrypts approximately to
$\sum_{0\le t<N} \sum_{0\le i<N} X^{2jt}\cdot X^i\cdot \lfloor\Delta\cdot\bm_i(X^{2t+1})\rceil.
$
Finally, for each $0\le j <N$, we have
that $N\cdot\ctvec'_{(N-j)\bmod{N}}$ decrypts approximately to
\begin{equation*}\begin{split}
    -X^j \cdot\sum_{0\le t<N} \sum_{0\le i<N} X^{2jt}\cdot X^i \cdot\lfloor\Delta\cdot\bm_i(X^{2t+1})\rceil
        &= -\sum_{0\le t, i < N} X^{j(2t+1)}\cdot X^i \cdot\lfloor \Delta \cdot \bm_i(X^{2t+1})\rceil \\ &=-\sum_{0\le i<N} X^i \cdot\trace{X^j\cdot\lfloor \Delta\cdot\bm_i\rceil}\\
        &= -\sum_{0\le i<N} X^i \cdot\trace{(-\lfloor \Delta \cdot m_{i,N-j} \rceil,\ldots, \lfloor \Delta \cdot m_{i,N-j-1} \rceil)}\\
        &= N\cdot\sum_{0\le i<N} \lfloor \Delta \cdot m_{i, N-j}\rceil \cdot X^i \enspace.
\end{split}\end{equation*}
This completes the correctness proof.

We now turn to the complexity analysis. Algorithm~\ref{alg:transpose} consists of two calls to the \tweak algorithm and $N$ key switchings. 
Each call to \tweak consumes $\widetilde{O}(N^2)$ arithmetic operations over~$\ZZ_q$, by Theorem~\ref{thm:tweak-correctness}. Further, the cost of each key switching is $\widetilde{O}(N)$.
Consequently, the overall cost of Algorithm~\ref{alg:transpose} is $\widetilde{O}(N^2)$. \qed
\end{proof}

\subsubsection{Lightweight algorithm with small key size.}\label{sec:light}
In Algorithm~\ref{alg:transpose}, we use $N$ switching keys for $N$ automorphisms (Steps~$2$--$5$ in Algorithm~\ref{alg:transpose}). 
To reduce the key size, in the lightweight algorithm, we  update a single switching key before using it for the automorphism. 

The fundamental observation, inspired from~\cite{LLKN23}, is that the switching key is an encryption of $p\cdot\sk(X^{2t+1})$ for each $0\le t < N$, where $p\approx q$ is the auxiliary modulus used for key switchings. As $5$ and $-1$ generate the multiplicative group $(\ZZ/N\ZZ)^\star$, the set of all automorphism keys can be described as 
$$\{\enc{p\cdot\sk(X^{5^i})}\}_{0 \le i <N/2}\cup\{\enc{p\cdot\sk(X^{-5^i})}\}_{0 \le i <N/2} \enspace.$$
Consequently, we can generate another switching key from a given switching key by performing an automorphism to the given switching key. If $p'$ is an auxiliary modulus dedicated to key generation, then 
using only two master rotation keys ($\enc{pp'\cdot\sk(X^5)}$ and $\enc{pp'\cdot\sk(X^{-1})}$ and an initial switching key $\enc{p\cdot\sk(X)}$, then we can generate all switching keys. 

In Algorithm~\ref{alg:transpose}, we sequentially use the switching keys in the consisting of Steps~$2$--$5$. 
Using the above technique, we can use a single switching key by repeatedly updating it.
The switching key is initially set to $\enc{p\cdot\sk(X)}$. 
For the first $N/2$ iterations, we use the key and update it using the master rotation key $\enc{pp'\cdot\sk(X^5)}$. 
More precisely, for each $i$-th loop iteration, the switching key will be updated to $\enc{p\cdot\sk(X^{5^i})}$. 
After $N/2$ iterations, we update the switching key using the master conjugation key $\enc{pp'\cdot\sk(X^{-1})}$ and continue the last $N/2$ iterations. 
In particular, for each $(i+N/2)$-th loop iteration, the switching key will be updated to~$\enc{p\cdot\sk(X^{-5^i})}$. 

With the above strategy, we require only three switching keys. The computational cost increases by a constant factor, since $N$ additional key switchings have been introduced, but the asymptotic complexity remains~$\widetilde{O}(N^2)$.

\subsection{From \RLWE\ to \RGSW\ for square matrices}
\label{sec:lwetogsw}
As an application of the \cmt algorithm, we show how to compute a \RGSW\ encryption of $p\lfloor \Delta\cdot \BM\rceil$ when given an \RLWE\ encryption of $p \lfloor \Delta \cdot \BM \rceil $ modulo $pq$, for an $N \times N$ matrix~$\BM$. We recall from Section~\ref{sec:gsw_format} that the \RGSW\ format can handle two different secret keys $\isst{\BS}$ and ${\isst{\BS'}}$. We assume that they are defined with the same ring degree~$N$, and that ${\isst{\BS}} = \toep{\sk}$ and  ${\isst{\BS'}} = \toep{\sk'}$ for some $\sk, \sk'\in \ring_{N}$. 

From (\ref{eq:cmt_matrixform}), we see that using the \cmt algorithm on input $(\BA, {\bf 0})$ yields a pair $(\BA', \BB')$ such that
\[
{\isst{\BS'}}^t \cdot \BA' + \BB' \approx \BA \cdot {\isst{\BS'}}  \bmod pq \enspace.
\]
Similarly, from $(\BB, {\bf 0})$, we can use the \cmt algorithm to compute $(\BA'', \BB'')$ such that
\[
{\isst{\BS'}}^t \cdot \BA'' + \BB'' \approx \BB \cdot {\isst{\BS'}} \bmod pq \enspace. 
\]

We start from the \RLWE\ encryption of $p \Delta \cdot \BM$ modulo $pq$. This is the first part of the \RGSW\ encryption, namely:
\begin{equation}\label{eq:gsw_part1}
 p \cdot (\Delta \cdot \BM) \approx \isst{\BS} \cdot \BA + \BB \bmod pq \enspace. 
\end{equation}
We now explain how to get the second part of the  \RGSW\ encryption. We first derive
\begin{align*}
p \cdot (\Delta \cdot \BM) \cdot {\isst{\BS'}} & \approx \isst{\BS} \cdot \BA \cdot {\isst{\BS'}} + \BB \cdot {\isst{\BS'}} \bmod pq\enspace, \\
& \approx \isst{\BS}\cdot {\isst{\BS'}}^t \cdot \BA' + \isst{\BS}\cdot \BB' + {\isst{\BS'}}^t \cdot \BA'' + \BB'' \bmod pq\enspace. 
\end{align*}

We notice that ${\isst{\BS'}}^{t} = \toep{\sk'(X^{-1})}$ and $\isst{\BS}\cdot {\isst{\BS'}}^t = \toep{\sk\cdot \sk'(X^{-1})}$. Using a key-switching from $\sk \cdot \sk'(X^{-1})$ to $\sk$ on the columns of $(\BA', {\bf 0})$ and from $\sk'(X^{-1})$ to $\sk$ on the columns of $(\BA'', {\bf 0})$, we obtain
\begin{align*}
\isst{\BS}\cdot{\isst{\BS'}}^t \cdot \BA' & \approx \isst{\BS}\cdot \BA_0 + \BB_0 \bmod pq\enspace, \\
{\isst{\BS'}}^t\cdot \BA'' & \approx \isst{\BS}\cdot \BA_1 + \BB_1 \bmod pq\enspace, 
\end{align*}
leading to 
\[
p \cdot (\Delta \cdot \BM) \cdot {\isst{\BS'}} \approx \isst{\BS}\cdot \left(\BB' + \BA_0 + \BA_1 \right) + \BB'' + \BB_0 + \BB_1 \bmod pq\enspace,
\]
which is the second part of the \RGSW\ encryption of~$\BM$.

As we need to perform key-switching on ciphertexts modulo $pq$, we require switching keys defined modulo $p'pq$, where $p' \approx pq$ (note that this requirement can be reduced by using gadget decomposition). This eventually gives Algorithm~\ref{alg:rlwe_to_gsw}.

\begin{algorithm}
    \caption{\RLWE\ to \RGSW\ matrix conversion\label{alg:rlwe_to_gsw} }
\begin{algorithmic}[1]
\REQUIRE Switching keys $\swk_{p'pq,pq, \sk(X^{2t+1})\rightarrow \sk}$ for $1\le t < N$ required by Algorithm~\ref{alg:transpose}
\REQUIRE Keys $\swk_{pq,p'pq,\sk\cdot \sk'(X^{-1})\rightarrow \sk}$ and $\swk_{pq, p'pq,\sk'(X^{-1}) \rightarrow \sk}$. 
\REQUIRE An encryption $(\BA, \BB) \in \ZZ_{pq}^{N \times N} \times \ZZ_{pq}^{N\times N}$, under $\isst{\BS} = \toep{\sk}$, of the plaintext matrix $p\left\lfloor \Delta \cdot \BM \right\rceil \in \ZZ^{d_1\times d_2}$, such that Equation~\eqref{eq:gsw_part1} holds.
 \ENSURE A \RLWE\ encryption $(\widetilde{\BA}, \widetilde{\BB})$ of $p \left \lfloor \Delta \cdot \BM\right\rceil \isst{\BS'}$ under $\isst{\BS}$, where $\isst{\BS'} = \toep{\sk'}$. 
\STATE $(\BA', \BB') \leftarrow \transpose((\BA, {\bf 0}), (\swk_{p'pq,pq, \sk(X^{2t+1})\rightarrow \sk})_{1\le t < N})$;
\STATE $(\BA'', \BB'') \leftarrow \transpose((\BB, {\bf 0}), (\swk_{p'pq,pq, \sk(X^{2t+1})\rightarrow \sk})_{1\le t < N})$;
\STATE $(\BA_0, \BB_0) \leftarrow \textsf{\KeySwitch}_{\sk\cdot \sk'(X^{-1})\rightarrow \sk}(\BA', {\bf 0})$;
\STATE $(\BA_1, \BB_1) \leftarrow \textsf{\KeySwitch}_{\sk'(X^{-1})\rightarrow \sk}(\BA'', {\bf 0})$;
\RETURN $(\BA_0 + \BA_1 + \BB', \BB_0 + \BB_1 + \BB'')$.
\end{algorithmic}
\end{algorithm}

The correctness statement of the following theorem follows from the previous discussion. The cost estimate follows from inspection.  
\begin{theorem}
Assume that $pq$ is coprime with~$N$. Then Algorithm~\ref{alg:rlwe_to_gsw} is correct; in particular, it returns the second part of the \RGSW\ encryption of~$\BM$ under the secret key~$\isst{\BS}$. Furthermore, it consumes  $\widetilde{O}(N^2)$ arithmetic operations in~$\ZZ_{p'pq}$, where~$p' \geq pq$ is the auxiliary modulus used in~\KeySwitch.
\end{theorem}

Note that the cost of this conversion is too large for it to be used in the context of a ciphertext to cleartext reduction for Mv, as it would
not be negligible compared to the cost of~$\ppmv$. In the MM case with dimensions~$d_1 \times d_2 \times d_3$, the conversion cost is negligible even with moderate dimension~$d_3$. 
\section{MM and MV Algorithms}\label{sec:algorithms}

We now present the MM and MV algorithms. As encryption format conversion have been handled in Section~\ref{sec:conversions}, we shall assume, in each algorithm, that the input matrices / vectors are given in the most suitable format, and that the output matrix / vector is also returned in the most suitable format. This simplifies the description of the linear algebra algorithms and increases modularity. 
After each algorithm description, we discuss the costs of pre- and post-processing \emph{from and to the shared-$s$ format}, which is the most common one in the context of FHE. For a matrix-matrix product of dimension $d_1 \times d_2 \times d_3$, pre- and post-processing is negligible as soon as~$d_3$ and~$d_2$ are sufficiently large; we derive the corresponding bounds.  
When applicable, we shall also discuss pre-computations. 

For convenience, we extend $\Rescale$ and $\Relin$ to let them operate on matrices as follows: 
\begin{itemize}
\item[$\bullet$] $\Rescale(\BA, \BB)$ is defined as returning $(\lfloor \BA/\Delta \rceil, \lfloor \BB/\Delta \rceil),$ where $\Delta$ is the underlying scaling factor.
\item[$\bullet$] $\Relin(\BA, \BB)$ takes as inputs $\BA, \BB \in \ZZ_q^{N\times N}$ in \RLWE\ format and performs a relinearization on the ciphertexts corresponding to the columns of $\BA$ and~$\BB$.
\end{itemize}
Further, in all this section, we implicitly assume that the scaling factor $\Delta$ divides the ciphertext modulus~$q$. If it is not the case in practice, then one needs to choose~$q_1 \approx \Delta$ that divides~$q$ and replace the moduli~$q/\Delta$ by~$q/q_1$. 

We let the notation $\ModPPMM_{q; d_1, d_2, d_3}$ denote the multiplication of two plaintext matrices of respective dimensions $d_1 \times d_2$ and $d_2 \times d_3$, modulo the positive integer~$q$. Most of the time, the modulus~$q$ is clear from the context and will be omitted from the subscript. The case $d_3 = 1$ will be denoted by $\ModPPMv$.

\subsection{Reducing $\cpmm$ and $\cpmv$ to their PP counterparts}\label{sse:pcmm}
We start by observing that structured-$\BS$ matrix formats, as described in Lemma~\ref{le:mat_struct_s}, yield a $\cpmm$ algorithm. Indeed, from Equation~\eqref{eq:mat_lwe} and using the notations of the aforementioned lemma, we deduce, for~$\BU_0 \in \ZZ^{d_2\times d_3}$:
\[
\isst{\BS}\cdot \BA \cdot \BU_0 + \BB \cdot \BU_0 \approx \Delta \cdot \BM \cdot \BU_0 \bmod q \enspace.  
\]
This shows that $(\BA \cdot \BU_0, \BB \cdot \BU_0)$ is a  structured-$\BS$ matrix encryption of~$\BM \cdot \BU_0$, using the same encryption format as $\BM$. We thus obtain Algorithm~\ref{alg:pcmm}.

\begin{algorithm}
\caption{\label{alg:pcmm}$\cpmm$ algorithm for shared-$a$, \MLWE\ and \RLWE\ formats.}
\begin{algorithmic}[1]
\REQUIRE
    A matrix $\BU \in \RR^{d_2\times d_3}$ encoded as the plaintext matrix $\BU_0 = \lfloor \Delta \cdot \BU \rceil \in \ZZ_q^{d_2\times d_3}$ with scaling factor $\Delta$. \\      
\REQUIRE
    An encryption $(\BA, \BB) \in \ZZ_q^{N\times d_2} \times \ZZ_q^{d_1\times d_2}$ (under the matrix secret $\isst{\BS}$) of the plaintext matrix $\lfloor \Delta \cdot \BM \rceil \in \ZZ_q^{d_1\times d_2}$ encoding $\BM \in \RR^{d_1\times d_2}$, such that Equation~\eqref{eq:enc_mat} holds.
\ENSURE
    An encryption $(\BA', \BB') \in \ZZ_{q'}^{N\times d_3} \times \ZZ_{q'}^{d_1\times d_3}$ such that 
    \[ \Delta \cdot (\BM \cdot \BU) \approx \isst{\BS} \cdot \BA' + \BB' \bmod q' \enspace. \]

\STATE $(\BA', \BB')\gets (\BA \cdot \BU_0, \BB \cdot \BU_0)$;
\RETURN $\Rescale_{q\rightarrow q'}(\BA', \BB')$.
\end{algorithmic}
\end{algorithm}

\begin{theorem}\label{th:pcmm}
Let $d_1, d_2, d_3$ be positive integers with $d_1|N$ or $N|d_1$.
Algorithm~\ref{alg:pcmm} is correct and reduces $\CPMM_{d_1, d_2, d_3}$ in shared-$a$ (resp.\ \MLWE \ or \RLWE) encryption format to one call to $\ModPPMM_{N, d_2, d_3}$ and one call to $\ModPPMM_{d_1, d_2, d_3}$, plus $O(\max(d_1, N) d_3)$ arithmetic operations on integers whose absolute values are smaller than the ciphertext modulus. 
\end{theorem}

Correctness follows from the discussion above, while cost can be bounded by algorithm inspection.

The uniform character of this theorem hides an important point: the dominant cost, among the two plaintext matrix multiplications, may be either the first $\ModPPMM$ (in the \MLWE\ case where $d_1 | N$) or the second one (in the shared-$a$ case where $N | d_1$). In the second case, when $d_1 \rightarrow \infty$ for fixed $N$, we observe that the costs of $\CPMM_{d_1, d_2, d_3}$ and $\ModPPMM_{d_1, d_2, d_3}$ are asymptotically equal. 

In the shared-$s$ encryption format, Algorithm~\ref{alg:pcmm} remains valid up to the fact that $\BA \in \ZZ_q^{d_1\times d_2}$, $\BA' \in \ZZ_q^{d_1\times d_3}$. Theorem~\ref{th:pcmm} is then adapted as follows.  

\begin{theorem}\label{th:pcmm-shs}
Let $d_1, d_2, d_3$ be positive integers with $N|d_1$. Algorithm~\ref{alg:pcmm} modified as described above for the shared-$s$ encryption format is correct and reduces $\CPMM_{d_1, d_2, d_3}$ in the shared-$s$ format to two calls to $\ModPPMM_{d_1, d_2, d_3}$, plus $O(d_1 d_3)$ arithmetic operations on integers whose absolute values are smaller than the ciphertext modulus.
\end{theorem}

\subsubsection{Costs of format conversions.}
The pre-processing and post-processing costs for Algorithm~\ref{alg:pcmm}
can be found in~Table~\ref{tab:conversion_costs}, by taking $(d, d') = (d_1, d_2)$ (pre-processing) or $(d, d') = (d_1, d_3)$ (post-processing). As mentioned as the start of this section, we only consider pre-processing from (resp.\ post-post-processing to) the shared-$s$ format.
Assuming a plaintext linear algebra cost of $\Theta(d_1 d_2 d_3)$ for a product of matrices of dimensions $d_1\times d_2$ and~$d_2\times d_3$, these costs are negligible as soon as
\begin{itemize}
\item[$\bullet$] $d_3 = \omega(\log N \log (d_1/N))$ and $d_2 = \omega(\log N)$ when $N|d_1$ (large matrices, where we use the structured-$\BS$ shared-$a$ format); 
\item[$\bullet$] $d_3 = \omega(1)$ and $d_1 d_2 d_3 = \omega( \max(d_3, N/d_1) N\log N)$ when $d_1|N$ (small matrices, where we use the structured-$\BS$ $\MLWE$ format); if we assume $d_3 \ge N/d_1$, this condition takes the simpler form $d_1 d_2 = \omega(N \log N)$. 
\end{itemize}

In the case of the shared-$s$ variant of Algorithm~\ref{alg:pcmm}, there is no pre-processing nor post-processing cost, as the structured-$\BS$ shared-$s$ format that we consider to be the default one.
Both Algorithm~\ref{alg:pcmm} and its shared-$s$ version apply for any $d_3$ and, in particular, yield reductions from $\cpmv$ to $\ppmv$. However, in view of the cost discussion above, format conversion no longer has a negligible cost for very small~$d_3$. In such a situation, the shared-$s$ format, which incurs no pre- or post-processing costs, might be the best one. In practice, the best encryption format for a given instantiation of $d_1, d_2, d_3$ will be implementation-dependent. 

\subsection{CPMM and CPMv with precomputation}\label{sse:precomputation}
If the plaintext matrix (or vector) is known beforehand, we can take advantage of the structured-$\BA$ shared-$a$ format to obtain a reduction from $\cpmm$\ and $\cpmv$\ to a single instance of its plaintext counterpart in the same dimensions. Indeed, in that case, the encryption equation is 
\begin{equation}\label{eq:mat-shared-a}
\isst{\BA} \cdot \BS + \BB \approx \Delta\cdot \BM \bmod q \enspace, 
\end{equation}
with a structure assumption on $\isst{\BA}$, whereas $\BS$ is a general matrix.
Multiplying on the right by the cleartext matrix $\BU \in \ZZ_q^{d_2\times d_3}$ thus yields 
\[
\isst{\BA} \cdot (\BS  \cdot \BU) + \BB \cdot \BU \approx \Delta \cdot \BM \cdot \BU \bmod q \enspace. 
\]
This identity is a structured-$\BA$ shared-$a$ encryption of $\BM$ under the matrix secret key $\BS \cdot \BU$.
From a practical point of view, to use this identity, we assume that the matrix $\BU$ is known beforehand, and precompute the switching keys from the columns $\sk_i^\BU$ of $\BS\cdot \BU$ to a target secret key~$\sk$. 

The above discussion leads to Algorithm~\ref{alg:pcmm-with-precomp}.

\begin{algorithm}
\caption{\label{alg:pcmm-with-precomp} Precomputation $\cpmm$ algorithm, for structured-$\BA$ shared-$a$ and shared-$a$ \MLWE\  formats.}
\begin{algorithmic}[1]
\REQUIRE
    A matrix $\BU = \lfloor \Delta \textbf{U}_0 \rceil\in\ZZ_q^{d_2\times d_3}$, where $\BU_0\in \RR^{d_2\times d_3}$. \\      
\REQUIRE
    A structured-$\BA$ shared-$a$ (or shared-$a$ \MLWE)  encryption $(\isst{\BA}, \BB) \in \ZZ_q^{d_1\times N} \times \ZZ_q^{d_1\times d_2}$ of the matrix $\lfloor \Delta \BM\rceil \in \ZZ_q^{d_1\times d_2}$ as in Equation~\eqref{eq:mat-shared-a}.
\ENSURE
    A structured-$\BA$ shared-$a$ (or shared-$a$ \MLWE)  encryption $(\isst{\BA''}, \BB'') \in \ZZ_{q'}^{d_1\times N} \times \ZZ_{q'}^{d_1\times d_3}$ of $\lfloor \Delta (\BM \cdot \BU_0) \rceil $ under the matrix secret key $\BS \cdot \BU$. 
\STATE $\BB'\gets \BB \cdot \BU$;
\RETURN $\Rescale_{q\rightarrow q'}(\isst{\BA}, \BB')$.
\end{algorithmic}
\end{algorithm}

When $d_3$ is small compared to $N$, the size  of $\isst{\BA}$ ($d_1 N$ entries) could be viewed as significantly exceeding the size of the underlying message matrix ($d_1d_3$ entries).
However, in that case, the matrix~$\isst{\BA}$ is structured and should be seen as $d_1 / N$ polynomials of degree $N$ (in the structured-$\BA$ shared-$a$ case) or~$N / d_1$ polynomials of degree~$d_1$ (in the shared-$a$ \MLWE\ case). \textsf{Rescale} should be performed on this polynomial representation of~$\isst{\BA}$ in $O(\max(d_1, N))$ arithmetic operations. 
If $d_3 = 1$, we obtain a $\cpmv$ algorithm with precomputation. As structured-$\BA$ shared-$a$ and shared-$s$ coincide for a vector, we note that for large~$d_1$, the output is actually a vector in shared-$s$ format. 

\begin{theorem}\label{th:pcmm-with-precomp}
Let~$d_1, d_2, d_3$ be positive integers. Assuming pre-computation,
Algorithm~\ref{alg:pcmm-with-precomp} is correct and reduces 
$\cpmm_{d_1,d_2,d_3}$   in structured-$\BA$ shared-$a$ format or shared-$a$ \MLWE\ format to $\ModPPMM_{d_1,d_2,d_3}$, plus $O(d_1 \max(N, d_3))$ arithmetic operations on integers whose absolute values are smaller than the ciphertext modulus.
\end{theorem}

For simplicity, Algorithm~\ref{alg:pcmm-with-precomp} performs Rescale before the key-switchings required to go back from the secret key $\BS\cdot \BU$ to another secret key format that would be more suitable for subsequent computations. In practice, exchanging those two steps (as was done in~\cite{BCHPS24}) could yield slightly better precision. 

\subsubsection{Cost of pre-computation.}
The precomputation on~$\BS \cdot \BU$ corresponds to generating switching keys from the columns of $\BS\cdot \BU$ to a common secret key~$\sk$; it can be implemented as a product of~$\BU$ by an encryption of~$\BS$, so amounts to a $\cpmm$. This is indeed a precomputation as it is at least as costly as the actual computation itself. This shows that Algorithm~\ref{alg:pcmm-with-precomp} makes sense mainly when the result of this precomputation can be reused multiple times. Weight matrices for large-language models, mentioned in the introduction, provide a good example of this situation.

\subsubsection{Cost of pre- and post-processing.}
Using Table~\ref{tab:conversion_costs}, we notice that the costs of pre-processing and post-processing are negligible under the following conditions. 
\begin{itemize}
     \item[$\bullet$] For large matrices ($N|d_1$), as soon as $d_2 = \omega(\log N)$, $d_3 = \omega(\log d_2 \log N)$;
    \item[$\bullet$] For small matrices ($d_1 |N$), as soon as $d_3 = \omega(1)$ and $d_1 d_2 d_3 = \omega(\max(d_3, N/d_1) N \log N)$.
\end{itemize}

\subsection{Reducing $\ccmm$ to its PP counterpart}\label{sse:ccmm}
The algorithm focuses on large matrices, whose dimensions are no smaller than the \RLWE\ ring-degree~$N$. 
For ease of discussion, we focus on square $N\times N$ matrices, but note that the algorithm extends  to larger square matrices and, in that case,  provides a reduction to four $\ppmm$ of the same dimensions. 
In the case of smaller-dimensional matrices, however, using the \MLWE\ format would only provide a reduction to $\ppmm$ in dimensions~$N\times d_1 \times N$, which is of lesser interest.

Consider two matrices $\BM$ and $\BM'$ of dimensions~$N\times N$, given in  their column-wise \RLWE\ encryptions~$(\BA,\BB)$ and~$(\BA',\BB')$ under a common secret key~$\sk$. 
The following holds with $\isst{\BS} = \toep{\sk}$: 
\begin{equation}
    \label{eq:ccmm_input}
\Delta\cdot \BM~\approx~\isst{\BS}\cdot\BA+\BB 
\ \ \mbox{ and }
\ \ \Delta \cdot \BM'~\approx~\isst{\BS}\cdot \BA'+\BB' \enspace.
\end{equation}
Then, by using the \cmt algorithm, we transpose the column-wise encryption $(\BA',\BB')$ to obtain a row-wise encryption~$(\wBA',\wBB')$ of~$\BM'$:
\[
\Delta \cdot \BM' \approx \widetilde{\BA}' \cdot {\isst{\BS}}^t + \widetilde{\BB}' \enspace.
\]
We now multiply $(\BA, \BB)$ and $(\wBA',\wBB')$ in matrix forms:
\begin{equation*}\begin{split}
    \Delta^2 \cdot \BM \cdot \BM' & 
        \approx (\isst{\BS} \cdot \BA + \BB)\cdot(\wBA' \cdot {\isst{\BS}}^t +\wBB')\\
        &\approx \isst{\BS} \cdot \BA \cdot \wBA' \cdot {\isst{\BS}}^t + 
            \isst{\BS} \cdot \BA \cdot \wBB' + 
            \BB \cdot \wBA' \cdot {\isst{\BS}}^t + 
            \BB \cdot \wBB' \enspace. 
\end{split}\end{equation*}
We then transpose two row-wise \RLWE\ encryptions $(\BA \cdot \wBA', \textbf{0})$ of $\BA \cdot \wBA' \cdot {\isst{\BS}}^t$ and $(\BB \cdot \wBA',\textbf{0})$ of $\BB \cdot \wBA' \cdot {\isst{\BS}}^t$, obtaining two column-wise \RLWE\ encryptions $(\BD_0,\BD_1)$ and $(\BD_2,\BD_3)$ such that 
\begin{align*}
\BA \cdot \wBA' \cdot {\isst{\BS}}^t & \approx \isst{\BS} \cdot \BD_0 + \BD_1 \enspace, \\
\BB \cdot \wBA' \cdot {\isst{\BS}}^t & \approx\isst{\BS} \cdot \BD_2 + \BD_3 \enspace. 
\end{align*}

Putting it all together results in the following \RLWE\ encryption form of~$\BM \BM'$: 
\begin{equation*}
    \Delta^2 \cdot \BM \cdot \BM' \approx 
    {\isst{\BS}}^2 \cdot \BD_0 + \isst{\BS} \cdot (\BD_1 + \BD_2 + \BA \cdot \wBB') + (\BD_3+\BB \cdot \wBB') \enspace.
\end{equation*}
As the last step, we relinearize and rescale each column of $(\BD_0,~\BD_1+\BD_2+\BA\cdot\wBB',~\BD_3  + \BB \cdot \wBB')$.

This step outputs the column-wise desired encryption of~$\BM \cdot \BM'$, hence completing the $\ccmm$ procedure. 

To sum up, from the column-wise \RLWE\ encryptions of $\BM$ and~$\BM'$, we obtained a column-wise \RLWE\ encryptions of the product~$\BM \cdot \BM'$.
The overall procedure is detailed in Algorithm~\ref{alg:ccmm}. 
Transposition is implemented using Algorithm~\ref{alg:transpose}.

\begin{algorithm}
\caption{$\ccmm$ algorithm}\label{alg:ccmm} 
\begin{algorithmic}[1]
\REQUIRE Encryptions $(\BA, \BB), (\BA', \BB') \in \ZZ_q^{N \times N} \times \ZZ_q^{N\times N}$ of the matrices $\lfloor \Delta \BM\rceil,  \lfloor \Delta \BM'\rceil\in \ZZ_q^{N\times N}$ under the key $\isst{\BS}$, such that Equation~\eqref{eq:ccmm_input} holds.
\REQUIRE Switching keys $\swk_{pq,q, \sk(X^{2t+1})\rightarrow \sk}$ for $1\le t < N$, and relinearization key $\swk_{pq, q, \sk^2 \rightarrow \sk}$.
\ENSURE An encryption $(\BA'', \BB'') \in \ZZ_{q'}^{N \times N} \times \ZZ_{q'}^{N\times N}$ under $\isst{\BS}$, such that 
    \[ \Delta \cdot (\BM \cdot \BM') \approx \isst{\BS} \cdot \BA'' + \BB'' \bmod q' \enspace. \]
\STATE $(\wBA', \wBB') \gets \text{Transpose}((\BA', \BB');~\{\swk_{pq,q, \sk(X^{2t+1})\rightarrow \sk}\}_{1\le t < N} )$\enspace;
\STATE $\begin{bmatrix}\textbf{C}_{00}&\textbf{C}_{01}\\ \textbf{C}_{10}&\textbf{C}_{11}\end{bmatrix}
\gets \begin{bmatrix}\textbf{A}\\\textbf{B}\end{bmatrix} \cdot
\begin{bmatrix}\wBA'~\wBB'\end{bmatrix}$\enspace;
\STATE $(\BD_0,\BD_1)     \gets \text{Transpose}((\textbf{C}_{00}, \textbf{0});~\{\swk_{pq,q, \sk(X^{2t+1})\rightarrow \sk}\}_{1\le t < N} )$\enspace;
\STATE $(\BD_2,\BD_3) \gets \text{Transpose}((\textbf{C}_{10}, \textbf{0});~\{\swk_{pq,q, \sk(X^{2t+1})\rightarrow \sk}\}_{1\le t < N} )$\enspace;
\STATE $(\textbf{A}'',\textbf{B}'') \gets \rescale_{q\rightarrow q'}\left(\relin(\BD_0,0;~\swk_{pq, q, \sk^2 \rightarrow \sk})+ (\BD_1 + \BD_2 + \textbf{C}_{01} ,\BD_3 + \textbf{C}_{11})\right)$ ;
\RETURN $(\textbf{A}'',\textbf{B}'')$.
\end{algorithmic}
\end{algorithm}

The previous discussion and the correctness of Algorithm~\ref{alg:transpose} yield the following theorem. We observe that the call to $\Relin$ requires $N$ key-switchings, so $\widetilde{O}(N^2)$ operations on integers with absolute values smaller  than the ciphertext modulus, and the call to $\Rescale$ consumes $O(N^2)$ divisions on integers less than the ciphertext modulus.

\begin{theorem}\label{th:ccmm}
Assume that the ciphertext modulus is coprime with~$N$. Then Algorithm~\ref{alg:ccmm} is correct and reduces $\ccmm$ for square~$N$-dimensional matrices to four $\ModPPMM$
for square $N$-dimensional matrices, three calls to Algorithm~\ref{alg:transpose} (for~\cmt), plus $\widetilde{O}(N^2)$ operations on integers whose absolute values are smaller than the ciphertext modulus.
\end{theorem}

We note that transpose takes time $\tilde{O}(N^2)$, whereas the best theoretical cost known for $\ModPPMM$ as of today is  $\omega(N^{2.37})$~\cite{ADVSZ24}. The transposition steps are thus negligible compared to $\ModPPMM$, and Theorem~\ref{th:ccmm} can indeed be considered as a reduction from $\ccmm$ to four~$\ModPPMM$'s.

\subsubsection{Encoding structure: row or column.}

Algorithm~\ref{alg:ccmm} takes two column-wise encryptions of matrices as inputs and outputs a column-wise encryption of the product. It can be generalized it to row-wise encrypted inputs and outputs (or even to column-wise inputs and row-wise outputs, and vice versa). 
An interesting scenario is 
when we have a column-wise encryption of~$\BM$ and  a row-wise encryption of~$\BM'$.
In that case, we can complete $\ccmm$ with two $\cmt$'s rather than three, skipping the first~$\cmt$ (i.e., Step~$1$ in Algorithm~\ref{alg:ccmm}). 
More importantly, this gives flexibility on the shapes of~$\BM$ and~$\BM'$;
Algorithm~\ref{alg:ccmm} with Step~$1$ omitted works well with an~$N\times n$ matrix $\BM$ and an $n\times N$ matrix $\BM'$ for any $1\le n \le N$.

\subsubsection{Lightweight $\ccmm$ algorithm.}\label{ssec:ccmm-lightweight}
The $\cmt$ algorithm with a small key size directly implies a $\ccmm$ algorithm with a small key size, by using the lightweight version of the transpose algorithm outlined in Section~\ref{sec:light},  in Algorithm~\ref{alg:ccmm}.
With this modification, we obtain a lightweight version of our $\ccmm$ algorithm requiring four switching keys: one for relinearization, one for automorphisms, and two for updating the automorphism key.

\subsection{The \RGSW\ approach: reducing all operations to their PP counterparts}\label{sse:gsw}

We now turn to our most general reduction, from any homomorphic linear algebra operation to its plaintext counterpart (possibly with different dimensions). This generality comes at a price: the reduction rests on a \RGSW\ encryption of the matrix, on top of any \RLWE-type encryption format for this same matrix. In practice, this implies that 
\begin{itemize}
    \item[$\bullet$] we require this encryption to either be given as an input, or we perform a precomputation as in Section~\ref{sec:lwetogsw}; 
    \item[$\bullet$] we lose a constant factor in the reduction, as the computation involves an \emph{auxiliary integer} in order to prevent an excessive growth of error terms. 
\end{itemize}

\subsubsection{Description of the algorithm.}
We describe our algorithm for the matrix-vector product. The exact same description would apply for matrix-matrix computation, by seeing the second matrix as a batch of column vectors. 
The description is almost agnostic to the encryption formats of the inputs~$\BM \in \RR^{d_1 \times d_2}$ and~$\bv \in \RR^{d_2}$. The format conversion algorithms described in Section~\ref{sec:conversions} adapt to the \RGSW\ encryption by handling it as a pair of \RLWE\ encryptions. 
We stress that the algorithm can handle different secret keys, ring degrees and encryption formats for the input matrix and the input vector. For a matrix vector product $\BM\cdot \bv$, we let~$\BS^{\BM}$ and~$N^{\BM}$ (resp.\ $\BS^{\bv}$ and $N^{\bv}$) denote the matrix form of the secret key and ring degree for~$\BM$ (resp.~$\bv$).
As may be seen in the formal description in Algorithm~\ref{alg:general}, the main idea is that the so-called external product $\RGSW \times \RLWE \rightarrow \RLWE$ extends from scalars to matrices and vectors.

\begin{algorithm}[H]
\caption{\label{alg:general} Encrypted matrix-vector multiplication}
    \begin{algorithmic}[1]
\REQUIRE {$(\BA_0, \BB_0) \in \mathbb{Z}_{pq}^{N^\BM\times d_2} \times \mathbb{Z}_{pq}^{d_1 \times d_2}$ encrypting $p\lfloor \Delta \BM\rceil$ (modulo $pq$).}
\REQUIRE{
$(\BA_1, \BB_1) \in \mathbb{Z}_{pq}^{N^\BM\times N^\bv} \times \mathbb{Z}_{pq}^{d_1 \times N^\bv}$ encrypting $p\lfloor \Delta \BM\rceil\BSv$ (modulo $pq$).}
\REQUIRE{$(\ba, \bb)\in \mathbb{Z}_{q}^{N^\bv} \times \mathbb{Z}_{q}^{d_2}$ encrypting $\lfloor \Delta \bv\rceil$ (modulo $q$).}
\ENSURE{$(\ba', \bb') \in {\mathbb Z}_{q'}^{d_1} \times {\mathbb        Z}_{q'}^{N^{\BM}}$ encrypting $\lfloor \Delta (\BM\cdot \bv)\rceil $ (modulo $q$).} 
    \STATE $\ba'\leftarrow \lfloor \frac{1}{p}(\BA_1\ba + \BA_0\bb \bmod pq)\rceil \bmod q$;
    \STATE $\bb' \leftarrow \lfloor \frac{1}{p}(\BB_1\ba+\BB_0\bb \bmod pq)\rceil  \bmod q$;
    \STATE Return $\Rescale_{q\rightarrow q'}(\ba', \bb')$.
    \end{algorithmic}
\end{algorithm}

If, as in the previous subsections, we  neglect error terms, then the purely algebraic verification of this algorithm is direct, as illustrated by the following chain of equalities modulo~$pq$: 
\begin{align*}
\Delta^2 \cdot p\BM \cdot \bv 
& \approx \Delta \cdot  p\BM \cdot (\BS^\bv \ba + \bb) \\
& \approx \Delta \cdot  (p\BM \cdot \BS^\bv) \cdot \ba + \Delta \cdot  p\BM \cdot \bb \\
& \approx \BSM \BA_1 \cdot \ba + \BB_1 \cdot \ba + \BSM \BA_0 \cdot \bb + \BB_0 \cdot \bb  \\
& \approx \BSM (\BA_1\cdot \ba + \BA_0\cdot \bb) + (\BB_1\cdot \ba + \BB_0 \cdot \bb) \enspace . 
\end{align*}
At first sight, it may seem that the correctness of the algorithm follows (making the role of~$p$ seem somewhat artificial). This description is however too rough, because the error terms in the encryptions of~$p\BM$ and~$p\BM \cdot \BS^{\bv}$ are multiplied by~$\ba$ and~$\bb$. This yields terms that are not small when compared to the ciphertext modulus. This is the reason why an auxiliary modulus is introduced.

\subsubsection{Error analysis.}
We now provide a rigorous analysis of Algorithm~\ref{alg:general}, taking into accounts error terms. For this reason and oppositely to previous subsections, we state the result for \emph{encoded} integer matrices, without approximate equalities. 

\begin{theorem}\label{th:gsw}
Let $d_1$ and~$d_2$ be two positive integers, $\BM$ be an integer matrix of size $d_1 \times d_2$ and $\bv \in \ZZ^{d_2}$ be an integer vector. We define $K = \max(\|\BM\|_\infty, \|\bv\|_\infty)$.  

 We let $\BSM \in \ZZ^{d_1\times N^\BM} $ and $\BSv \in \ZZ^{d_2\times N^\bv}$ be secret keys in matrix form in a format corresponding to the encryption formats for~$\BM$ and~$\bv$. Let $p$ and $q$ be two positive integers.
Let $(\BA_0, \BB_0)\in \mathbb{Z}_{pq}^{N^\BM\times d_2} \times \mathbb{Z}_{pq}^{d_1 \times d_2}$, $(\BA_1, \BB_1)\in \mathbb{Z}_{pq}^{N^\BM\times N^\bv} \times \mathbb{Z}_{pq}^{d_1 \times N^\bv}$ and~$(\ba, \bb)\in \mathbb{Z}_{q}^{N^\bv} \times \mathbb{Z}_{q}^{d_2}$ respectively be encryptions of $p\BM$, $p\BM \BSv$ and~$\bv$, as follows: 
\begin{align*}
p\BM & = \BSM  \BA_0 + \BB_0 + \BE_0 \bmod pq \enspace,\\
p\BM \BSv & = \BSM   \BA_1 + \BB_1 + \BE_1 \bmod pq \enspace,\\
\bv & = \BSv \ba  + \bb + \be \bmod q\enspace, 
\end{align*}
where 
$\BE_0 \in \mathbb{Z}^{d_1 \times d_2}$, $\BE_1 \in \mathbb{Z}^{d_1 \times N^{\bv}}$ and~$\be \in \mathbb{Z}^{d_2}$. We define $\varepsilon =
\max(\|\BE_0 \|_\infty, \|\BE_1\|_\infty, \|\be \|_\infty)$.

Then $(\ba', \bb') = (\lfloor \frac{1}{p}(\BA_1\ba + \BA_0 \bb \bmod pq)\rceil, 
\lfloor \frac{1}{p}(\BB_1\ba+\BB_0\bb \bmod pq)\rceil) \in \mathbb{Z}_q^{N^\BM} \times \mathbb{Z}_q^{d_1}$ satisfies 
\[
\BSM \ba' + \bb' = \BM\cdot \bv  + \be' \bmod q \enspace \ \mbox{ with } \ \enspace  \|\be'\| \leq \frac{\opnorm{\BSM}_\infty+1}{2} + \left( \frac{q}{p} + K\right) \max(d_2, N^\bv) \varepsilon 
\enspace,
\]
where $\opnorm{\BSM}_\infty = \max_{0\le i < d_1}\sum_{0 \le j< N^\BM} |\BSM_{ij}|$. 
\end{theorem}

\begin{proof}
There exists an integer vector $\bk \in \mathbb{Z}^{d_2}$ such that 
we can rewrite $p\cdot \BM \cdot (\bv- \be)$ as follows, modulo~$pq$:
\begin{align*}
p\cdot \BM \cdot (\bv - \be) & = p \BM \cdot (\BSv \cdot \ba  + \bb + q \bk) \\
& = (p \BM \BSv)\cdot \ba + (p \BM) \bb  \\
& = \left( \BSM \BA_1 + \BB_1 + \BE_1 \right) \ba
+ (\BSM\BA_0 + \BB_0 + \BE_0) \cdot \bb \\
& = \BSM\cdot (\BA_1 \cdot \ba + \BA_0 \cdot \bb) + (\BB_1 \cdot \ba + \BB_0 \cdot \bb) +  
 (\BE_1 \cdot \ba + \BE_0 \cdot \bb) \enspace. 
\end{align*}

As such, we have (with~$\{ \cdot \}$ denoting the centered fractional part):
\begin{eqnarray*}
\left\| 
(\BSM \cdot \ba' + \bb') - \BM \cdot (\bv  - \be) 
\right\|_\infty & \le &
\left\| \BSM \cdot \left\{ \frac{1}{p}(\BA_1 \cdot \ba + \BA_0 \cdot \bb \bmod pq)
\right\}\right\|_\infty 
\\ 
&&  + \left\| \left\{ \frac{1}{p}(\BB_1 \cdot \ba  + \BB_0 \cdot \bb \bmod pq) \right\} \right\|_\infty \\
&& + \frac{1}{p} \left\| \BE_1 \cdot \ba + \BE_0 \cdot \bb \right\|_\infty\\
& \le &  \frac{\opnorm{\BSM}_\infty+1}{2} + \max(d_2, N^\bv) \cdot \frac{q}{p}\varepsilon.
\end{eqnarray*}
To complete the proof, we observe that
$\| \BM \cdot (\bv - \be) - \BM \cdot \bv \|_\infty  \le  Kd_2 \varepsilon$.
\qed
\end{proof}

The term $\opnorm{\BS^M}_\infty$ generalizes the Hamming weight of  ternary polynomial secret keys which usually appears in  estimates for the rescaling error. If considering ternary secret keys, 
\begin{itemize}
    \item[$\bullet$] in the shared-$a$ case, it is equal to the largest Hamming weight of all keys involved; 
    \item[$\bullet$] in the \MLWE\ case, it is equal to the total Hamming weight of the key (i.e., the sum of Hamming weights of all secret key parts).  
\end{itemize}   

In Algorithm~\ref{alg:general}, we use the error bound from Theorem~\ref{th:gsw} for an encoded matrix $\BM_0$, i.e., for $\BM_0 = \lfloor \Delta\cdot \BM \rceil$ where $\BM \in \RR^{d_1\times d_2}$.  Let us assume that $\BM$ has coefficients in~$[-1, 1]$; in that case, we can take $K = \Delta$. We also assume that $\varepsilon = O(1)$. The final rescaling step in Algorithm~\ref{alg:general} will have the effect of dividing the error term of Theorem~\ref{th:gsw} while introducing a rounding error bounded by $(1 + \opnorm{\BS^\BM}_\infty)/2$. 
Discarding lower-order terms, we find that after rescaling the error is of the order of
\[
\left(1 + \frac{1}{\Delta}\right) \frac{\opnorm{\BS^\BM}_\infty+1}{2} + \left(1 + \frac{q}{\Delta p}\right) \max(d_2, N^\bv) \enspace. 
\]
By taking sparse ternary secret keys and $p \approx q/\Delta,$ we thus obtain a final error after rescaling of the order of $\max(d_2, N^\bv)$. Oppositely, the matrix-vector product is scaled by a factor~$\Delta$. Overall, we expect to lose $\approx \log \max(d_2, N^\bv)$ bits of precision, which is similar to what happens in a cleartext computation.

\subsubsection{Shared-$s$ case.}
When $N^\BM|d_1$ and/or $N^\bv | d_2$, the shared-$s$ representation can be preferred over the shared-$a$ one; this means that the product $\BM\cdot \bv$ is treated as $d_1/N^{\BM}$ products $\BM_i \cdot \bv$, or $d_2/N^\bv$ products $\BM \cdot \bv_j$, or $d_1 d_2/(N^{\BM} N^{\bv})$ products $\BM_i \cdot \bv_j$. 

The proof of Theorem~\ref{th:gsw} carries over by decomposing $\BM$ and/or $\bv$ into blocks of dimensions~$N^{\BM}$ and/or~$N^{\bv}$. For each of these blocks, the theorem applies. It can be checked that the $\ppmv$'s of Algorithm~\ref{alg:general} for various blocks can be grouped, so that Algorithm~\ref{alg:general} still holds. The dimensions of the matrices involved are then described in Table~\ref{tbl:shareds_dims}.

\begin{table}
\label{tbl:shareds_dims}
\caption{Dimensions of the matrices used in Algorithm~\ref{alg:general} while blocking for shared-$s$ representations of~$\BM$ and/or $\bv$.}
\begin{center}
\begin{tabular}{c|c|c|c|}
        &        shared-$s$ $\BM$ &  shared-$s$ $\bv$ &  shared-$s$ $\BM$ and $\bv$ \\ \hline
$\BA_0$   &      $d_1\times d_2$     & $N^\BM \times d_2$ &   $d_1 \times d_2$ \\ \hline
$\BB_0$   &      $d_1 \times d_2$    &  $d_1 \times  d_2$ &   $d_1 \times d_2$  \\ \hline
$\BA_1 $  &      $d_1 \times N^\bv$  & $N^\BM \times d_2$  &  $d_1 \times d_2$  \\ \hline 
$\BB_1$   &      $d_1 \times N^\bv$  &  $d_1 \times d_2$  &   $d_1 \times d_2$  \\ \hline 
$\ba $    &      $N^\bv$             &  $d_2$             &   $d_2$           \\ \hline 
$\bb$     &      $d_2$               &  $d_2$             &   $d_2$           \\    
\hline
\end{tabular}
\end{center}
\end{table}

\subsubsection{Cost analysis of Algorithm~\ref{alg:general}.}

The inspection of Algorithm~\ref{alg:general} leads to the following theorem.
\begin{theorem}\label{th:cost}
Let $d_1, d_2$ be positive integers. Let $\BM \in \mathbb{R}^{d_1\times d_2}$ and $\bv\in \mathbb{R}^{d_2}$ be encoded as $\lfloor \Delta \BM \rceil$ and~$\lfloor \Delta \bv\rceil$, respectively. Let~$N^\BM$ and~$N^{\bv}$ be the encryption dimensions of~$\BM$ and~$\bv$, respectively, and let~$pq$ be the encryption modulus of~$\BM$.  Algorithm~\ref{alg:general} is correct and reduces $\ccmv$ in dimension $d_1 \times d_2$ to four
$\ModPPMv$'s modulo~$pq$ for each pairs of dimensions in $\{N^\BM, d_1\} \times \{ N^{\bv}, d_2\}$, plus $O(d_1 + N^\BM)$ divisions on integers whose absolute values are~$< pq$. 
\end{theorem}

\subsubsection{Extensions.}
From Algorithm~\ref{alg:general}~and~Theorem~\ref{th:cost} described in the $\ccmv$ case, we derive a number of extensions as particular cases. 
\begin{itemize}
    \item[$\bullet$] A $\ccmm$ algorithm  by viewing the latter as a batch of $\ccmv$'s; this corresponds to using Algorithm~\ref{alg:general} with $(\ba, \bb)$ being matrices rather than vectors. We notice that the reduction carries over to that setting, i.e., we obtain a reduction from $\ccmm$ to modular $\ppmm$'s of various dimensions. 
    \item[$\bullet$] A $\pcmv$ algorithm for  by using $\BA_0 = {\bf 0}$, which yields a reduction from $\pcmv$ to three $\ppmv$'s modulo~$pq$ of various dimensions.
    \item[$\bullet$] A $\cpmv$ algorithm by setting $\ba = {\bf 0}$.
    \item[$\bullet$] Finally, both $\pcmm$ and $\cpmm$ can be handled as batches of $\pcmv$'s and $\cpmv$'s, respectively. 
\end{itemize}

In all those cases, our approach reduces the ciphertext or mixed plaintext-ciphertext operations to a number of similar operations (but with dimensions which may vary) modulo a small integer.  We state the following simplified theorem in the $\ccmm$ case.
\begin{corollary}\label{cor:gsw-mat}
Let $d_1, d_2$ be positive integers. Let $\BM \in \mathbb{R}^{d_1\times d_2}$ and $\bv\in \mathbb{R}^{d_2}$ be encoded as $\lfloor \Delta \BM \rceil$ and~$\lfloor \Delta \bv\rceil$, respectively. 
Let~$N^\BM$ and~$N^{\bv}$ be the encryption dimensions of~$\BM$ and~$\bv$, respectively, and let~$pq$ be the encryption modulus of~$\BM$.  
Algorithm~\ref{alg:general} reduces the approximate computation of $\lfloor \Delta (\BM \cdot \bv)\rceil$ using \RLWE, shared-$a$ or \MLWE\ encryption formats to four $\ModPPMv$'s modulo~$pq$, one for each pair of dimensions in $\{d_1, N^{\BM}\} \times \{d_2, N^{\bv}\}$, plus $O(d_1 + N^\BM)$ divisions on integers whose absolute values are~$<pq$. 
\end{corollary}

If both matrices use the same encryption format:
\begin{itemize}
    \item[$\bullet$] when $d_1 \ll N^{\BM}, d_2 \ll N^{\BU}$, the dominant cost corresponds to the products of the $\BA$ parts, with cost
    $O(N^\BU N^\BM)$; 
    \item[$\bullet$] when $d_1 \approx N^\BM, d_2 \approx N^\BU$, all the products have similar costs; 
    \item[$\bullet$] when $d_1 \gg N^{\BM}, d_2 \gg N^{\BU}$ tend to infinity, using shared-$a$ gives a cost equivalent to the cost of the corresponding plaintext computation modulo~$pq$ (as soon as the precomputation cost of moving to shared-$a$ can be neglected).
\end{itemize}

For the sake of comparison with~Algorithm~\ref{alg:ccmm} which reduces $\ccmm$ to $\ModPPMM$ modulo $q$, we assume that $p$ and $q$ are similar-magnitude integers and that a matrix~$\BM$ 
modulo~$pq$ is represented as a pair $(\BM_p, \BM_q)$ via the Chinese Remainder Theorem. 
In this setting, a matrix-matrix product modulo~$pq$ has the same cost as two matrix-matrix products modulo~$q$. As such, we note that Corollary~\ref{cor:gsw-mat} reduces to~8 $\ModPPMM$'s with a modulus of the order of~$q$ while Algorithm~\ref{alg:ccmm} reduces to 4 $\ModPPMM$'s with a similar modulus. We thus expect the latter to outperform the former by a factor~$\approx 2$, when it is applicable (recall that it is more limited in scope, as it requires square matrices of dimension~$\ge N$).

\subsubsection{Costs of format conversions.}
We now estimate the costs of format conversions for the \RGSW-based algorithm, in the general $\ccmm$ case (recall that $\ccmv$ corresponds to the case $d_3 = 1$). Before running the algorithm, we switch to the right encryption format either for the two \RLWE\ ciphertexts encrypting~$\BM$ and~$\BM \BSU$, or for the \RLWE\ ciphertext encrypting~$\BU$, or both. After running the algorithm, we  switch the encryption format for the \RLWE\ ciphertext encrypting~$\BM\cdot \BU$. The encryption formats for~$\BM$ and $\BM\BSU$ must be the same, but we still have to consider four cases depending on the encryption formats of~$\BM$ and~$\BU$. For the sake of brevity, we ignore the case where at least one of the operands is in shared-$s$ format. Table~\ref{tab:costsgsw}, built using Table~\ref{tab:conversion_costs}, gives the formatting costs in all those cases. 

\begin{table}
\caption{\label{tab:costsgsw}Pre- and post-processing costs for the $\RGSW$-based $\ccmm$ algorithm, up to a multiplicative constant.}
\[
\begin{array}{|c|c|c|}
\hline
& \textrm{pre-processing} & \textrm{post-processing} \\ \hline
N^\BM|d_1, N^\BU |d_2 & d_1 d_2 \log N^\BM \log(d_1/N^\BM) + d_2 d_3 \log N^\BU \log(d_2/N^\BU) & d_1d_3 \log N^\BM\\ \hline
d_1|N^\BM, N^\BU|d_2 &  d_1 d_2 + d_2 d_3 \log N^\BU \log(d_2/N^\BU)&  \max(d_3, N^\BM/d_1) N^\BM \log N^\BM\\ \hline
N^\BM|d_1, d_2|N^\BU & d_1 d_2 \log N^\BM \log (d_1/N^\BM) + d_2 d_3 & d_1d_3 \log N^\BM \\ \hline
d_1 |N^\BM, d_2|N^\BU & d_1 d_2 + d_2 d_3 &  \max(d_3, N^\BM/d_1) N^\BM \log N^\BM \\ \hline
\end{array}
\]
\end{table}

We note, in particular, that all pre-processings cost at least $d_1 d_2$. This shows that for $\ccmv$, pre-processing is non-negligible and even dominates for large matrices with $N^\BM|d_1$ (it  is hence rather a pre-computation). 

\subsection{Integrating the  linear algebra algorithms in a larger computation}
Our algorithms rely on coefficient encoding and dedicated encryption formats (shared-$a$, MLWE, RGSW, etc.) while almost all CKKS computations are performed in slot-encoded representation, in order to benefit from the SIMD capability, and under plain RLWE encryption. 

Regarding format conversions, we suggest to perform conversion whenever needed using the algorithms from Section~\ref{sec:conversions}; the cost of linear algebra dominates the format conversion costs in almost all cases. In some very specific cases, the computation can be continued using the output format of the linear algebra step; see for instance~\cite{BCHPS24} for the case of bootstrapping using the shared-$a$ format. 

Regarding the use of the coefficient encoding, we suggest the following strategies in order to integrate linear algebra in the context of larger computations.
\begin{itemize}
    \item[$\bullet$] Perform linear algebra computations during the StC-first variant of bootstrapping. This is likely the best solution for large linear algebra computations, as their cost, when performed using slot-based methods~\cite{JKLS18}, may be higher than performing the linear algebra computation at low level and then bootstrapping; 
    this is even more true when several linear algebra computations happen in a row, as one can remain in coefficient encoding throughout all of them (consider, e.g., the attention phase of LLMs where the result of two PCMMs are immediately used as operands to a CCMM). In that case, as the number of ciphertexts to be bootstrapped is large, the use of the shared-$a$ technique for bootstrapping~\cite{BCHPS24} is recommended.
    
    \item[$\bullet$] Some of our algorithms are compatible with slot encoding; this is in particular the case of the CPMM algorithms when the number of rows of the encrypted matrix is no smaller than the ring-degree, and of the RGSW algorithm where the matrix can be made of slot-encoded ciphertexts, when the number of rows is no smaller than the ring degree. 
    The vector, on the other hand, must use coefficient encoding. In that case, the result is slot-encoded. 
\end{itemize}

\section{Implementation}\label{sec:experiments}
We have implemented Algorithms~\ref{alg:fast_fmtswitch} and~\ref{alg:transpose} for format conversions, and Algorithms~\ref{alg:pcmm},~\ref{alg:ccmm} and~\ref{alg:general} for $\cpmm$ and $\ccmm$. Regarding Algorithms~\ref{alg:pcmm} and~\ref{alg:general}, we focused on the case of matrices with $d_1 = d_2 \ge N$ and varying~$d_3$.  

Our implementations make use of highly optimized floating-point arithmetic libraries implementing the BLAS (basic linear algebra subroutines). In this context, we mostly use the {\tt dgemm} (double precision general matrix-matrix multiplication) and {\tt dgemv} (double precision general matrix-vector multiplication) primitives. 
As our algorithms provide reductions to $\ModPPMM$, we describe how we implemented $\ModPPMM$ based on $\fpPPMM$.
We note that, in order to implement our algorithms, one might also choose to opt for ad hoc libraries offering implementations of $\ModPPMM$, often by reducing them to BLAS libraries; FLINT~\cite{flint} is one such library. 

Our experiments are based on the HEaaN library~\cite{heaanlib}, which implements CKKS. We used the Open\-\mbox{BLAS-0.3.26} implementation of BLAS. All the provided running times correspond to experiments on an Intel(R) Xeon(R) Gold 6342 CPU @ 2.80GHz, using a single thread. Even though we choose to provide single-thread timings to allow for easier comparison, we stress that most BLAS implementations (including OpenBLAS) are able to take advantage of available parallelism (including GPU parallelism in the case of cuBLAS). Our algorithms mostly rely on calls to BLAS and can thus take advantage of the available parallelism in the same proportion. The only exception to this is the Transpose algorithm, for which the \tweak component exhibits less parallelism. 

In all cases, our timings are averaged over 10 runs. All our parameters enable around 128-bit security, according to the lattice estimator~\cite{lattice-estimator}.\footnote{commit 352ddaf4a288a0543f5d9eb588d2f89c7acec463, Sep.\ 17, 2025.} 
In Tables~\ref{tab:pcmm},~\ref{tab:ccmm-jaihyun} and~\ref{tab:ccmm-yj-symmetric-degree}, the precision (in bits) of an approximation~$\widetilde{\BM}$ to a matrix  $\BM\ne {\bf 0}$ refers to the \emph{relative} precision defined as:
\[
\mathrm{prec}_{\widetilde{\BM},\BM} =\log\|\BM\|_\infty - \log\left(\|\BM-\widetilde{\BM}\|_\infty\right) \enspace.
\]

\subsection{Reduction to BLAS}\label{ssec:impl-aspects}

Our algorithms from Section~\ref{sec:algorithms} all reduce homomorphic linear algebra operations to exact modular $\ppmm$ or $\ppmv$, i.e., multiplications of matrices/vectors over~$\ZZ_{Q}$ for some positive integer~$Q$. 
In practice, we want to perform these $\ModPPMM$ operations using $\fpPPMM$, so as to enjoy fast implementations of the latter. 
We shall describe three different techniques for this task. The first one is a reduction, in the complexity-theoretic sense, from $\ModPPMM$ to $\fpPPMM$, thus providing a complete reduction chain from $\cpmm$ and $\ccmm$ to~$\fpPPMM$. The second one sees the modular~$\ModPPMM$ as an integral~$\ppmm$ and computes an approximation of it using floating-point arithmetic; in that case, we analyze the impact of this approximation on the error in the $\cpmm$ or $\ccmm$ computation. 
Finally, the last technique rests on the use of the $\textsf{ModSwitch}$ operation on ciphertexts, which allows to change the modulus of a ciphertext. Using the Chinese Remainder Theorem (CRT), this allows us to turn a $\cpmm$ or $\ccmm$ into a number of $\ModPPMM$'s, but with a small modulus, for which our first technique applies in a more efficient way.  

In our implementations, we shall use the first method combined with the second one in the PC case, and the last one in the CC case. 

\subsubsection{Strategy 1: chopping off the inputs.}
This technique is folklore and is implemented in libraries such as~FFLAS-FFPACK~\cite{fflas-ffpack, DGP08}.
The IEEE-754~\cite{ieee754} standard for floating-point computation ensures that integer computations can be performed exactly using floating-point arithmetic as long as the input, output, and any intermediate results are integers whose absolute values are $< 2^{53}$. We deduce the following lemma.

\begin{lemma}\label{le:strategy1}
Let $d_1$ and~$d_2$ be two positive integers,
 $\BM\in \ZZ^{d_1 \times d_2}$ and~$\bv\in \ZZ^{d_2}$. Define $\theta_1 = \|\BM\|_\infty, \theta_2 = \|\bv\|_\infty$, 
 $K = \lfloor \sqrt{2^{53}/d_2} \rfloor$
 and  $k_i = \lceil \log (\theta_i) / \log K \rceil$ for $i \in \{1, 2\}$. 
Then the product $\BM\cdot \bv$ can be computed exactly using~$k_1 k_2$ calls to double-precision floating-point arithmetic matrix-vector multiplication in dimension~$d_1 \times d_2$, plus~$O(k_1\cdot k_2\cdot d_2)$ operations on integers whose absolute values are~$\le d_2 \theta_1 \theta_2$.
\end{lemma}
\begin{proof}
An integer~$r$ with $|r| < K^k$ can be decomposed in base~$K$ as 
\[
r = \sum_{0 \leq i< k} r_i K^i \enspace, 
\]
where the $r_i$'s are of the sign of~$r$ and satisfy~$|r_i| < K$ (we allow the sign of $0$ to be both 1 and $-1$).
We decompose the coefficients of $\BM$ and~$\bv$ in this way to obtain 
\begin{equation}\label{eq:decomp}
\BM = \sum_{0 \leq i < k_1} \BM_i \cdot K^i, 
\ \ \ \bv = \sum_{0 \leq  i < k_2} \bv_i \cdot K^i \enspace,
\end{equation}
where the $\BM_i$'s and~$\bv_i$'s are matrices and vectors with coefficients smaller than~$K$ in absolute value.
Then, we have 
\[
\BM \cdot \bv = \sum_{0\le i < k_1,0\le j < k_2} \BM_i \cdot \bv_j \cdot K^{i+j} \enspace. 
\]
Each of the matrix-vector products $\BM_i \cdot \bv_j$ leads to computations where the largest intermediate result is, in absolute value, no larger than
\[
d_2 \|\BM_i\|_\infty \| \bv_j\|_\infty < d_2 K^2 \leq 2^{53} \enspace.
\]
These products 
can be thus computed exactly using double-precision floating-point arithmetic. We then compute their sum using integer arithmetic on integers $\le d_2 \theta_1 \theta_2$. \qed
\end{proof}

When handling $\ModPPMM$, we use this lemma to compute the underlying product exactly over the integers, then reduce it modulo the modulus. We obtain a reduction of $\ModPPMM$ to $k_1 k_2$ $\fpPPMM$'s.

\subsubsection{Strategy 2: truncation.}
This is a modification of Strategy~1 where we allow the $\ModPPMM$ to be approximate and keep only the largest term in a variant of the decomposition of Equation~\eqref{eq:decomp}. Concretely, we set $K=2^{53}$ and the decomposition starts from the high-order bits of~$\BM$ and~$\bv$ rather than from the low-order bits. This method is mostly applicable to the $\ModPPMM$ which contributes to the $\BB$-part of the final result as, in view of the decryption equations
\[
\isst{\BS}\cdot \BA + \BB \approx \Delta \cdot\BM  \ \textrm{ or } \ \isst{\BA}\cdot \BS + \BB \approx \Delta \cdot \BM \enspace, 
\]
an error of magnitude $e$ on $\BB$ results in the same error on the decrypted plaintext, and in an error $\approx e/\Delta$ on the underlying decoded message. 

We first note that we can round any integer $x$  to 53 bits of precision (assuming that $|x|$ is not larger than the numbers handled by the floating-point format, i.e., $\approx 2^{1023}$ for the IEEE-754 double precision) and thus obtain a double-precision floating-point number which we will denote by~$\tx$, such that 
\[
| x - \tx | \lesssim 2^{-53} \cdot |x| \enspace. 
\]
In order to use floating-point arithmetic to approximate a product $\BB_1 \cdot \BB_2$, where $\BB_1 \in \ZZ_Q^{d_1\times d_2}$ and~$\BB_2 \in \ZZ_Q^{d_2\times d_3}$, we perform this truncation coordinate-wise in order to obtain matrices of floating-point numbers $\widetilde{\BB_1}$ and $\widetilde{\BB_2}$ such that 
\[
\|\widetilde{\BB_i} - \BB_i \|_\infty \lesssim 
2^{-53} \cdot \|\BB_i\|_\infty, \ \ i \in \{1,2\} \enspace .
\]

From this inequality, we derive:
\begin{align*}
\| \widetilde{\BB_1}\cdot \widetilde{\BB_2} - \BB_1 \cdot \BB_2 \|_\infty  & \le d_2 \|\widetilde{\BB_2}\|_\infty \| \widetilde{\BB_1} - \BB_1 \|_\infty + d_2 \|\BB_1\|_\infty \cdot \| \widetilde{\BB_2} - \BB_2 \|_\infty \\
& \lesssim 2^{-52} \cdot d_2 \|\BB_1\|_\infty \|\BB_2\|_\infty \enspace. 
\end{align*}
\textbf{}For $x \in \RR$, we define $\{x\}_Q := x - Q\lfloor x/Q\rceil$.
If $x,y \in \RR$, we then have $|\{x\}_Q - \{y\}_Q \bmod Q| \le |x - y|.$  Applying this remark to the coordinates of $\widetilde{\BB_1}\cdot \widetilde{\BB_2}$ and $\BB_1 \cdot \BB_2$, we obtain 
\[
\left \|
\widetilde{\BB_1}\cdot \widetilde{\BB_2} - Q \left \lfloor \frac{\widetilde{\BB_1} \cdot \widetilde{\BB_2}}{Q} \right\rceil
- \BB_1 \cdot \BB_2 \bmod Q \right\|_\infty \lesssim
2^{-51} \cdot  d_2 \left\|\BB_1\right\|_\infty \left\|\BB_2\right\|_\infty \enspace.
\]
As $\widetilde{\BB_1}$ and $\widetilde{\BB_2}$ are matrices of floating-point numbers, an integer approximation $\widehat{\BB}$ of $\widetilde{\BB_1}\cdot \widetilde{\BB_2} - Q \left \lfloor \widetilde{\BB_1} \cdot \widetilde{\BB_2} / Q \right\rceil$ can be computed using a double-precision floating-point matrix multiplication and double precision floating-point arithmetic, with an additional numerical error  of the order of $\lesssim 2^{-53} \cdot  d_2 \|\BB_1\|_\infty \|\BB_2\|_\infty$. Overall, we obtain
\[
\| \widehat{\BB} - \BB_1 \cdot \BB_2 \bmod Q\|_\infty \lesssim 2^{-50} \cdot  d_2 \|\BB_1\|_\infty \|\BB_2\|_\infty \enspace.
\]

When we consider the error on the underlying cleartext, encoded with scaling factor~$\Delta$, we deduce that using one single floating-point $\ppmm$ for computing the $\BB$-part thus induces, after rescaling, decryption and decoding, an error on the underlying message of the order of~$2^{-50} \cdot d_2 \|\BB_1\|_\infty \|\BB_2\|_\infty / \Delta^2$. 
In particular, combined with Algorithm~\ref{alg:pcmm-with-precomp} in the case where the plaintext matrix is known beforehand and precomputation is available, a reduction from $\cpmm$ in dimension $d_1 \times d_2 \times d_3$ to a single $\fpPPMM$ with the same dimensions.

One can also attempt to use floating-point arithmetic for the $\BA$-part, in the case of the structured-$\BS$ formats, for Algorithms~\ref{alg:pcmm} and~\ref{alg:general}. The main difference is that the $\BA$-part occurs in the decryption equation as~$\isst{\BS}\cdot \BA$: the  error resulting from using floating-point approximation for the computation of $\BA_1\cdot \BA_2$ on the underlying plaintext will now be of the order of 
\[
\lesssim 2^{-50} \cdot \max(d_2, N) \opnorm{\isst{\BS}}_\infty \|\BA_1\|_\infty \|\BA_2\|_\infty \enspace. 
\]
When using the structured-$\BS$ shared-$a$ format and  ternary keys with Hamming weight~$h$, we thus expect to lose~$(\log h)/2$ further bits of relative precision (assuming~$h$ is sufficiently large). 

Finally, note that intermediate compromises between Strategies~1 (more precise) and~2 (more efficient) are possible by varying the number of terms that we keep in Equation~\eqref{eq:decomp}. 

\subsubsection{Strategy 3: modulus switching and CRT.}\label{sec:strategy3}
This strategy follows a different path, by operating at the ciphertext level before the reduction to $\ModPPMM$. We leverage the fact that ciphertexts are amenable to \emph{modulus switching}. 
\begin{itemize}
     \item[$\bullet$] \modswitch. Given a ciphertext $\ct \in \R_{Q}^2$ encrypted under the secret key $\sk \in \R$, and $Q' \approx Q$, \modswitch\ returns a ciphertext $\ct' \in \R_{Q'}^2$ such that $\Dcd(\Dec(\sk, \ct)) \approx \Dcd(\Dec(\sk, \ct'))$.
\end{itemize}

Compared to Lemma~\ref{le:strategy1}, we keep the same value for~$K$ and define $k = \max(k_1, k_2)$. We define~$Q' = p_0 \ldots p_{k-1}$, where the $p_i$'s are pairwise coprime positive integers smaller than~$K$. 
Using \textsf{ModSwitch} , we can thus assume that the $\cpmm$ or $\ccmm$ to be computed is defined modulo~$Q'$. Algorithms~\ref{alg:pcmm},~\ref{alg:ccmm} and~\ref{alg:general} then reduce these problems to a number of~$\ModPPMM$'s modulo~$Q'$. If the underlying matrices have dimensions $d_1\times d_2$ and $d_2 \times d_3$, using the CRT, we reduce each matrix-matrix product modulo~$Q'$ to 
\begin{itemize}
    \item[$\bullet$] one matrix-matrix product of the same dimensions modulo $p_i$, for all~$0\le i < k$; 
    \item[$\bullet$] $d_1 d_3$ calls to the CRT with moduli~$p_0, \ldots, p_{k-1}$. 
\end{itemize}
The second step can be performed in~$O(k)$ arithmetic operations modulo $Q'$ per matrix element once we have precomputed the elements $((Q'/p_i)^{-1} \bmod p_i) (Q'/p_i) \bmod Q'$ (for~$0 \leq i <k$).
We then compute the result of $\cpmm$ or $\ccmm$ modulo $Q'$. If needed, we can \textsf{ModSwitch} back to the initial modulus~$Q$.

A variant of this approach\footnote{replacing \textsf{ModSwitch} by lifting the product modulo $q$ to a product over the integers, and computing it modulo $q' = q_0 \dots q_{k-1}$ with $q' > d_2 q^2$} is implemented in the Flint library~\cite{flint}.

\subsubsection{Comparison between the strategies.}
When applicable (mostly for the $\BB$-part in combination with Strategy~1, when the precision needs are moderate), Strategy~2~is preferable.
Between the other two options, Strategy~3 reduces to only $\max(k_1, k_2)$ calls to $\fpPPMM$, and may thus seem preferable compared to a reduction to $k_1 k_2$ $\fpPPMM$. 
Yet, \emph{in the case of our $\cpmm$ experiments}, the matrix $\BU$ is a low-precision plaintext matrix, which we encode with a small scaling factor~$\Delta$. As long as $\Delta^2 d_2 < 2^{53}$, we can take~$k_2 = 1$. The number of calls to $\fpPPMM$ is then the same as in Strategy~1 and Strategy~2. In that case, Strategy~1 is preferable for two reasons: the associated pre- and post-processing are lighter (chopping and adding, compared to $\textsf{ModSwitch}$, modular reductions and CRT).

\subsection{On CKKS parameter selection}
Our algorithms work for an arbitrary ciphertext modulus. However, our implementation performs the linear algebra at a low modulus. There are three main reasons for this choice.
\begin{itemize}
\item[$\bullet$] Our algorithm uses the coefficient encoding of CKKS. In the CKKS computation chain, in the case of the most common (S2C-first) bootstrapping variant, this format is only used at the lowest and highest computational levels.
\item[$\bullet$] The cost of computation at a given level is roughly proportional to the corresponding level number. This means that one should try to organize the computations so that heavy computations are performed at the lowest possible level. This gives a strong incentive to perform large-dimensional homomorphic algebra at the lowest possible modulus. 
\item[$\bullet$] Our reductions are more efficient when the matrix dimensions are at least as large as the RLWE ring-degree~$N$. When the ciphertext modulus is small, 128-bit security can be achieved for a smaller ring degree, down to $2^{13}$ or sometimes even $2^{12}$ depending on the amount of  auxiliary modulus required. This improves the range of applicability of our most efficient reductions. 
\end{itemize}

We have chosen ring degrees that are sufficiently large to allow for the auxiliary modulus required for the pre- and post-processings between the various encryption formats. In practice, this means $\log N = 13$ for Algorithm~\ref{alg:general} and $\log N = 12$ for Algorithm~\ref{alg:pcmm}. For Algorithm~\ref{alg:ccmm}, we use $N = d \in \{2^{12}, 2^{13}, 2^{14}\}$, as in the description of the algorithm. 
In all cases, we chose the structured-$\BS$ shared-$a$ encryption format as soon as the relevant matrix dimension is larger than the ring degree, and the plain $\RLWE$ encryption format otherwise.

\subsection{Format conversion experiments }\label{sse:exp-fmt-conv}
We present experimental timings for some of the format conversion algorithms described in Section~\ref{sec:conversions}. 

\subsubsection{Shared-$a \leftrightarrow$ shared-$s$.} 
We have implemented Algorithm~\ref{alg:fast_fmtswitch} and measured its performance. We have also implemented backward conversion from shared-$a$ to shared-$s$, which  consists in independent key-switchings for all the ciphertexts. 
Timings (in seconds) are presented in Table~\ref{tab:exp_fmt_conversion}, for varying  numbers of ciphertexts converted to or from shared-$a$ format. In these experiments, we have used $N = 2^{12}$, ciphertext modulus~$q = 2^{54}$ and auxiliary prime $P = 2^{54}$ for key-switching. 
The timings are consistent with the expected behavior: shared-$s$ to shared-$a$ shows a slightly superlinear cost as a function of~$n$, while shared-$a$ to shared-$s$ shows a linear cost as a function of~$n$. 

\begin{table}[h]
    \centering
    \begin{tabular}{|c|c|c|}
        \hline\hline
        $n$ & Algorithm~\ref{alg:fast_fmtswitch}~ & $~\text{sh.-}a \rightarrow \text{sh.-}s$ \\\hline
        $2$ & $1.11\textrm{e-}3$ &  $6.27\textrm{e-}4$ \\\hline
        $2^2$ & $2.83\textrm{e-}3$ &  $1.03\textrm{e-}3$ \\\hline
        $2^6$ & $9.1\textrm{e-}2$ &  $1.37\textrm{e-}2$ \\\hline
         $2^7$ & $0.212$  & $2.7\textrm{e-}2$ \\ \hline
        $2^{12}$ & $1.05\textrm{e}1$ & $0.757$ \\\hline
        $2^{13}$ & $2.45\textrm{e}1$ & $1.68$ \\\hline
        $2^{14}$ & $5.07\textrm{e}1$ & $3.13$\\\hline   
    \end{tabular}
    \caption{\label{tab:exp_fmt_conversion}Timings for  shared-$s$ to shared-$a$ conversion (left) and shared-$a$ to shared-$s$ conversion (right), in seconds (average over 10 runs).}
\end{table}

We have also used these timings to estimate the conversion cost from shared-$s$ representation to either structured-$\BS$ shared-$a$ or structured-$\BA$ shared-$a$ representation, and back. These estimates can be found in Table~\ref{tab:matfmtswitch}.

\begin{table}[h]
    \centering
    \begin{tabular}{|c|c|c||c|c|}
        \hline\hline
        $(d_1, d_2, d_3)$ & $\text{sh.-}s\rightarrow$ & 
       $\text{St.-}\BS\ \text{sh.-} a$ & 
        $\text{sh.-}s \rightarrow $ &
        $\text{St.-}\BA\ \text{sh.-} a$\\         
        &$\text{St.-}\BS\ \text{sh.-} a$  &
        $\rightarrow \text{sh.-}s$ &
        $\text{St.-}\BA\ \text{sh.-}s$ &
              $ \rightarrow \text{sh.-}s$ \\
              \hline
        $(2^{12}, 2^{12}, 1)$ & N/A & N/A  & $1.05\textrm{e}1$ & N/A\\ \hline
        $(2^{12}, 2^{12}, 2^6)$ & N/A & N/A & $1.05\textrm{e}1$ & $1.37\textrm{e-}2$  \\ \hline
        $(2^{12}, 2^{12}, 2^{12})$ & N/A & N/A & $1.05\textrm{e}1$ & $7.57\textrm{e-}1$                 
       \\ \hline
        $(2^{13}, 2^{13}, 1)$ & $9.09$ & $6.27\textrm{e-}4$  & $4.90\textrm{e1}$ & N/A\\ \hline
        $(2^{13}, 2^{13}, 2^6)$ & $9.09$ & $4.01\textrm{e-}2$ & $4.90\textrm{e1}$ & $2.74\textrm{e-}2$ \\ \hline
        $(2^{13}, 2^{13}, 2^{13})$ & $9.09$ & $5.14$ & $4.90\textrm{e1}$ & $3.36$   
        \\\hline
        $(2^{14}, 2^{14}, 1)$ & $4.64\textrm{e}1$ & $1.03\textrm{e-}3$  & $2.03\textrm{e}2$ & N/A\\ \hline
        $(2^{14}, 2^{14}, 2^7)$ & $4.64\textrm{e}1$ & $1.32\textrm{e-}1$ & $2.03\textrm{e}2$ & $1.08\textrm{e-}1$ \\ \hline
        $(2^{14}, 2^{14}, 2^{14})$ & $4.64\textrm{e}1$ & $1.69\textrm{e}1$ & $2.03\textrm{e}2$ & $1.25\textrm{e}1$ \\ \hline
    \end{tabular}
    \caption{\label{tab:matfmtswitch}Estimates for the pre- and post-processing costs for large-dimensional matrix formats (shared-$a$), in seconds. ``N/A''~indicates that the format coincides with Structured-$\BS$ shared-$s$, so that there is no processing cost.}
\end{table}

\subsubsection{Transpose.}
Table~\ref{tab:exp-transpose} gives timings for~Algorithm~\ref{alg:transpose}. As required by the algorithm, the ring dimension is equal to the matrix dimension. 
Despite the limited number of possible values for~$d=N$, we still observe the expected quasi-quadratic behaviour.

\begin{table}[H]
    \centering
    \begin{tabular}{|c|c|}
        \hline
        $N$ & Algorithm~\ref{alg:transpose}\\
        \hline
        $2^{12}$ & 5.60 \\
        \hline
        $2^{13}$ & 25.4 \\
        \hline 
        $2^{14}$ & 117 \\
        \hline
    \end{tabular}
    \caption{\label{tab:exp-transpose}Timings for the transpose of an encrypted matrix of dimension $d \times d$ with~$d = N$, using~Algorithm~\ref{alg:transpose}. Timings are in seconds  (average over 10 runs).}
\end{table}

\subsection{Matrix multiplication experiments}
We turn to our experiments for matrix/matrix and matrix/vector multiplication. 
We chose a base modulus~$q \approx 2^{54}$, with scaling factors of $\approx 2^{20}$, so that after rescaling our final result is obtained modulo~$q_0 \approx 2^{34}$, with a scaling factor $\approx 2^{20}$. 
The auxiliary modulus for the $\RGSW$ method is set to $p = 2^{40}$, while the auxiliary modulus for key-switching is set to $P = 2^{54}$.
Our input matrices are sampled with coefficients that are i.i.d.\ uniform in~$[-1, 1]$. Gadget rank (``{\tt dnum}") is set to 2.

\subsubsection{Experiments with Algorithm~\ref{alg:pcmm}.}
We have implemented Algorithm~\ref{alg:pcmm} using Strategy~1 for the $\BA$-part. As~$\|\BM\|_\infty \approx q/2 \le 2^{53}$ and $\|\BU\|_\infty \le 2^{19}$, we have $k_1 = 3$ and $k_2 = 1$ with the notations of Lemma~\ref{le:strategy1}. Handling the $\BA$-part hence consumes $3$ calls to $\fpPPMM$. We use Strategy~2 for the $\BB$-part, handling it with a single $\fpPPMM$. We only report the timings for $\ModPPMM$, as \textsf{Rescale} uses a negligible amount of time. 

The timings are reported in Table~\ref{tab:pcmm}. The factor 3 between the computation times for the $\BA$-part and the $\BB$-part is clearly visible in the square case $d_1 = d_2 = d_3$. For smaller $d_3$, by comparing the timings of Table~\ref{tab:pcmm} with the cost of one $\fpPPMM$ in the same dimension, we observe that most of the computation time is caused by the process of turning $\ModPPMM$ into $\fpPPMM$ rather than in linear algebra itself. In our opinion, this is explained by the following facts: 
\begin{itemize}
    \item[$\bullet$] for small $d_3$, this processing has the same asymptotic complexity as the $\fpPPMM$ itself; as $\fpPPMM$ only performs very simple arithmetic operations, each component (processing, linear algebra) takes a significant constant proportion of the timings; 
    \item[$\bullet$] we are comparing decades-polished and fine-tuned code (BLAS) to fresh and prototype code (our reduction). Much more work would be required to turn the latter in a code of similar quality as the former. 
\end{itemize}

\begin{table}[h]
    \centering
    \begin{tabular}{|c|c|c|c|c|}
        \hline\hline
        $d_1 \times d_2 \times d_3$ & \multicolumn{3}{c|}{$\ModPPMM$ time} 
        & Precision
        \\\cline{2-4}
        & $a$-part & $b$-part & total & (bits)
        \\ 
        & & (Strategy~1) & (Strategy~2) & 
        \\\hline
$ 2^{12} \times 2^{12} \times 1 $ &  $ 0.178 $ & $ 0.249 $ & $ 0.427 $ & $ 13.5 $\\ \hline
$ 2^{12} \times 2^{12} \times 2^{6} $ &  $ 0.253 $ & $ 0.270 $ & $ 0.523 $ & $ 13.5 $\\ \hline
$ 2^{12} \times 2^{12} \times 2^{12} $ &  $ 5.00 $ & $ 1.78 $ & $ 6.78 $ & $ 13.4 $\\ \hline
$ 2^{13} \times 2^{13} \times 1 $ & $ 0.402 $ & $ 1.20 $ & $ 1.60 $ & $ 13.8 $\\ \hline
$ 2^{13} \times 2^{13} \times 2^{7} $ & $ 0.699 $ & $ 1.42 $ & $ 2.12 $ & $ 13.6 $\\ \hline
$ 2^{13} \times 2^{13} \times 2^{13} $ & $ 19.1 $ & $ 13.3 $ & $ 32.4 $ & $ 13.7 $\\ \hline
$ 2^{14} \times 2^{14} \times 1 $ & $ 0.914 $ & $ 6.19 $ & $ 7.10 $ & $ 14.0 $\\ \hline
$ 2^{14} \times 2^{14} \times 2^{7} $ & $ 1.45 $ & $ 6.74 $ & $ 8.19 $ & $ 13.5 $\\ \hline
$ 2^{14} \times 2^{14} \times 2^{14} $ & $ 73.7 $ & $ 102 $ & $ 176 $ & $ 13.5 $\\ \hline
    \end{tabular}

    \caption{Timings, in seconds, for $\cpmm$ using~Algorithm~\ref{alg:pcmm}, averaged over 10 runs. The computation of the $\BA$-part is reduced to 3 $\fpPPMM$ using Strategy~1, while the computation of the $\BB$ part is reduced to 1 $\fpPPMM$ using Strategy~2. The first three rows use the $\RLWE$ format, the last six the structured-$\BS$ shared-$a$ format. We do not estimate the pre- and post-processing times, for which we refer to Table~\ref{tab:matfmtswitch}. We give the worst precision over all coefficients and the 10 runs.   
    \label{tab:pcmm} }
\end{table}

These experiments also illustrate the impact of the shared-$a$ technique. As we chose $N = 2^{12}$, timings for $(d_1, d_2, d_3)$ with $d_1, d_2 \ge 2^{13}$ all use the structured-$\BS$ shared-$a$ format. We can spot the impact of this format by noticing that if we were using the shared-$s$ format, 
the timings for the $b$-part would remain the same,
but the timings for the $a$-part would be three times the timings for the $b$-part. 
For instance, for $d_1 = d_2 = d_3 = 2^{14}$ we would expect shared-$s$ to take $\approx 400$s instead of $176$s. Even adding the conversion costs ($\approx 60$s, see Table~\ref{tab:matfmtswitch}), using structured-$\BS$ shared-$a$ brings a very significant improvement. 

\subsubsection{Experiments with Algorithm~\ref{alg:ccmm}.}
In Table~\ref{tab:ccmm-jaihyun}, we present  timings for the $\ccmm$ algorithm (Algorithm~\ref{alg:ccmm}), with $d_1 = d_2 = d_3$ equal to the ring degree~$N$, using the \RLWE\ format. 
We used Strategy~3 based on $\textsf{ModSwitch}$, from $q \approx 2^{54}$ to a product $q'_0 q'_1 q'_2$ with $q'_i \le 2^{19}$.   

We observe that, despite the fact that Transpose is asymptotically negligible, it still consumes a significant proportion of the computation time for the values of~$N$ that we consider. Post-processing (Relin/Rescale) accounts for around~$5\%$ of the total computation time. Finally, the $\ppmm$ column shows timings that are close to 12 calls to $\fpPPMM$ in the same dimensions, as expected.

\begin{table}[h]
    \centering
    \begin{tabular}{|c|c|c|c|c|c|c|c|c|}
        \hline\hline
        $d_1 \times d_2 \times d_3$  & \multicolumn{4}{c|}{Latency of each step (s)} & Total &  Precision
        \\\cline{2-5}
        &  Transpose & $\ModPPMM$ & Relin. & Rescale &  & (bits)
        \\\hline
        
        $2^{12} \times 2^{12} \times 2^{12}$ &   $16.8$ & $22.1$ & $1.17$ & $1.07$ & $41.2$ & $9.1$ \\\hline
        $2^{13} \times 2^{13} \times 2^{13}$ &  $73.6$ & $162$ & $5.17$ & $4.81$ & $245$ & $8.8$\\\hline
        $2^{14} \times 2^{14} \times 2^{14}$ &  $352$ & $1300$ & $24.0$ & $21.0$ & $1710$ & $8.3$\\\hline
    \end{tabular}
    \caption{Timings for~Algorithm~\ref{alg:ccmm}, averaged over 10 runs. All three experiments use the $\RLWE$ format. The precision displayed is the worst precision over all coefficients over the runs.
    \label{tab:ccmm-jaihyun}}
\end{table}

\subsubsection{Experiments with Algorithm~\ref{alg:general}.}\label{sse:exp-ccmm}
Finally, in Table~\ref{tab:ccmm-yj-symmetric-degree}, we present the timings for the RGSW-based $\ccmm$ algorithm (Algorithm~\ref{alg:general}), with $d_1 = d_2$. 
We have used two different parameters for the ring degree: $N = 2^{12}$ for both input matrices and $N = 2^{13}$ for both input matrices.
Compared to our other choices, the first choice does not allow  pre- or post-processing to and from the structured-$\BS$ shared-$a$ format, with a decent precision and/or efficiency. However, as the \RGSW-based method is likely to be more useful when  precomputations are allowed, we include those timings which are legitimate if switching to shared-$a$ representation is done as part of the precomputation. Another choice of ring degrees for this method would be $N^\BM = 2^{12}$ for the matrix in $\RGSW$ format and $N^{\BU} = 2^{11}$ for the matrix in $\RLWE$ format. However, this would require precomputation (rather than pre-processing) on both $\BM$ and $\BU$, which seems very limited in terms of applications. 

In the implementation, we use Strategy 3 based on $\modswitch$ for $\ModPPMM$. As $pq \approx 2^{94}$, we used  $q'=q'_0 q'_1 q'_2 q'_3 q'_4$ with $q'_j \approx 2^{19}$ for all~$j$. We thus reduce each $\ModPPMM$ modulo~$pq$ to 5 $\fpPPMM$'s. We do not decompose the timings in several steps as our experiments show that over $95\%$ of the time is spent in the $\ModPPMM$. As for Table~\ref{tab:pcmm}, however, we can see that, for small~$d_3$, the cost does not grow linearly with~$d_3$, as one would expect. Again, our finer measurements suggest that for small~$d_3$, even elementary processing steps (e.g., converting an $\texttt{int64}$ matrix in a $\texttt{float64}$ matrix) take a significant portion of the total time. This suggests that a more careful implementation would be required to minimize these overheads to have a better assessment of the full potential of our reductions in the small $d_3$ case. 

The last row of~Table~\ref{tab:ccmm-yj-symmetric-degree} illustrates the impact of the structured-$\BS$ shared-$a$ format. We expect and observe a factor~$\approx 5/3$ (corresponding to our use of \textsf{ModSwitch} for $\ModPPMM$) between Tables~\ref{tab:ccmm-jaihyun} and~\ref{tab:ccmm-yj-symmetric-degree} for rows using the $\RLWE$ format, notably in the square case where overheads can be neglected. For  $(d_1, d_2, d_3) = (2^{14}, 2^{14}, 2^{14})$, however, our implementation of Algorithm~\ref{alg:ccmm} makes 12 calls to BLAS  in dimensions $2^{14} \times 2^{14}\times  2^{14}$ while, using the structured-$\BS$ shared-$a$ format, our implementation of Algorithm~\ref{alg:general} makes 
\begin{itemize}
    \item[$\bullet$] 5 calls to BLAS in dimensions $2^{14} \times 2^{14}\times  2^{14}$; 
    \item[$\bullet$] 5 calls to BLAS in dimensions $2^{13} \times 2^{14}\times  2^{14}$ and 5 calls in dimensions $2^{14}\times 2^{13}\times 2^{14}$; 
    \item[$\bullet$] 5 calls to BLAS in dimensions $2^{13} \times 2^{13}\times  2^{14}$. 
\end{itemize}
Assuming that the cost grows linearly with the product of the dimensions, we thus expect a $\ModPPMM$ time for Algorithm~\ref{alg:general} of the order of $\approx 11.25$ calls to BLAS
in dimensions $2^{14} \times 2^{14}\times  2^{14}$, so slightly cheaper than for Algorithm~\ref{alg:ccmm}. This is in line with what we observe in Tables~\ref{tab:ccmm-jaihyun} and~\ref{tab:ccmm-yj-symmetric-degree}.

\begin{table}[h]
    \centering
    \begin{tabular}{|c|c|c|c|c|}
        \hline
        $d_1 \times d_2 \times d_3$ & 
        $N^U(=N^M)$&
        $\ModPPMM$ time & Precision
        \\\hline    
        $2^{12}\times 2^{12}\times 1$ & $2^{12}$ & $0.882$ & $17.5$\\\hline
        $2^{12}\times 2^{12}\times 2^6$ & $2^{12}$ & $1.44$ & $17.4$\\ \hline
        $2^{12}\times 2^{12}\times 2^{12}$ & $2^{12}$ & $42.6$ & $17.4$\\ \hline
        $2^{13}\times 2^{13}\times 1$ & $2^{13}$ & $4.35$ & $17.2$\\\hline
        $2^{13}\times 2^{13}\times 2^6$ & $2^{13}$ & $6.47$ & $17.3$\\ \hline
        $2^{13}\times 2^{13}\times 2^{13}$ & $2^{13}$ & $291$ & $17.3$\\ \hline
        $2^{14} \times 2^{14} \times 1$ & $2^{13}$ & $10.5$ & $17.4$ \\\hline
        $2^{14} \times 2^{14} \times 2^7$ & $2^{13}$ & $19.0$ & $17.4$ \\\hline
        $2^{14} \times 2^{14} \times 2^{14}$ &  $2^{13}$ & $1270$ & $17.4$ \\\hline
    \end{tabular}
    \caption{\label{tab:ccmm-yj-symmetric-degree}Timings for the $\ModPPMM$ part of $\ccmv$ and $\ccmm$ using the \RGSW\ approach (Algorithm~\ref{alg:general}), in seconds,  averaged over 10 runs. The precision displayed is the worst over all coefficients and all the 10 runs. The six first rows use the $\RLWE$ encryption format, while the last three rows use the structured-$\BS$ shared-$a$ format.}
\end{table}

\subsubsection{Precision aspects.}
These experiments exhibit somewhat diverse precision results for the three methods. This is related to choices that we have made in order to provide a reliable comparison in terms of efficiency of the approaches, which impact on the precision. 
\begin{itemize}
    \item[$\bullet$] Using the floating-point approach for the $b$-part in Algorithm~\ref{alg:pcmm} reduces the overall precision, as analyzed in Section~\ref{sec:strategy3}.
    \item[$\bullet$] Taking the same moduli across the methods (ciphertext modulus~$q$ and auxiliary modulus~$P$) with a somewhat tight budget of auxiliary modulus for key switching favours Algorithms~\ref{alg:pcmm} and~\ref{alg:general}, which involve fewer key-switchings. Within the same total modulus budget, different choices would yield a better precision for Algorithm~\ref{alg:ccmm}.
    \item[$\bullet$] Finally, Algorithm~\ref{alg:general} is inherently slightly more precise, all other things being equal, as part of the computations are handled at modulus~$pq$, which increases the value of the underlying scaling factor. This impacts some, though not all, of the error terms. 
\end{itemize}

\bibliographystyle{splncs04}
\bibliography{bibliography}

\end{document}